\documentclass[sigconf]{acmart}
\usepackage{amsmath,amsfonts}
\usepackage{algorithmic}
\usepackage{graphicx}
\usepackage{textcomp}
\usepackage{xcolor}
\usepackage{booktabs}
\usepackage{pgfplots}
\usepackage{subcaption}
\usepackage{colortbl}
\usepackage{bbding}
\usepackage{braket}
\usepackage{epsfig,latexsym,graphics}
\usepackage[font={small}]{caption}
\usepackage{setspace}
\usepackage{enumitem}
\usepackage{framed}
\usetikzlibrary[patterns]
\usepackage{cmd}
\usepackage{multirow}
\usepackage{xspace}
\colorlet{shadecolor}{gray!30}

\newcommand{\ie}{\emph{i.e.,}\xspace}
\newcommand{\eg}{\emph{e.g.,}\xspace}

\AtBeginDocument{%
  }

\setcopyright{acmlicensed}
\copyrightyear{2027}
\acmYear{2027}
\acmDOI{XXXXXXX.XXXXXXX}
\acmConference[SIGMOD '27]{the 2027 ACM SIGMOD/PODS International
  Conference on Management of Data}{June 13--19,
  2027}{Huntington Beach, CA, USA}
\acmISBN{978-1-4503-XXXX-X/2027/06}

\begin{document}

%%
%% The "title" command has an optional parameter,
%% allowing the author to define a "short title" to be used in page headers.
%\title{Efficient Proximity Graph Merging via Structure-Aware Neighbor Reuse}
%\title{Efficient Proximity Graph Merging via Reverse Neighbor Sliding}
\title{Reverse Neighbor Sliding and Order Selection for Efficient Multi-Proximity Graph Merging}

%%
%% The "author" command and its associated commands are used to define
%% the authors and their affiliations.
%% Of note is the shared affiliation of the first two authors, and the
%% "authornote" and "authornotemark" commands
%% used to denote shared contribution to the research.
\author{Liuchang Jing}
\affiliation{%
  \institution{HKUST (GZ)}
  \city{Guangzhou}
  \state{Guangdong}
  \country{China}
}
\email{ljing248@connect.hkust-gz.edu.cn}

\author{Mingyu Yang}
\affiliation{%
  \institution{HKUST (GZ) \& HKUST}
  \city{Guangzhou}
  \state{Guangdong}
  \country{China}
}
\email{myang250@connect.hkust-gz.edu.cn}

\author{Lei Li}
\affiliation{%
  \institution{HKUST (GZ) \& HKUST}
  \city{Guangzhou}
  \state{Guangdong}
  \country{China}
}
\email{thorli@ust.hk}

\author{Jianbin Qin}
\affiliation{%
  \institution{Shenzhen University}
  \city{Shenzhen}
  \state{Guangdong}
  \country{China}
}
\email{qinjianbin@szu.edu.cn}

\author{Wei Wang}
\affiliation{%
  \institution{HKUST(GZ) \& HKUST}
  \city{Guangzhou}
  \state{Guangdong}
  \country{China}
}
\email{weiwcs@ust.hk}

%%
%% By default, the full list of authors will be used in the page
%% headers. Often, this list is too long, and will overlap
%% other information printed in the page headers. This command allows
%% the author to define a more concise list
%% of authors' names for this purpose.
\renewcommand{\shortauthors}{Liuchang Jing, Mingyu Yang, Lei Li, Jianbin Qin, and Wei Wang}

%%
%% The abstract is a short summary of the work to be presented in the
%% article.
\begin{abstract}
Approximate $k$ Nearest Neighbor (AKNN) search in high-dimensional space is a foundational problem in vector databases with widespread applications. Among the numerous AKNN indexes, Proximity Graph-based indexes achieve state-of-the-art search efficiency across various benchmarks. In many real-world scenarios, datasets are maintained as multiple segment-level graph indexes to support continuous writes and segment management. However, these fragmented indexes complicate maintenance and degrade search efficiency, making fast graph index merging essential.
In this paper, we focus on the efficient merging of multiple existing graph indexes into a single one. To achieve this, we propose a Reverse Neighbor Sliding Merge (RNSM) that exploits structural information to boost merging efficiency. We further propose Merge Order Selection (MOS) to minimize total merge cost across multiple indexes by eliminating redundant operations. Experiments show that our approach yields up to a 3.86$\times$ speedup over existing index merge methods and a 9.92$\times$ speedup over index reconstruction, while maintaining comparable search performance. Moreover, our method scales to merging up to 50 sub-indexes on datasets of 100 million vectors, maintaining consistent speedups.
\end{abstract}

%%
%% The code below is generated by the tool at http://dl.acm.org/ccs.cfm.
%% Please copy and paste the code instead of the example below.
%%


%%
%% Keywords. The author(s) should pick words that accurately describe
%% the work being presented. Separate the keywords with commas.
%% A "teaser" image appears between the author and affiliation
%% information and the body of the document, and typically spans the
%% page.

\received{17 April 2026}
\received[revised]{19 August 2026}
\received[accepted]{12 September 2026}

%%
%% This command processes the author and affiliation and title
%% information and builds the first part of the formatted document.
\maketitle

\section{Introduction}\label{sec:intro}

The $k$ nearest neighbor (KNN) search in high-dimensional space is a fundamental problem that has received widespread attention due to its extensive applications in data mining~\cite{patternclassification}, information retrieval~\cite{ImageRetrieval,ObjectRetrieval}, recommendation systems~\cite{GoogleRecommendation_2007,CollaborativeFilteringRecommendation_2007}, and the Retrieval-Augmented Generation (RAG) of Large Language Models (LLMs)~\cite{RAG}. However, the cost of exact KNN search is extremely high for large-scale datasets~\cite{ApproximateNeighbors-indyk-STOC-1998}. Therefore, extensive research has focused on studying approximate $k$ nearest neighbor (AKNN) search ~\cite{LSH-datar-SODA-2004,DFLSH-shan-ENG-2025,DBLSH-tian-IEEE-2024,DETLSH-wei-VLDB-2024,LSBTree-tao-ACM-2010,DynamicCollisionCounting-gan-SIGMOD-2012,SimilaritySearchHashing-gionis-VLDB-1999,SRS-sun-VLDB-2014,ProductQuantization-jegou-IEEE-2011,OptimizedProductQuantization-ge-PAMI-2014,InvertedMultiIndex-babenko-IEEE-2015,RaBitQ-gao-SIGMOD-2024,ExtendRaBitQ-gao-sigmod-2025,PQFastScan-andre-VLDB-2015,EffectiveGeneralDistance-yang-ICDE-2024,LVQ-aguerrebere-VLDB-2023,ADSampling-gao-ACM-2023,MRQ-yang-arxiv-2024,CoverTrees-beygelzimer-ICML-2006,sa-tree-navarro-vldb-2002,RandomProjectionTrees-dasgupta-ACM-2008,Hercules-echihabi-vldb-2022,FixedDimensions-arya-SODA-1993, FANNG-harwood-CVPR-2016,HNSW-malkov-IEEE-2020, NSW-malkov-IS-2014, NSG-fu-VLDB-2019, SSG-fu-IEEE-2022, Efficient-peng-SIGMOD-2023,DiskANN-subramanya-NIPS-2019,FINGER-chen-ACM-2023, SymphonyQG-gou-sigmod-2025,Hemi-Sphere-Graph-SIGMOD-2026-Jing-Tang,Fast-Conver-PG-SIGMOD-2026-shangqi}, trading accuracy for search efficiency. 
Among various AKNN methods, graph-based indexes are the state-of-the-art solution, demonstrating comparable performance compared to other approaches~\cite{GraphBasedANNS-wang-vldb-2021,HighDimensionalDataExperiments-li-IEEE-2016,ParlayANN-dobson-PPoPP-2023,ANNBenchmarks-aumuller-IS-2020}. Graph-based methods represent each vector in the dataset as a node in a graph, where edges connecting nodes reflect proximity relationships designed to facilitate efficient navigation during search. Most graph-based methods construct a Proximity Graph (PG) through one of two approaches: (1) refining a random or approximate $k$ nearest neighbor graph (KNNG)~\cite{NSG-fu-VLDB-2019,SSG-fu-IEEE-2022,Efficient-peng-SIGMOD-2023,RevisitingIndexConstruction-yang-vldb-2025}, or (2) incrementally inserting nodes into an empty graph~\cite{HNSW-malkov-IEEE-2020,NSW-malkov-IS-2014}. 

\begin{figure}[!t]
    \centering
    \includegraphics[width=\columnwidth]{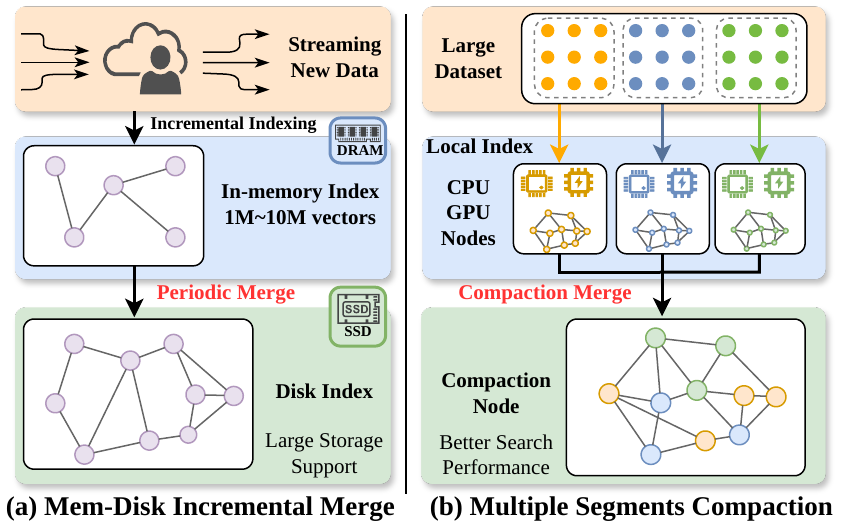}
    \vspace*{-6mm}
        \caption{Examples of Multiple Indexes Merging Scenarios}
        \vspace*{-3mm}
  \label{fig:scenario}
\end{figure}

However, both indexing methods incur substantial time costs for graph construction~\cite{GraphBasedANNS-wang-vldb-2021,ANNBenchmarks-aumuller-IS-2020} due to extensive distance computations. Therefore, commercial vector databases like Milvus~\cite{Milvus-wang-SIGMOD-2021} and BlendHouse~\cite{BlendHouse-ICDE-2025} typically partition data \emph{by design} to support dynamic and independent segment data management. Despite the scalability of this approach, directly searching multiple sub-indexes incurs a significant search performance degradation. This phenomenon is intuitive because searching $m$ separate graph indexes independently requires traversing each from its entry point. Given that graph search complexity is logarithmic with respect to dataset size~\cite{NSG-fu-VLDB-2019,Efficient-peng-SIGMOD-2023,Fast-Conver-PG-SIGMOD-2026-shangqi}, querying $m$ partitions scales the computational cost up to a factor of approximately $m$ in the worst case. As shown in the experimental results in Fig~\ref{fig:merge-necessity}, the search efficiency of the partitioned index with 10 partitions is only half that of the complete index, and this efficiency gap widens to nearly $1/6$ with 50 partitions. Therefore, merging multiple indexes is a critical operation in various scenarios.

Fig~\ref{fig:scenario} shows two index-merging scenarios. 
In \textit{incremental index merging} of Fig~\ref{fig:scenario}(a), streaming data are indexed in a small index and periodically merged into the larger index of prior data. To support continuous writes within memory limits, commercial systems keep the large index on SSD and new segments (typically fewer than 10M vectors) in memory under an LSM architecture~\cite{BlendHouse-ICDE-2025}. Existing methods treat this merge as costly incremental indexing~\cite{FreshDiskANN-singh-arxiv-2021,VSAG-zhong-vldb-2025}, which is computationally expensive. In \textit{multiple segment compaction} of Fig~\ref{fig:scenario}(b), CPU or GPU nodes independently build and search local segment indexes~\cite{DistributedStorage-yu-arxiv-2025,BinaryDistributedGraph-zhao-CIKM-2019}. Periodic compaction reduces storage and management overhead and improves search, but existing systems rebuild the full index from scratch~\cite{SOGAIC-shi-www-2025,BinaryDistributedGraph-zhao-CIKM-2019}, which introduces high computational cost. In both scenarios, speeding up index merging is crucial for efficiency.

% To address this performance degradation, index merge methods ~\cite{FGIM_wu_2026_SIGMOD, SimJoin-xie-POM-2025} have been proposed recently to enable navigation across partitions by connecting separate graphs.

\begin{figure}[!t]
\centering
\begin{small}
% 全局图例（两方法）
\begin{tikzpicture}
\begin{customlegend}[legend columns=4,
legend entries={Separated, Merged},
legend style={at={(0.5,1.15)},anchor=north,draw=none,font=\scriptsize,column sep=0.3cm}]
% 替换颜色
\addlegendimage{line width=0.2mm,color=violate,mark=o,mark size=0.5mm}
\addlegendimage{line width=0.2mm,color=amaranth,mark=triangle,mark size=0.5mm}
\end{customlegend}
\end{tikzpicture}
\\[-\lineskip]

\subfloat[10 Partitions of DEEP1M]{\vspace{-2mm}
\begin{tikzpicture}[scale=1]
\begin{axis}[
    height=3.0cm,
    width=\columnwidth/2.10,
    xlabel=recall@10,
    ylabel=Query Per Second, ylabel style={yshift=-15pt}, 
    label style={font=\scriptsize},
    tick label style={font=\scriptsize},
    ymajorgrids=true,
    xmajorgrids=true,
    grid style=dashed,
]
% Partition fanout (recall > 0.8)
\addplot[line width=0.2mm,color=violate,mark=o,mark size=0.5mm]
plot coordinates {
(0.8139, 1486.74)
(0.8500, 1316.62)
(0.9016, 1074.05)
(0.9254, 915.717)
(0.9346, 852.077)
(0.9765, 528.032)
(0.9904, 391.223)
(0.9953, 311.092)
(0.9966, 261.897)
(0.9974, 227.169)
(0.9980, 200.399)
(0.9984, 179.933)
(0.9985, 164.142)
(0.9988, 150.621)
(0.9989, 139.432)
(0.9991, 129.902)
(0.9992, 120.988)
(0.9993, 113.940)
};
% NGM merged (recall > 0.8)
\addplot[line width=0.2mm,color=amaranth,mark=triangle,mark size=0.5mm]
plot coordinates {
(0.8549, 3283.58)
(0.9075, 2522.23)
(0.9384, 2052.15)
(0.9556, 1771.44)
(0.9673, 1548.52)
(0.9751, 1396.58)
(0.9795, 1265.29)
(0.9837, 1165.08)
(0.9868, 1072.08)
(0.9887, 1004.54)
(0.9907, 940.720)
(0.9926, 881.754)
(0.9936, 833.817)
(0.9945, 789.341)
(0.9950, 750.238)
(0.9952, 716.931)
(0.9958, 685.331)
(0.9962, 657.385)
(0.9966, 630.696)
};
\end{axis}
\end{tikzpicture}}\hspace{2mm}
\subfloat[50 Partitions of DEEP1M]{\vspace{-2mm}
\begin{tikzpicture}[scale=1]
\begin{axis}[
    height=3.0cm,
    width=\columnwidth/2.10,
    xlabel=recall@10,
    ylabel=Query Per Second, ylabel style={yshift=-15pt}, 
    label style={font=\scriptsize},
    tick label style={font=\scriptsize},
    ymajorgrids=true,
    xmajorgrids=true,
    grid style=dashed,
]
% Partition fanout (recall > 0.8)
\addplot[line width=0.2mm,color=violate,mark=o,mark size=0.5mm]
plot coordinates {
(0.8548, 393.772)
(0.9002, 334.530)
(0.9282, 296.503)
(0.9577, 243.162)
(0.9725, 207.184)
(0.9780, 192.893)
(0.9933, 119.838)
(0.9974, 89.1972)
(0.9983, 71.3246)
(0.9990, 60.2576)
(0.9993, 52.4078)
(0.9994, 46.2442)
(0.9996, 41.6359)
(0.9997, 38.0118)
(0.9999, 34.9617)
(0.9999, 32.4297)
(1.0000, 30.2290)
(1.0000, 28.1731)
(1.0000, 26.5546)
};
% NGM merged (recall > 0.8)
\addplot[line width=0.2mm,color=amaranth,mark=triangle,mark size=0.5mm]
plot coordinates {
(0.8520, 3242.43)
(0.9095, 2425.70)
(0.9376, 1935.68)
(0.9547, 1672.07)
(0.9660, 1464.48)
(0.9754, 1313.09)
(0.9794, 1190.75)
(0.9835, 1095.17)
(0.9872, 1008.28)
(0.9897, 939.225)
(0.9914, 882.329)
(0.9924, 821.974)
(0.9936, 778.620)
(0.9943, 731.493)
(0.9948, 699.189)
(0.9956, 668.576)
(0.9960, 637.629)
(0.9961, 610.559)
(0.9963, 586.284)
};
\end{axis}
\end{tikzpicture}}

\vgap
\caption{Separated and Merged index single-thread search performance. Sub-indexes are constructed on 10 and 50 equal-sized partitions of the DEEP1M dataset. Separated indicates searching sequentially on all partitions. Merged indicates searching a single index merged by our proposed method.}
% 说明是单线程的
\vgap
\label{fig:merge-necessity}
\end{small}
\end{figure}
\stitle{Previous Methods.}\label{sec:previous}
We categorize existing index merge methods into two families based on their strategies for connecting multiple sub-graph indexes: (1) \textbf{Overlap}-based~\cite{DiskANN-subramanya-NIPS-2019, ExtendedNeighborhoodGraph-zhang-SIGIR-2025, CSPG-yang-NIPS-2024} methods intentionally duplicate data points across multiple segments during sub-index construction. The sub-indexes can be merged by concatenating the neighbor lists of the overlapping nodes. (2) \textbf{Search}-based~\cite{HNSW-malkov-IEEE-2020} treats the data in one index as queries to search in another index. The neighbors of the subgraph index node are updated with the search results. However, for already-built indexes, overlap-based methods require building indexes with duplicated vectors. Their merge cost also scales with the overlapping ratio, which is at least 2~\cite{DiskANN-subramanya-NIPS-2019}. Therefore, we focus on search-based methods that can also handle overlapping partitions: applying the neighbor pruning rule of the PG index to the concatenation of search results and existing neighbors from different indexes.
% Recent works~\cite{SimJoin-xie-POM-2025,FGIM_wu_2026_SIGMOD,MergeOfKNNG-zhao-IEEE-2022,hnswmerger-jin-sigmod-2026} focus on accelerating search-based methods, as they are distribution-agnostic and directly benefit from ANNS query optimizations. 
Recent search-based methods~\cite{SimJoin-xie-POM-2025,FGIM_wu_2026_SIGMOD,MergeOfKNNG-zhao-IEEE-2022,hnswmerger-jin-sigmod-2026} focus on reducing the time cost of merging multiple existing graph indexes in main memory by reducing the Number of Distance Computations (NDC). This is a practical setting with growing DRAM capacity~\cite{GraphBasedANNS-wang-vldb-2021,Flash-wang-SIGMOD-2025} and distance computation dominating the cost.

% The overlap-based method is integrated into disk solutions such as DiskANN~\cite{DiskANN-subramanya-NIPS-2019,SPANN-NIPS-2021-MicroSoft} and Milvus~\cite{Milvus-wang-ACM-2021}, where each node is forced to be placed into multiple partitions, and cross-partition connections are achieved through edge composites of that node. Systems built on LSM-tree–like organizations (\eg BlendHouse~\cite{BlendHouse-ICDE-2025} and VSAG~\cite{VSAG-zhong-vldb-2025}) continuously compact multiple small segments into larger ones in the background. During compaction, new per-segment graph indexes are built while old ones are discarded. However, the cost of merging small graphs approaches that of a full reconstruction. Existing overlap- and search-based approaches mainly support efficient write based on the LSM tree, but do not substantially reduce the computation overhead of index construction. This motivates us to design new methods that explicitly accelerate the merging process of a set of existing segment indexes.

% [1.KNNG transformation definition in sentence; 2.edge pruning no need; 3. m^2 all not necessary/redundant; 4. scalability limitation]
\stitle{State-of-the-art.}
A recent search-based study, FGIM~\cite{FGIM_wu_2026_SIGMOD}, has emerged as the most efficient method for index merging. It consists of a 3-step pipeline: 1) determine which pair of partitions to merge, 2) accelerate the pairwise merge operation, and 3) refine the final merged index. For step 1, FGIM searches each vector in all other partitions for nearest neighbors, noted as \textbf{cross-querying}. 
In step 2, its \textbf{Minimum Candidate Set Strategy} assumes that the cross-querying \textbf{search effort $ef$} can be evenly allocated among $m-1$ other partitions for a good enough merge quality. Thus, each search is \textbf{$1/(m-1)$} times cheaper. For step 3, with AKNN results of the previous step, indexes are roughly merged into a KNNG, noted as \textbf{KNNG transformation}. To recover the missing $k$ nearest neighbors due to the cheaper cross-querying in step 2, the KNNG is further refined with NNDecent~\cite{NNDecent-dong-WWW-2011}. The refined KNNG is further pruned into a PG index. However, FGIM has issues in all three steps: (1) In cross-querying, the $m^2$ all-pairs merges contain redundant operations, since not all partition pairs need to be merged; (2) The assumption in Minimum Candidate Set Strategy limits its applicability to unbalanced partition sizes, where few sub-indexes could contribute significantly more neighbors than others therefore requiring larger $ef$; (3) KNNG refining takes $O(k^2\cdot n)$ time~\cite{FGIM_wu_2026_SIGMOD}, where $n$ is the data size and $k$ increases with $m$, limitating its scalability.

\stitle{Our Solution.}
In this paper, we propose a novel search-based multi-index merging framework through two key techniques: \emph{Reverse Neighbor Sliding Merge} (RNSM) and \emph{Merge Order Selection} (MOS). We introduce them with 3 similar steps. In step 1, we propose MOS to reduce the $m^2$ merge operations. MOS treats partitions as vertices and pairwise merges as edges, and constructs a sparse and proximity-aware merge order graph that prioritizes the merge of closer partition pairs while preserving connectivity between partitions. For step 2, RNSM merges two partitions efficiently by reusing search results. Only selected vectors perform standard searches, their nearby vectors reuse the corresponding results as effective entry points for lower-cost searches. Thus, RNSM accelerates pairwise merging without sacrificing search quality or relying on uniform search-effort allocation of FGIM. 
% Specifically, RNSM is built on three key components: \emph{Neighbor Expand}, \emph{Pivot Selection}, and \emph{Sliding Merge}. \emph{Neighbor Expand} first traverses the graph index from each vector itself to expand its neighbors. This step is efficient given that PG indexes are sparse approximations of KNNG~\cite{NSG-fu-VLDB-2019,SSG-fu-IEEE-2022}. Based on the expanded neighboring information, \emph{Pivot Selection} identifies vectors as pivots whose search results can be used to benefit more vectors. Each pivot performs a standard search in the target index. \emph{Sliding Merge} then performs parallel-friendly and efficient searching via reusing search results of pivots for their neighbors, called local sliding. 
For step 3, we execute the RNSM operations selected by MOS and directly update the merged PG, avoiding expensive KNNG refining and cross-querying with faster pairwise merge and fewer merge operations. Moreover, we further determine the direction of search-based merging by defining \textbf{source and target indexes}: a source index supplies the vectors that issue cross-index queries on the target index, \ie source indexes are merged into target indexes. We always merge smaller indexes into larger ones~\cite{Milvus-wang-ACM-2021} and merge closer partition pairs first~\cite{douze2026FaissLibrary} in step 3.

\stitle{Contribution.}
We summarize our main contributions as follows:

\sstitle{Problem Analysis and Motivation (\S~\ref{sec:problem-analysis}).} We explored the key factors in index merging and found that nearest neighbor connections are crucial, thus transforming the index merging problem into a search optimization problem. We also analyzed the patterns of multi-index merging, providing new insights into the merge order selection.

\sstitle{Efficient Two-Index Merge Method (\S~\ref{sec:Bi-Index-Merge}).} We propose RNSM, an efficient two-index merge method that leverages structural information from both the source and target indexes to accelerate merging. RNSM first identifies \emph{pivots}: source-index vectors that perform standard searches in the target index and whose results are reused by nearby \emph{followers} as entry points for local sliding.

\sstitle{Merge Order Selection (\S~\ref{sec:MOS-technical}).} We propose MOS, which formulates merge order selection as a graph optimization problem that minimizes merge cost with two constraints: preserving efficiency through degree control and search quality through diameter control. We introduce tailored strategies for both random and cluster-based partitions for different application scenarios.

\sstitle{Extensive Experiment Analysis (\S~\ref{sec:Exp}).} We conduct comprehensive experiments on real-world datasets for multi-index merging. 
We integrate our framework into popular graph-based indexes, including HNSW, NSG, SSG, and $\tau$-MNG, demonstrating its broad applicability. Our method delivers up to 9.92$\times$ speedup while preserving search quality and efficiency. Moreover, the speedup ratio remains stable as the number of partitions increases with 50 sub-indexes and 100 million data, indicating robust scalability.

Due to space constraints, some additional proofs and experiments are omitted, which can be found in the Appendix.  

\begin{table}[!t]% h asks to places the floating element [h]ere.
  \caption{A Summary of Notations}\vgap
  \vspace*{-2mm}
  \label{tab:notation}
  \resizebox{\columnwidth}{!}{%
  \begin{tabular}{r|l}
    \toprule
    \textbf{Notation}   &  \textbf{Description}\\
    \midrule
    $X_i$        &  A set of points/vectors \\
    $x, y$        &  A data point/vector \\   
    $\mathbb{R}^D$ & $D$-dimensional Euclidean space \\
    $\delta\left(x, y\right)$ & Euclidean/$L_2$ distance between x and y \\
    $G_i=\left(V_i, E_i\right)$  &  The graph index for $X_i$ with vertex set $V_i$ and edge set $E_i$ \\
    $\mathcal{L}(p, q)$   & The local sliding NDC cost from point $p$ to $q$ \\
    $N(v,G)$ & The neighbors of $v$ in $G$\\
    $N_k\left(x, X\right)$ & The $k$-nearest neighbors of $x$ in $X$ \\
    $R_k\left(x, X\right)$ & The reverse $k$-nearest neighbors of $x$ in $X$ \\
    % $R$ & The maximum out-neighbors of any node in the graph index \\
    \bottomrule
  \end{tabular}%\vgap\vspace{-1.5em}
  }
  \vspace*{-2mm}
\end{table}

\section{Preliminaries}\label{sec:preliminary}
%In this section, 
We first formally define the AKNN search problem in \S~\ref{sec:AKNN-definition} and then explain the graph-based AKNN index in \S~\ref{sec:pg-definition}.

\subsection{The AKNN Search Problem}\label{sec:AKNN-definition}
Table \ref{tab:notation} lists the frequently used notations and their definitions. %Before formally defining the AKNN search, 
We first describe the KNN search problem:
\begin{definition}[KNN Search]
Given a dataset $X$ with $n$ vectors in the $D$-dimensional Euclidean space $\mathbb{R}^D$, and a query vector $q\in\mathbb{R}^D$. Let $\delta(x,y)$ denote the distance between two $D$-dimensional vectors $x$ and $y$. The KNN search retrieves $k$ vectors KNN($q$) in $X$ such that $\delta(x,q)\leq \delta(x',q), \forall x\in \text{KNN}(q)$ and $x'\notin \text{KNN}(q)$.
%with minimal distance to $q$ by evaluating $\delta(x,q)$.
\end{definition}
With the rapid growth of dataset dimensionality and scale, exact KNN search fails to meet the efficiency requirements, known as the curse of dimensionality~\cite{ApproximateNeighbors-indyk-STOC-1998}. Therefore, most efforts have focused on approximate AKNN search, which trades a small loss in accuracy for a significant gain in efficiency. Recall@K is commonly used as a metric to evaluate the precision of the AKNN search algorithm. Recall@K is defined as the ratio of the true $k$ nearest neighbors that appear in the $k$ nearest results returned by the AKNN search algorithm, and it is computed as follows:
\begin{align}
    \text{Recall@K} = \frac{|\text{AKNN}(q) \cap \text{KNN}(q)|}{k}
    \label{eq:recall_at_k}
\end{align}
where $\text{AKNN}(q)$ is the top-$k$ nearest results returned by the AKNN search, and $\text{KNN}(q)$ is the actual KNN for a given $q$. Based on the above definition, numerous AKNN Search Indexes have been proposed with the aim of achieving higher recall at a lower cost.

% 
% We also define the RKNN as a variant of AKNN search and it is critical for our RNSM algorithm.
% \begin{definition}[Reverse $k$-Nearest Neighbor (RKNN)~\cite{RNN-query}]
%   \label{def:rknn}
%   Given a dataset $X$, let $N_k(y,X)\subseteq X\setminus\{y\}$ denote the set of the $k$ nearest neighbors of $y$ in $X$. For two
%   distinct points $x,y\in X$, $y$ is a reverse $k$-nearest neighbor of $x$ if $x\in N_k(y,X)$. Accordingly, the reverse $k$-nearest-neighbor set
%   of $x$ is
%   \[
%   R_k(x,X)=\{y\in X\setminus\{x\}\mid x\in N_k(y,X)\}.
%   \]
%   When $X$ is clear from the context, we abbreviate $N_k(x,X)$ and $R_k(x,X)$ as $N_k(x)$ and $R_k(x)$, respectively.
%   \end{definition}
%   \Rcolor{SHRINK: no need for a formal definition.[DISCUSS WITH LEI]}
% 

% Reverse $k$-nearest neighbor (R$k$NN) search defined as data points that regard a query vector as one of their $k$-nearest neighbors ($k$NNs). 

\subsection{Graph-Based AKNN Search Index}\label{sec:pg-definition}
Most AKNN search algorithms find the approximate nearest neighbors via well-designed indexes, which can be divided into four categories: hash-based~\cite{DynamicCollisionCounting-gan-SIGMOD-2012,LSH-datar-SODA-2004,DFLSH-shan-ENG-2025,DBLSH-tian-IEEE-2024,DETLSH-wei-VLDB-2024,LSBTree-tao-ACM-2010,DynamicCollisionCounting-gan-SIGMOD-2012,SimilaritySearchHashing-gionis-VLDB-1999,SRS-sun-VLDB-2014, SubspaceCollision-wei-SIGMOD-2025}, quantization-based~\cite{ProductQuantization-jegou-IEEE-2011,OptimizedProductQuantization-ge-PAMI-2014,InvertedMultiIndex-babenko-IEEE-2015,RaBitQ-gao-SIGMOD-2024,ExtendRaBitQ-gao-sigmod-2025,PQFastScan-andre-VLDB-2015,EffectiveGeneralDistance-yang-ICDE-2024,LVQ-aguerrebere-VLDB-2023,ADSampling-gao-ACM-2023,MRQ-yang-arxiv-2024}, tree-based~\cite{CoverTrees-beygelzimer-ICML-2006,sa-tree-navarro-vldb-2002,RandomProjectionTrees-dasgupta-ACM-2008,Hercules-echihabi-vldb-2022}, and graph-based~\cite{FixedDimensions-arya-SODA-1993, FANNG-harwood-CVPR-2016,HNSW-malkov-IEEE-2020, NSW-malkov-IS-2014, NSG-fu-VLDB-2019, SSG-fu-IEEE-2022, Efficient-peng-SIGMOD-2023,DiskANN-subramanya-NIPS-2019,FINGER-chen-ACM-2023,SymphonyQG-gou-sigmod-2025, MIRAGEANNS-Sairaj-SIGMOD-2025, DIGRA-jiang-SIGMOD-2025, EnhanceGraph-zhong-arxiv-2025, Flash-wang-SIGMOD-2025,MSPCA-SIGMOD-2026-Xiaochun-Yang}. Among them, the Graph-based methods have gained prominence due to their state-of-the-art search efficiency. Specifically, the Graph-based AKNN indexes take each vector in $X$ as a node on the graph and construct a Proximity Graph (PG) to represent the similarity relationships between vectors. Then, given a dataset $X$ with $n$ vectors in the $D$-dimensional Euclidean space $\mathbb{R}^D$, we construct a graph index $G=(V, E)$ for $X$ with vertex set $V$ and edge set $E$: Every $v\in V$ corresponds to a unique vector $x$ in $X$, and an edge $(u,v)\in E$ denotes the neighbor relationship between $u,v\in V$.

Most graph-based methods utilize a beam search algorithm presented in Algorithm~\ref{algo:search-graph} to retrieve AKNNs. During the beam search, set $C$ is maintained as the candidate nearest neighbors of $q$, and its size $ef$ is used to measure the search effort (line 1). The beam search starts from the given entry points (line 2). At each step of the search, it checks the neighborhood of the nearest unexpanded node and tries to insert these neighbors into $C$ (lines 4-9). The algorithm terminates when no unexpanded node can be found in $C$ (line 10).

The PG index is constructed in two stages: (1) Perform beam search in Algorithm~\ref{algo:search-graph} on the current index and obtain neighbor candidates $W$ from top-$ef_{construction}$ ($ef_{construction}\gg ef_{search}$) results~\cite{NSW-malkov-IS-2014, HNSW-malkov-IEEE-2020} or visited nodes during searching~\cite{NSG-fu-VLDB-2019, Efficient-peng-SIGMOD-2023,DiskANN-subramanya-NIPS-2019}. (2) Apply the neighbor pruning rule of RNG (Relative Neighbourhood Graph)~\cite{RNG-toussaint-PR-1980, sa-tree-navarro-vldb-2002} or its variants~\cite{Efficient-peng-SIGMOD-2023,SSG-fu-IEEE-2022,DiskANN-subramanya-NIPS-2019, RevisitingIndexConstruction-yang-vldb-2025,Fast-Conver-PG-SIGMOD-2026-shangqi,Hemi-Sphere-Graph-SIGMOD-2026-Jing-Tang} to $W$ to select neighbors. Reverse edges are further added for better connectivity.

\begin{algorithm}[!t]
  \caption{AKNN Beam Search}
  \label{algo:search-graph}
  \scriptsize
  %\begin{small}
  \KwIn{Graph Index $G=\left(V,E\right)$, Entry Points $\{ep\}$, Query $q$, Search Effort $ef$, Result Number $k$}
  \KwOut{Top-$k$ approximate nearest neighbors of $q$}
  $\text{Initialize priority queue } C\gets\emptyset$\;
  $C \gets \{(ep, \delta(ep, q))\}$\;
  $\text{Each node }v\in V \text{ set } v.\text{expanded}\gets false$\;
  \While{$C \text{ has unexpanded nodes}$}{
    $u \gets \text{nearest unexpanded node in } C$; $u.\text{expanded}\gets true$\;
    \For{$\text{each unexpanded neighbor } v \text{ of } u$}{
       $C\gets C\bigcup\{(v,\delta(v, q))\}$\; 
    }
    \If{$|C|>ef$}{
        $\text{resize } C \text{ to size } ef$\;
    }
  }
  \Return{$k$ nearest nodes in $C$}
  %\end{small}
  
\end{algorithm}

\section{PROBLEM ANALYSIS AND MOTIVATION}\label{sec:problem-analysis}

In this section, we formalize the multi-index merging problem and present key observations that motivate our method (\S~\ref{sec:analysis-motivation}). We then introduce the pipeline of our merging framework (\S~\ref{sec:framework-intro}).

\subsection{Problem Analysis and Motivation}\label{sec:analysis-motivation}

We begin our analysis by defining the different partitions to merge. We consider $m$ disjoint partitions $\{X_1,X_2,\ldots,X_m\}$ of a dataset $X$, where $X=\bigcup_{i=1}^{m}X_i$ and $X_i\cap X_j=\emptyset$ for all $i\neq j$. A \textbf{random partition} consists of equal-sized subsets formed by randomly assigning vectors, whereas a \textbf{cluster partition} consists of subsets with different sizes reflecting clustered (\ie kmeans) data.
% \begin{definition}[Random and Cluster Partitions]
% \label{def:partition-types}
% Given a dataset $X$, let $\{X_1,X_2,\ldots,X_m\}$ be $m$ disjoint partitions such that $X=\bigcup_{i=1}^{m}X_i$ and $X_i\cap X_j=\emptyset$ for all $i\neq j$. A \textbf{random partition} is obtained by randomly assigning the vectors to the $m$ partitions while enforcing equal partition sizes. A \textbf{cluster partition} is obtained by running clustering methods (\ie kmeans) with $m$ centroids and assigning each vector to its nearest centroid. Consequently, cluster partitions may differ in size and exhibit non-uniform vector distributions.
% \end{definition}
This disjoint-partition setting simplifies the formulation, while real-world updates may produce more complex scenarios such as overlapping indexes. However, as discussed in \textbf{\S~\ref{sec:intro}}, search-based methods also support merging overlapping indexes with unique vector IDs. Given $m$ random/cluster partitions $\{X_1, X_2,\ldots, X_m\}$ following the above definition, the $m$-graph index merge problem is defined as follows:

\stitle{Problem Definition.}
Given $m$ well-constructed graph indexes $I=\{G_1, G_2,\ldots, G_m\}$, where $G_i=(V_i, E_i)$ and $V_i=X_i$, the $m$-graph index merge problem aims to construct a unified graph index $G=(V, E)$ that satisfies the following constraints:
\begin{itemize}[left=0pt, after=\vspace{0pt}]
    \item \emph{Search Quality}: The merged graph $G$ achieves comparable search performance to an index built from scratch on $X$.
    \item \emph{Computational Efficiency}: The merging of multiple indexes requires less time than reconstructing an index for $X$ from scratch.
\end{itemize}

The straightforward index merging paradigm \textbf{\textit{Cross-query Merge (CM)}} is based on pairwise searching: For each point $x$ in a source index $X_i$, it queries all $m-1$ target indexes $I\setminus \{G_i\}$ to retrieve a candidate neighbor set $W$ of size $ef_c$, which corresponds to the construction parameter. However, since it is unknown which indexes contain the top-$ef_c$ candidates, all $m-1$ target indexes must be searched, incurring $m-1$ times $ef_c$-NN searches and $(m-1)\cdot ef_c$ candidates in total. As we have $|V|$ points to merge, it requires $O(|V|(m-1))$ searches with $O(|V|\cdot (m-1)\cdot ef_c)$ intermediate candidates, which is both time-consuming and could even take longer time than reconstruction. Obviously, reducing $m-1$ and $ef_c$ could reduce this large complexity. In the following, we provide two observations on how we could reduce them without sacrificing index quality.

\begin{figure}[!t]
\centering
\begin{small}

% 全局图例（两方法）
\begin{tikzpicture}
\begin{customlegend}[legend columns=3,
legend entries={Rebuild, $\CM_{k=200}$, $\CM_{k=1}$, $\CM_{k=5}$, $\CM_{k=10}$, $\CM_{k=50}$},
legend style={at={(0.5,1)},anchor=south,draw=none,font=\scriptsize,column sep=0.3cm, yshift=-10pt}]
% 替换颜色
\addlegendimage{line width=0.2mm,color=violate,mark=o,mark size=0.5mm}
\addlegendimage{line width=0.2mm,color=amaranth,mark=triangle,mark size=0.5mm}
\addlegendimage{line width=0.2mm,color=navy,mark=square,mark size=0.5mm}
\addlegendimage{line width=0.2mm,color=amber,mark=pentagon,mark size=0.5mm}
\addlegendimage{line width=0.2mm,color=blue,mark=diamond,mark size=0.5mm}
\addlegendimage{line width=0.2mm,color=forestgreen,mark=star,mark size=0.5mm}
\end{customlegend}
\end{tikzpicture}
%\\[-15pt]
\vspace*{-2mm}

\begin{tikzpicture}[scale=1]
\begin{axis}[
    height=3.0cm,
    width=\columnwidth/2.10,
    xlabel=recall@1, xmin=0.90, xmax=1.0,
    ylabel=Query Per Second, ylabel style={yshift=-15pt}, 
    label style={font=\scriptsize},
    tick label style={font=\scriptsize},
    ymajorgrids=true,
    xmajorgrids=true,
    grid style=dashed,xlabel style={inner sep=0pt, yshift=4pt},
]
% buildasone
\addplot[line width=0.2mm,color=violet,mark=o,mark size=0.5mm]
plot coordinates {
(0.9235, 2716.5)
(0.9539, 2107.5)
(0.9679, 1749.6)
(0.9761, 1483.8)
(0.9801, 1282.2)
(0.9822, 1158.0)
(0.9842, 1025.5)
(0.9845, 938.2)
(0.9861, 873.7)
(0.9864, 809.3)
(0.9865, 700.4)
(0.9867, 658.8)
(0.9868, 622.6)
(0.9869, 746.1)
};
% merge
\addplot[line width=0.2mm,color=amaranth,mark=triangle,mark size=0.5mm]
plot coordinates {
(0.8952, 2992.0)
(0.9346, 2375.7)
(0.9555, 1941.2)
(0.9664, 1640.9)
(0.974, 1421.1)
(0.9791, 1262.1)
(0.9831, 1154.9)
(0.9853, 1030.4)
(0.9861, 957.9)
(0.9878, 884.6)
(0.9883, 827.6)
(0.9884, 777.4)
(0.989, 729.9)
(0.9895, 676.7)
};
% 1
\addplot[line width=0.2mm,color=navy,mark=square,mark size=0.5mm]
plot coordinates {
(0.9081, 2608.0)
(0.9347, 2136.2)
(0.9529, 1814.2)
(0.9642, 1580.8)
(0.971, 1401.0)
(0.9752, 1264.9)
(0.9786, 1156.2)
(0.9812, 1051.6)
(0.983, 973.3)
(0.9859, 913.3)
(0.9864, 844.4)
(0.9881, 793.6)
(0.9882, 748.1)
};
% 5
\addplot[line width=0.2mm,color=amber,mark=pentagon,mark size=0.5mm]
plot coordinates {
(0.9249, 2675.4)
(0.9488, 2203.2)
(0.9622, 1877.2)
(0.9705, 1654.6)
(0.9763, 1467.8)
(0.9801, 1289.2)
(0.9838, 1200.8)
(0.986, 1100.8)
(0.9869, 1008.4)
(0.9891, 945.2)
(0.9894, 874.8)
(0.9897, 821.6)
(0.9897, 775.6)
};
% 10
\addplot[line width=0.2mm,color=blue,mark=diamond,mark size=0.5mm]
plot coordinates {
(0.8877, 3365.6)
(0.9528, 2185.2)
(0.964, 1874.9)
(0.9734, 1630.2)
(0.9777, 1457.1)
(0.9821, 1092.1)
(0.9835, 1175.3)
(0.9858, 1088.0)
(0.9864, 991.4)
(0.9873, 865.8)
(0.9875, 927.7)
(0.9883, 810.3)
(0.9891, 764.0)
};
% 50
\addplot[line width=0.2mm,color=forestgreen,mark=star,mark size=0.5mm]
plot coordinates {
(0.9081, 3208.7)
(0.9409, 2525.6)
(0.9584, 2076.8)
(0.9707, 1735.3)
(0.9784, 1526.8)
(0.9828, 1366.4)
(0.9853, 1235.6)
(0.9869, 1112.5)
(0.9877, 1014.7)
(0.9892, 949.7)
(0.9898, 814.9)
(0.9902, 766.3)
(0.9903, 872.3)
(0.9906, 736.6)
};
\end{axis}
\end{tikzpicture}
\begin{tikzpicture}[scale=1]
\begin{axis}[
    height=3.0cm,
    width=\columnwidth/2.10,
    xlabel=recall@10, xmin=0.90, xmax=1.0, 
    ylabel=Query Per Second, ylabel style={yshift=-15pt}, 
    label style={font=\scriptsize},
    tick label style={font=\scriptsize},
    ymajorgrids=true,
    xmajorgrids=true,
    grid style=dashed,xlabel style={inner sep=0pt, yshift=4pt},
]
% buildasone
\addplot[line width=0.2mm,color=violet,mark=o,mark size=0.5mm]
plot coordinates {
(0.76755, 19161.6)
(0.8811, 12573.9)
(0.92757, 9334.5)
(0.95138, 7556.3)
(0.96534, 6311.6)
(0.97363, 5457.6)
(0.97918, 4739.2)
(0.983, 4226.8)
(0.98544, 3830.1)
(0.98714, 3456.2)
(0.98844, 3224.6)
(0.98943, 2994.0)
(0.99007, 2758.7)
(0.99077, 2603.2)
(0.99127, 2443.0)
};
% merge
\addplot[line width=0.2mm,color=amaranth,mark=triangle,mark size=0.5mm]
plot coordinates {
(0.77005, 20327.4)
(0.88946, 13236.1)
(0.93593, 9815.7)
(0.95835, 7908.0)
(0.97181, 6567.0)
(0.98005, 5702.4)
(0.9855, 4994.6)
(0.98912, 4497.9)
(0.99151, 4058.5)
(0.99341, 3704.7)
(0.99437, 3414.4)
(0.99518, 3171.2)
(0.99596, 2942.1)
(0.99645, 2778.1)
(0.99685, 2614.6)
};
% 1
\addplot[line width=0.2mm,color=navy,mark=square,mark size=0.5mm]
plot coordinates {
(0.63746, 22520.8)
(0.792, 14313.4)
(0.8655, 10840.0)
(0.90566, 8792.4)
(0.93127, 7277.4)
(0.94798, 6330.8)
(0.95928, 5599.3)
(0.96756, 4996.5)
(0.97323, 4568.2)
(0.97811, 4091.5)
(0.98146, 3806.8)
(0.98421, 3474.8)
(0.98641, 3266.0)
(0.98822, 3069.5)
(0.98953, 2874.2)
};
% 5
\addplot[line width=0.2mm,color=amber,mark=pentagon,mark size=0.5mm]
plot coordinates {
(0.676, 23549.2)
(0.81676, 14703.8)
(0.88291, 11567.2)
(0.91914, 9280.0)
(0.94224, 7717.7)
(0.9575, 6735.1)
(0.96734, 5839.3)
(0.97386, 5273.5)
(0.97909, 4749.2)
(0.98312, 4264.5)
(0.98598, 3969.6)
(0.98835, 3580.4)
(0.99007, 3338.7)
(0.99117, 3175.6)
(0.99232, 2919.8)
};
% 10
\addplot[line width=0.2mm,color=blue,mark=diamond,mark size=0.5mm]
plot coordinates {
(0.69366, 22952.5)
(0.83076, 15574.6)
(0.89214, 11664.9)
(0.9261, 9403.6)
(0.94612, 7806.9)
(0.9611, 6686.0)
(0.97026, 5112.2)
(0.97655, 5295.0)
(0.98104, 4773.7)
(0.98487, 4329.7)
(0.98725, 3986.4)
(0.98973, 3686.3)
(0.99122, 3431.3)
(0.99256, 3136.3)
(0.9935, 3026.4)
};
% 50
\addplot[line width=0.2mm,color=forestgreen,mark=star,mark size=0.5mm]
plot coordinates {
(0.72894, 21691.5)
(0.8592, 14719.6)
(0.91391, 10970.7)
(0.94254, 8722.7)
(0.96005, 7319.4)
(0.97139, 6313.2)
(0.97852, 5546.2)
(0.98324, 4940.6)
(0.98678, 4427.7)
(0.98937, 4077.4)
(0.99126, 3743.6)
(0.99282, 3470.0)
(0.99397, 3229.0)
(0.99498, 3055.2)
(0.99553, 2896.1)
};
\end{axis}
\end{tikzpicture}

\vgap
\vspace*{-2mm}
\caption{Nearest Neighbors being Critical Factor}
\vgap
\label{fig:observation1}
\end{small}
\end{figure}

\stitle{Observation 1: Nearest Neighbors are Key Contributing Components of Merging.}\label{ob:Ob1}
As mentioned above, most graph-based methods construct proximity graphs using a two-stage framework, where higher-quality pruning candidates $W$ lead to better index quality. These pruning candidates $W$ aim to cover a large region around a given vertex, which directly leads to a high time cost. $W$ can be divided into two parts: (1) Nearest neighbors that provide good connectivity among KNNs. (2) Distant neighbors that enable efficient graph traversal by introducing long-range edges and covering broader regions of the data space~\cite{ExtendedNeighborhoodGraph-zhang-SIGIR-2025}. We find that \textbf{nearest neighbors are critical for merging quality}, as the distant neighbors are already captured within sub-indexes. 
This phenomenon can be verified in the following example:
%The experiment in Fig~\ref{fig:observation1} shows that including distant neighbors ($k=200$)  yields negligible improvement over using only top-50 neighbors. 

\begin{exmp}
In this example, we compare the search performance of the rebuilt index with the merged index from sub-indexes of 2 equal-sized disjoint random partitions of the SIFT1M dataset. We first construct an HNSW index directly on the complete dataset with $ef_c=200$ and denote it as \textbf{Rebuild}. Then we relax the per-target-index search effort of CM from the fixed $ef_c$ to $k=[1,5,10,50,200]$, where 200 corresponds to the default setting. $CM_k$ indicates that the top-$k$ results from the search are used for the merge of the 2 sub-indexes, so $k=1$ keeps only the nearest neighbor. As illustrated in Fig~\ref{fig:observation1}, we investigate the impact of $k$ from target indexes on merged index quality. Specifically, incorporating nearest neighbors from other indexes dramatically improves performance even with just $k=1$, and it achieves the identical performance of \textbf{Rebuild} at high recall ($\ge 0.99$). With larger $k$ from target indexes, the merged index achieves a better efficiency-accuracy tradeoff.% for AKNN search.

%\textbf{Cross-query Merge} (CM) is a merged index obtained by pairwise searching with $ef=200$ (retrieving top-$ef$ nearest neighbors). For $ef$=\{1, 5, 10, 50\}, pruning candidates $W$ are constructed by concatenating existing neighbors in the source index with $ef$ actual nearest neighbors from the target index. Incorporating nearest neighbors from other indexes dramatically improves performance even with just $ef=1$, and it achieves the identical performance of \textbf{Rebuild} at high recall ($\ge 0.99$). With larger $ef$ values, the merged index achieves a better efficiency-accuracy tradeoff for AKNN search.
\end{exmp}

Based on this observation, we can safely reduce the merging problem to a more efficient formulation. Standard index construction typically requires large pruning candidate sets with $|W| \geq 200$. In contrast, for merging purposes, we only need to identify the top-$k$ nearest neighbors from the target index for each vertex. Here, $k$ can be substantially smaller, \eg $k \in [1, 50]$, which saves the computation significantly by limiting the search to a small number of nearest neighbors from other indexes, making the merging process substantially more efficient with smaller $k$.

\begin{figure}[!t]
\centering
\begin{small}

% 全局图例（两方法）
\begin{tikzpicture}
\begin{customlegend}[legend columns=3,
legend entries={Path, Star, MST},
legend style={at={(0.5,1.15)},anchor=north,draw=none,font=\scriptsize,column sep=0.3cm}]
% 替换颜色
\addlegendimage{line width=0.2mm,color=amaranth,mark=triangle,mark size=0.5mm}
\addlegendimage{line width=0.2mm,color=navy,mark=square,mark size=0.5mm}
\addlegendimage{line width=0.2mm,color=violate,mark=o,mark size=0.5mm}
% \addlegendimage{line width=0.2mm,color=amber,mark=pentagon,mark size=0.5mm}
% \addlegendimage{line width=0.2mm,color=blue,mark=diamond,mark size=0.5mm}
% \addlegendimage{line width=0.2mm,color=forestgreen,mark=star,mark size=0.5mm}
\end{customlegend}
\end{tikzpicture}
\\[-\lineskip]

\begin{tikzpicture}[scale=1]
\begin{axis}[
    height=3.0cm,
    width=\columnwidth/2.10,
    xlabel=recall@1, xmin=0.0, xmax=1.0,
    ylabel=Query Per Second, ylabel style={yshift=-15pt}, 
    label style={font=\scriptsize},
    tick label style={font=\scriptsize},
    ymajorgrids=true,
    xmajorgrids=true,
    grid style=dashed,xlabel style={inner sep=0pt, yshift=4pt},
]
% MST
\addplot[line width=0.2mm,color=violate,mark=o,mark size=0.5mm]
plot coordinates {
(0.519, 17812.4)
(0.624, 10908.8)
(0.679, 7409.37)
(0.708, 6463.16)
(0.743, 5479.39)
(0.758, 4324.31)
(0.772, 4657.94)
(0.785, 4118.71)
(0.803, 3842.73)
(0.812, 3416.79)
(0.812, 3511.6)
(0.818, 3113.66)
(0.826, 2900.08)
(0.831, 2771.42)
(0.835, 2422.38)
};
% Path
\addplot[line width=0.2mm,color=amaranth,mark=triangle,mark size=0.5mm]
plot coordinates {
(0.266, 18646.4)
(0.317, 13557.3)
(0.342, 11065.0)
(0.357, 8986.38)
(0.366, 7474.41)
(0.383, 6041.94)
(0.39, 5596.59)
(0.397, 4922.74)
(0.406, 4415.85)
(0.411, 4088.75)
(0.415, 3661.07)
(0.416, 3477.02)
(0.42, 3013.68)
(0.426, 2138.4)
(0.429, 2356.28)
};
% Star
\addplot[line width=0.2mm,color=navy,mark=square,mark size=0.5mm]
plot coordinates {
(0.467, 18552.9)
(0.55, 11831.6)
(0.591, 8686.38)
(0.613, 6758.57)
(0.643, 5643.57)
(0.659, 4861.44)
(0.671, 4229.1)
(0.675, 3377.71)
(0.681, 2697.57)
(0.686, 1384.55)
(0.696, 1180.62)
(0.7, 2857.34)
(0.703, 2700.0)
(0.706, 2600.41)
(0.71, 2469.71)
};
\end{axis}
\end{tikzpicture}
\begin{tikzpicture}[scale=1]
\begin{axis}[
    height=3.0cm,
    width=\columnwidth/2.10,
    xlabel=recall@10, xmin=0.0, xmax=1.0,
    ylabel=Query Per Second, ylabel style={yshift=-15pt}, 
    label style={font=\scriptsize},
    tick label style={font=\scriptsize},
    ymajorgrids=true,
    xmajorgrids=true,
    grid style=dashed,xlabel style={inner sep=0pt, yshift=4pt},
]
% MST
\addplot[line width=0.2mm,color=violate,mark=o,mark size=0.5mm]
plot coordinates {
(0.4668, 17708.3)
(0.5824, 10734.8)
(0.6501, 7988.85)
(0.6915, 6288.09)
(0.7221, 5330.28)
(0.7402, 4654.79)
(0.7556, 4123.06)
(0.7743, 3587.86)
(0.7906, 3331.76)
(0.802, 3013.19)
(0.8059, 2985.89)
(0.8132, 3750.95)
(0.8221, 3475.71)
(0.8287, 3268.29)
(0.8343, 3156.84)
};
% Path
\addplot[line width=0.2mm,color=amaranth,mark=triangle,mark size=0.5mm]
plot coordinates {
(0.2412, 18021.6)
(0.2984, 8675.2)
(0.3254, 5805.3)
(0.3447, 4860.18)
(0.3571, 4316.82)
(0.3706, 2340.2)
(0.3792, 3081.77)
(0.3881, 3157.88)
(0.3965, 2900.28)
(0.4016, 2264.66)
(0.4045, 1520.32)
(0.4065, 1602.64)
(0.4109, 2709.82)
(0.4162, 2870.8)
(0.4196, 2771.54)
};
% Star
\addplot[line width=0.2mm,color=navy,mark=square,mark size=0.5mm]
plot coordinates {
(0.4359, 15641.5)
(0.524, 9559.29)
(0.5764, 7486.86)
(0.6111, 6007.59)
(0.6376, 4605.96)
(0.6505, 4198.11)
(0.6639, 3202.43)
(0.6724, 2233.31)
(0.679, 1066.14)
(0.6858, 1293.53)
(0.6975, 2681.38)
(0.7037, 2675.45)
(0.7087, 2435.36)
(0.712, 2357.85)
(0.7154, 2239.35)
};
\end{axis}
\end{tikzpicture}

\vgap
\vspace*{-2mm}
\caption{Search Performance of Different Merge Order}
\vgap
\label{fig:observation2}
\end{small}
\end{figure}

\begin{figure*}[!t]
    \centering
    \includegraphics[width=1.0\textwidth]{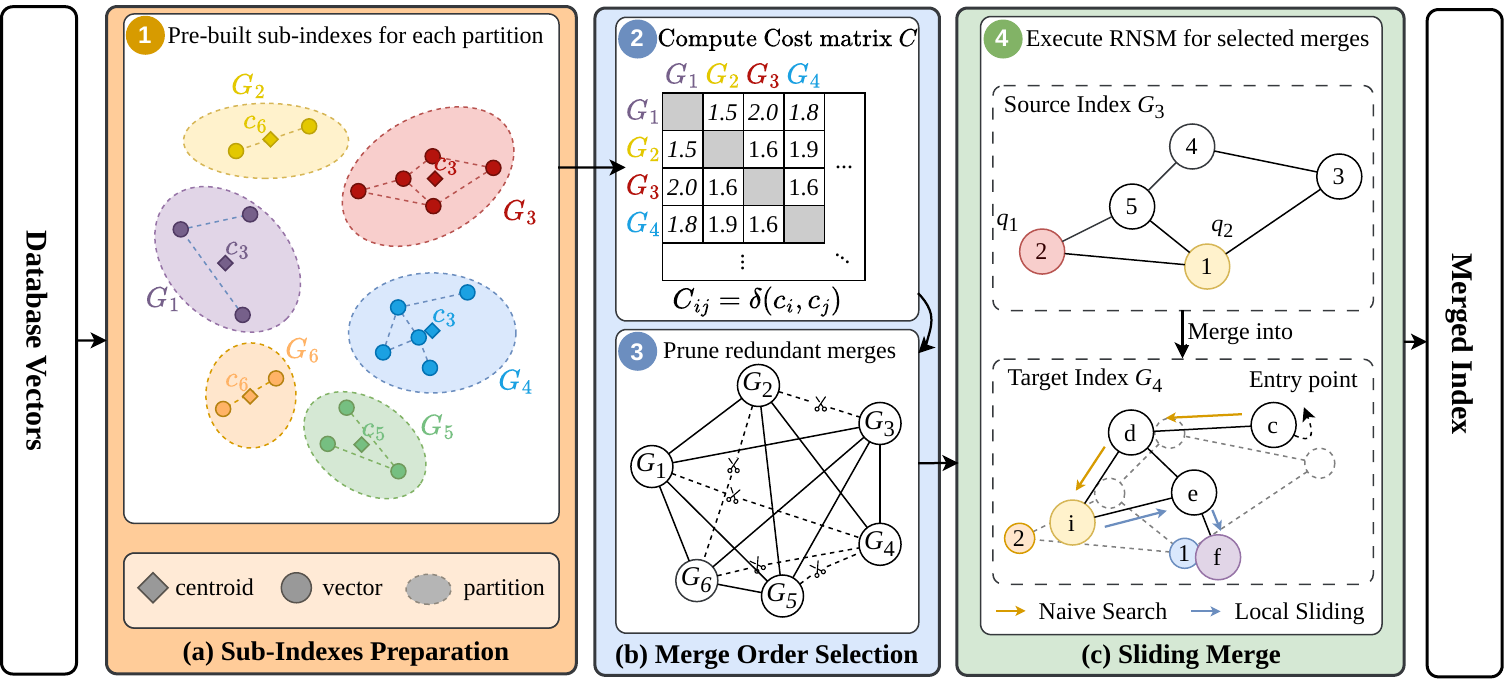}
    \vspace{-6mm}
    \caption{Pipeline of Our Multi-Index Merge Framework}
    \vspace{-3mm}
    \label{fig:framework}
\end{figure*}

\stitle{Observation 2: Different Merge Orders Lead to Varying Search Performances.} % We then consider different merge orders for merging multiple indexes, where 
The merge order specifies which pairs of indexes to merge next. According to our analysis, retrieving KNNs from other indexes for each vertex is critical, which naturally leads to a pairwise processing framework. However, this framework suffers from inherent computational redundancy for the following reasons: (1) KNNs from different partitions may occlude one another during merging. Empirically, only 9.2\% of the 100-NNs remain as neighbors of the vertex after neighbor selection. (2) For cluster partitions, KNNs of a vertex are often concentrated in a subset of partitions rather than uniformly distributed across all partitions. Moreover, we observe that \textbf{even with the same number of merge operations, different merge orders result in significantly different search performance,} as shown in the following example:

% [SHRINK ANALYSIS]
\begin{exmp}
    We compare ten random partitions of SIFT1M merged with the same CM pairwise algorithm under three merge orders, isolating the effect of the merge order. \textbf{Path} forms a merging chain, \textbf{Star} merges all indexes with a central one, and \textbf{MST} builds a minimum spanning tree based on inter-index similarity. All three orders use the same $m-1$ merge operations. As shown in Fig~\ref{fig:observation2}, \textbf{Path} yields low recall because cross-partition KNNs may lie at opposite ends of the chain, while \textbf{Star} and \textbf{MST} improve cross-partition reachability. Moreover, \textbf{Star} underperforms \textbf{MST} because its central hub incurs an imbalanced merge workload.
\end{exmp}

\subsection{Multi-index Merging Framework Overview}\label{sec:framework-intro}

Fig~\ref{fig:framework} illustrates the pipeline of our proposed RNSM + MOS framework, which has the following three stages:

%
% Revise Here
%

\stitle{a) Sub-Indexes Preparation}.  
While existing methods are tailored to either cluster partitions~\cite{DiskANN-subramanya-NIPS-2019, ExtendedNeighborhoodGraph-zhang-SIGIR-2025} or random partitions~\cite{CSPG-yang-NIPS-2024,FGIM_wu_2026_SIGMOD}, our framework can be applied to both partition types without compromising performance. Given pre-built sub-indexes, we first prepare the centroids of each partition. Then the partition centroids are utilized for merge order selection to compute the cost matrix $C$.

\stitle{b) Merge Order Selection (MOS)}. Next, we reduce the cost of multi-index merging through MOS. It acts as a lightweight merge planner on top of pairwise RNSM. We formulate MOS as selecting a connected merge order graph that minimizes the total cost of the selected pairwise merges, subject to degree and diameter bounds that limit total merge workload and cross-partition traversal length, respectively. We also use centroid distances to prioritize the merge of closer partitions, as nearby indexes are more likely to share KNNs and thus benefit more from merging~\cite{douze2026FaissLibrary}. Specifically, we first compute pairwise centroid distances to form a cost matrix $C$. For random partitions, we prove they are identical in expectation and use a constant. For cluster partitions, such pairwise centroid distance computation cost is negligible given that the partition number is much smaller than the dataset size. Then, we construct a sparse merge order graph that effectively prunes unnecessary merges (\eg $G_2$ and $G_3$) according to $C$ and the requirements of degree and diameter.

\stitle{c) Sliding Merge}.
Finally, we execute each edge selected by MOS as an RNSM operation following two common principles: merge smaller indexes into larger ones~\cite{Milvus-wang-ACM-2021}, and merge closer partitions first~\cite{douze2026FaissLibrary}. RNSM exploits prior search results to accelerate later searches. Specifically, we greedily select nodes with the maximum number of Reverse $k$-Nearest Neighbors (RKNN) in the source index as pivots. RKNN of a vector $x$ is a vector that regards $x$ as one of its KNNs. Then, the search result $i$ of pivot $2$ is used as an entry point to find cross-index neighbors for RKNNs of pivot $2$.

\section{RNSM for Pairwise Indexes Merge}
\label{sec:Bi-Index-Merge}

This section details the proposed Reverse Neighbor Sliding Merge framework for merging two graph indexes. We first formally define the concept of \textit{sliding merge} and derive a corresponding cost model in \S~\ref{sec:SlidingMergeCostModel} to quantify the merging overhead. Then, we introduce the RNSM algorithm in \S~\ref{sec:RNSM-detail}, demonstrating how it accelerates the merging process by reducing the derived cost objective.

\subsection{Sliding Merge Cost Model}\label{sec:SlidingMergeCostModel}

Our preliminary analysis identifies the AKNN search within the target index as the primary bottleneck in index merging (see Fig~\ref{fig:observation1}). Consequently, optimizing the search process on the target index is paramount. Prior research established that initializing graph-based searches with entry points topologically closer to the query target significantly reduces the search path length, thereby reducing latency~\cite{LSH-APG-zhao-VLDB-2023, DiskANN-subramanya-NIPS-2019, NSG-fu-VLDB-2019}. Building on this insight, we propose \textit{Sliding Merge}, a strategy that leverages search results from neighboring nodes to accelerate processing. In this section, we formalize the cost model and define the optimization problem.

\stitle{Notation and Sliding Semantics.}
Given a source graph index $G_1=(V_1, E_1)$ and a target index $G_2=(V_2, E_2)$, our objective is to merge $G_1$ into $G_2$ by performing an AKNN search in $G_2$ for every vector $x \in V_1$. The \textit{Sliding Merge} framework employs two distinct search operations:
\begin{enumerate}[leftmargin=4\labelsep]
    \item \textbf{Naive Search:} A standard AKNN search for query $x \in V_1$ starting from the default entry point of $G_2$.
    \item \textbf{Local Sliding:} A search for query $x$ initialized using the result of a previous search for a pivot node $p$.
\end{enumerate}
Let $\mathcal{L}(p, x)$ denote the search cost for query $x$. If a pivot $p$ is used, the cost depends on the proximity between $p$ and $x$; otherwise, a fixed base cost is incurred.

\stitle{Cost Metric.}
To ensure our model is hardware-agnostic, we employ the Number of Distance Computations (NDC) as the cost metric, a standard practice in nearest neighbor search analysis~\cite{VBASE-zhang-OSDI-2023,VSAG-zhong-vldb-2025,DiskANN-subramanya-NIPS-2019,HNSW-malkov-IEEE-2020,ADSampling-gao-ACM-2023}.
We define the cost of a \textit{Naive Search} as a constant $\gamma$, derived from the average NDC required for a search from the default entry point to reach the prescribed recall level, like 0.98. However, estimating $\gamma$ from target recall is impractical in real systems because of the non-existent ground truth. We therefore use a sufficiently large naive-search effort $ef_N$ (a common graph index setting~\cite{GraphBasedANNS-wang-vldb-2021}) to ensure reliable pivot searches. Thus, $\gamma$ is used for formalization while $ef_N$ is the practical parameter. Moreover, our experiment shows a strong linear correlation ($R^2=0.997$) between search effort $ef$ and NDC. We therefore approximate the NDC of a search with effort $ef$ as $\alpha\cdot ef$, where $\alpha$ is the empirically fitted NDC--$ef$ slope. Accordingly, the Naive Search cost is approximated as $\gamma\approx\alpha\cdot ef_N$.
For \textit{Local Sliding}, recent findings~\cite{SimJoin-xie-POM-2025} indicate that when a query $x$ starts from a pivot $p$ proximate in the high-dimensional space, the NDC required to converge follows a linear relationship with their distance $\delta(p,x)$ (see Fig~\ref{fig:ndc-dist}). We measured the 95th percentile of the $k$-th NN distance and found that the top-500 NNs empirically fall within this linear range. Since RNSM performs local sliding for top-3\textasciitilde10 NNs, we can approximate the sliding NDC cost as $\mathcal{L}(p,x)\approx \beta\cdot\delta(p,x)$, where $\beta$ is the empirically fitted distance--NDC slope. By matching the NDC with $\alpha\cdot ef=\beta\cdot\delta(p,x)$, we can estimate the search effort $ef_L$ for sliding from $p$ to $x$ as $ef_L=\lambda\cdot\delta(p,x)$, where $\lambda=\beta/\alpha$.

\stitle{Problem Formulation.}
While local sliding offers lower costs than naive search (see Fig~\ref{fig:ndc-dist}), it necessitates a prerequisite: the pivot $p$ must have already been searched. Ideally, every node would slide from its nearest resolved neighbor. However, identifying the optimal global pivot for $x$ would itself require a costly search.
To make this tractable, we exploit the topology of the source graph $G_1$. Since $G_1$ is a proximity graph, neighbors in $G_1$ are inherently close in the vector space~\cite{NSG-fu-VLDB-2019}. Thus, we restrict the candidate pivots for a node $v$ to its immediate neighbors $N(v)$ in $G_1$.
This transforms the merge optimization into a pivot selection problem: we must select a subset of nodes to perform Naive Search (pivots), such that all other nodes can perform low-cost Local Sliding from a pivot. Moreover, we require pivots to be independent of each other to fully decouple them. We formally define this as follows:

\begin{definition}[Minimum Cost Independent Dominating Pivot Selection (IDPS)]\label{defn:DPS}
\
\begin{itemize}[leftmargin=4\labelsep]
    \item \textbf{Input:} A vector dataset $X$, a graph $G=(V, E)$ where node $v_i \in V$ corresponds to $x_i \in X$, a distance function $\delta(\cdot,\cdot)$, and a naive search cost $\gamma$, distance-NDC slope $\beta$.
    \item \textbf{Output:} A pivot subset $P \subseteq V$ that minimizes the total cost $\mathcal{L}(P)$, subject to the independent dominating set constraint.
    \item \textbf{Constraint:} Every node $v \in V$ must be either a pivot or covered by a pivot:
      $$ v \in P \lor \exists p \in P \text{ s.t. } (v, p) \in E $$
      Additionally, $P$ is an \emph{independent set}: no two pivots are
      adjacent in $G$.
    \item \textbf{Objective:} Minimize $\mathcal{L}(P) = \mathcal{L}_{naive}(P) + \mathcal{L}_{slide}(P)$, where:
    \begin{align*}
        \mathcal{L}_{naive}(P) &= \gamma \cdot |P| \\
        \mathcal{L}_{slide}(P) &= \sum_{v \in V \setminus P} \min_{p \in P \cap N(v,G)} \beta\cdot\delta(x_v, x_p)
    \end{align*}
\end{itemize}
\end{definition}

\begin{figure}[!t]
\centering
\begin{small}

\begin{tikzpicture}
\begin{customlegend}[legend columns=4,
legend style={at={(0.5,1.15)},anchor=north,draw=none,font=\scriptsize,column sep=0.3cm}]
% \addlegendimage{line width=0.2mm,color=violate,mark=o,mark size=0.8mm}
\end{customlegend}
\end{tikzpicture}
\\[-\lineskip]

\subfloat[IMAGENET1M]{\vspace{-2mm}
\begin{tikzpicture}[scale=1]
\begin{axis}[
    height=3.0cm,
    width=\columnwidth/2.00,
    xlabel={$\delta(x_i,x_j)$},
    ylabel=NDC, ylabel style={yshift=-15pt}, 
    label style={font=\scriptsize},
    tick label style={font=\scriptsize},
    ymajorgrids=true,
    xmajorgrids=true,
    grid style=dashed,
    smooth,xlabel style={inner sep=0pt, yshift=4pt},
]
% NaiveMerge
\addplot[line width=0.3mm,color=violate,smooth]
plot coordinates {
(64.30,139.92)
(104.92,156.94)
(152.30,161.50)
(172.61,175.27)
(179.38,183.17)
(192.91,179.90)
(199.68,180.82)
(215.76,184.17)
(224.22,189.78)
(232.68,196.54)
(249.60,211.99)
(258.06,214.34)
(266.52,211.55)
(274.99,224.20)
(283.45,230.15)
(321.52,232.46)
(327.44,263.59)
(329.14,290.02)
(330.83,306.98)
(332.52,329.13)
(335.91,355.32)
(337.60,447.84)
};
\draw[red,dashed,line width=0.8pt]
    ({axis cs:242.8620,0} |- {rel axis cs:0,0}) --
    ({axis cs:242.8620,0} |- {rel axis cs:0,1});
\end{axis}
\end{tikzpicture}}\hspace{2mm}
\subfloat[DEEP10M]{\vspace{-2mm}
\begin{tikzpicture}[scale=1]
\begin{axis}[
    height=3.0cm,
    width=\columnwidth/2.00,
    xlabel={$\delta(x_i,x_j)$},
    ylabel=NDC, ylabel style={yshift=-15pt},
    label style={font=\scriptsize},
    tick label style={font=\scriptsize},
    ymajorgrids=true,
    xmajorgrids=true,
    grid style=dashed,
    smooth,xlabel style={inner sep=0pt, yshift=4pt},
]
% NaiveMerge
\addplot[line width=0.3mm,color=violate,smooth]
plot coordinates {
(0.10,79.00)
(0.16,300.33)
(0.22,296.07)
(0.28,392.92)
(0.34,414.85)
(0.40,407.42)
(0.46,372.62)
(0.52,476.79)
(0.55,513.83)
(0.58,606.89)
(0.61,626.28)
(0.64,657.11)
(0.67,629.63)
(0.70,722.28)
(1.10,999.09)
(1.13,1106.60)
(1.16,1471.70)
(1.25,1898.91)
(1.27,2160.85)
(1.29,2217.85)
(1.31,2657.59)
(1.33,7213.67)
(1.35,7699.70)
};
\draw[red,dashed,line width=0.8pt]
    ({axis cs:0.8607,0} |- {rel axis cs:0,0}) --
    ({axis cs:0.8607,0} |- {rel axis cs:0,1});
\end{axis}
\end{tikzpicture}}\hspace{2mm}

\vgap
\caption{Verification of Distance Cost Model. In datasets IMAGENET1M and DEEP10M, the cost of sliding from $x_i$ to $x_j$ is linear to the distance when $x_i$ and $x_j$ are close\cite{SimJoin-xie-POM-2025}. The red dashed lines mark the 95th percentile distance of the top-500 nearest neighbors in each dataset.}
\vgap
\label{fig:ndc-dist}
\end{small}
\end{figure}

The independence constraint (no pivot slides from another pivot) is a conservative design choice for parallel execution rather than a guarantee of the minimum pivot count. Although allowing adjacent pivots can sometimes reduce the number of Naive Searches, reusing search results along an adjacent-pivot chain introduces multi-hop sliding and sequential dependencies, exemplified by the MST-based propagation in SimJoin~\cite{SimJoin-xie-POM-2025}. Such dependencies limit parallel scalability, as shown in Fig~\ref{fig:parallel-scalability}. We next show the computational hardness of this IDPS problem.

\begin{theorem}\label{thm:dps-nphard}
The Minimum Cost Independent Dominating Pivot Selection (IDPS) problem is NP-hard.
\end{theorem}

\begin{proof}[Proof Sketch]
We reduce from the NP-complete Minimum Independent Dominating Set (MIDS) problem~\cite{MIDS-Irving-1991}.
Given a MIDS instance $\langle G', k \rangle$, set $G = G'$, define the discrete metric $\delta \equiv 1$ for distinct nodes, and set $\gamma = 2$, $\beta = 1$, $C = k + |V|$.
For any feasible $P$ (an independent dominating set), the cost simplifies to $\mathcal{L}(P) = 2|P| + (|V|-|P|) = |P| + |V|$.
Hence $\mathcal{L}(P) \leq C$ iff $|P| \leq k$, establishing a one-to-one correspondence between MIDS solutions and IDPS solutions.
The full proof is provided in the Appendix.
\end{proof}

Given the intractability of the IDPS problem, we resort to a greedy algorithm for pivot selection presented in Algorithm~\ref{algo:bi-merge}.

\subsection{Reverse Neighbor Sliding Merge}\label{sec:RNSM-detail}

%While the DPS problem provides a theoretical objective, it is NP-hard.  1) \textit{Selection Dependency}: The unit value of a candidate pivot changes whenever a neighbor is selected, requiring sequential updates that break parallel efficiency.
A straightforward algorithm for IDPS is to adopt a greedy strategy: iteratively selecting the pivot with the highest \textit{unit value} (coverage nodes divided by cost). However, applying this directly to index merging is problematic for two reasons. 1) \textit{Limited Visibility}: Relying solely on the sparse connections in the source PG limits the ability to find the true nearest pivot, potentially forcing sub-optimal sliding paths. 2) \textit{Computationally Expensive}: Since pivots and their sliding nodes are not necessarily adjacent in the source index, calculating unit values requires either exhaustive all-pairs distance computation or reliance on graph indexes that often store neighbor rankings rather than explicit distance values \cite{SimJoin-xie-POM-2025}, compromising both flexibility and precision.

\stitle{The RNSM Strategy.}
To overcome these limitations, we propose \textit{Reverse Neighbor Sliding Merge (RNSM)}. The core intuition is to resolve pivot selection \textit{globally} before merging, enabling greater parallelism through independent pivots. It uses RKNN relations to identify pivots: if $x\in N_k(y)$, then $x$ is inherently a good pivot for $y$ due to proximity. Consequently, nodes with many RKNNs are natural hubs that can serve as cost-efficient pivots for many neighbors. We first introduce two core components of the RNSM: Pivot Selection and Local Sliding.
% Instead of dependent selection, we exploit the proximity properties of the Reverse $k$-Nearest Neighbor (RKNN) graph.
% We find that if a node $x$ appears in the KNN list of $y$, then $x$ is inherently a good pivot for $y$ due to proximity. Consequently, nodes with a high-quality RNN graph are natural hubs that can serve as cost-efficient pivots for many neighbors. We first introduce two core components of the RNSM: Pivot Selection and Local Sliding. 

\sstitle{Pivot Selection.} We first identify the pivots and their corresponding followers to perform local sliding before executing the merge. Each node in the source index $G_1$ is classified into:
\begin{enumerate}[leftmargin=4\labelsep]
    \item \textbf{Pivots $p\in P$:} Neighbor candidates obtained by a naive search in target index $G_2$.
    \item \textbf{Follower set $F_p$:} Neighbor candidates obtained by local sliding from search results of pivot $p$.
    \item \textbf{Uncovered set $U$:} The remaining nodes.%Nodes assigned neither as pivots nor as followers.
\end{enumerate}

As shown in Fig~\ref{fig:ndc-dist}, for the search results of one node to serve as effective entry points for another node, the two nodes must lie in proximity in the metric space. This motivates us to select nodes that are close to many other nodes as pivots. We observe that RKNNs naturally capture this notion: unlike KNNs, the exact number of RKNNs of a node is not known in advance. Nodes with more RKNNs are ideal pivot candidates, serving as the nearest neighbors of many others. To be specific, we first expand the neighborhood of each node $v\in V_1$ by searching from $v$ itself for its KNNs $N_k(v)$ (nearest neighbors expansion). With $N_k(v)$ for each node $v$, we construct their RKNNs $R_k(v)$. Let $F_v=R_k(v)\cap U$, we define the gain of selecting $v$ as pivots to be:
$$g(v)=\frac{|F_v|}{\sum_{u\in F_v}\delta(x_v,x_u)}$$
For simplicity, we apply a uniform search effort $ef_L$ to all followers by taking the maximum possible distance of any pivot to its follower (see \textit{Local Sliding} below). The gain $g(v)$ thus simplifies to $|F_v|$, and the IDPS problem reduces to exactly the Minimum Independent Dominating Set (MIDS) problem. 
Since MIDS is inapproximable within $|V|^{1-\epsilon}$ for any $\epsilon > 0$ unless $\mathrm{P=NP}$~\cite{MIDS-HALLDORSSON-1993, MIDS-Irving-1991}, we use a lazy greedy algorithm. Specifically, it maintains a max-heap ordered by $g(v)$ and greedily selects the node with the maximum gain as a pivot, assigning its uncovered RKNN members as followers, until $U = \emptyset$. Since $g(v)$ can only decrease as $U$ shrinks, the stored value in the heap always upper bounds the true gain. We exploit this monotonicity via lazy evaluation: a node popped from the heap is selected immediately if its gain is still current, and reinserted with the updated gain otherwise, avoiding redundant recomputation.

\sstitle{Local Sliding.} Next, with each node classified into pivots and followers, we start efficient merging with local sliding. Specifically, each pivot performs a naive search from the default global entry point of $G_2$, which requires a larger search effort $ef_N$ to compensate for the potentially large distance to the target region. Once the naive search for pivot $p$ completes, all followers in $F_p$ proceed to local sliding using the naive search results $Res_p$ as entry points. Since each follower $v \in F_p$ has pivot $p$ as one of its KNNs in $G_1$, $Res_p$ in $G_2$ are likely to be close to the true neighbors of $v$, as the neighbors of a neighbor are likely neighbors~\cite{NNDecent-dong-WWW-2011}. This yields high-quality entry points for local sliding, allowing a much smaller $ef_L$ to suffice. In practice, we use the same $ef_L$ for all slidings, corresponding to $ef_L=\lambda\cdot\max_{v\in P,u\in F_v}\delta(x_v,x_u)$. Compared to naive search for all nodes with effort $ef_N$, local sliding approximately reduces the total search cost by $$(1 - |P|/|V_1|) \cdot (1 - ef_L/ef_N)$$ Finally, we update the merged graph by applying heuristic edge selection to the search/sliding results. For each source node, its candidates are concatenated with its existing neighbors and pruned together. The target index nodes acquire cross-partition connections solely via reverse edges inserted during this process, without performing explicit searches of their own. Empirically, this asymmetric strategy incurs negligible quality loss while halving the merge cost~\cite{hnswmerger-jin-sigmod-2026,HNSW-malkov-IEEE-2020}. The details of RNSM are presented in Algorithm~\ref{algo:bi-merge}.

\begin{algorithm}[!t]
  \caption{Reverse Neighbor Sliding Merge (RNSM)}
  \label{algo:bi-merge}
  \scriptsize
  %\begin{small}
  \KwIn{Source Graph $G_1=(V_1,E_1)$, Target Graph $G_2=(V_2,E_2)$, Expansion Factor $k^+$, Reverse $k$, Naive Search $ef_N$, Local Sliding $ef_L$}
  \KwOut{Merged Graph $G=(V,E)$}
  $V \gets V_1 \cup V_2$; $E \gets E_1 \cup E_2$\;
  $P \gets \emptyset$\tcp*{Set of Pivots}
  $Covered \gets \emptyset$\tcp*{Set of merged nodes}
  
  \tcc{Phase 1: Pre-computation}
  $G'_1 \gets$ \textbf{ExpandNeighbors}($G_1, k^+$)\;
  $\{R_k(x) \mid \forall x \in V_1\} \gets$ \textbf{BuildReverseIndex}($G'_1, k$)\;
  
  \tcc{Phase 2: Pivot Selection (Lazy Greedy)}               
  Initialize max-heap $H$ with $(|R_k(x)|,\, x)$ for all $x \in V_1$\;
  \While{$H \neq \emptyset$}{                           
    $(g,\, x) \gets \text{uncovered } x \text{ with maximum gain}$\;
    \If{updated gain $g'<g$}{$\mathrm{ReInsert}(g',x)$; \textbf{continue}\;}  
    $F_x\gets R_k(x) \setminus Covered$\;
    $P \gets P \cup \{x\}$; $Covered \gets Covered \cup \{x\} \cup R_k(x)$\;
  }

  \tcc{Phase 3: Parallel Merging}
  \For{$p \in P$ \textbf{in parallel}}{
    \tcp{Pivot: Naive Search}
    $Res_p \gets \mathrm{SearchGraph}(G_2, \text{query}=p, \text{entry}=\text{default}, ef_N)$\;
    $\mathrm{UpdateGraph}(G, p, Res_p)$\;
    
    \tcp{Followers: Local Sliding}
    \For{$y \in F_p$}{
       $Res_y \gets \mathrm{SearchGraph}(G_2, \text{query}=y, \text{entry}=Res_p, ef_L)$\;
       $\mathrm{UpdateGraph}(G, y, Res_y)$\;
    }
  }
  \Return{$G$}
  %\end{small}
\end{algorithm}

\begin{enumerate}[leftmargin=2\labelsep]
    \item \textbf{Neighbor Expansion (Line 4):} We first densify the source graph $G_1$. Existing PGs are sparse approximations of the true KNN graph. We expand the neighborhood of each node $x$ to a larger size $k^+$ by traversing the graph. This creates a richer set of candidate connections, ensuring that nodes can access their true nearest pivots, aligning with the IDPS objective.
    \item \textbf{Reverse Inversion (Line 5):} We construct the RKNN set $R_k(x) = \{y \mid x \in N_k(y)\}$ for every node with a relative small $k$ (\ie 3,5,7).
    \item \textbf{Pivot Selection (Lines 6-12):} We maintain a max-heap ordered by marginal gain $|R_k(x)\setminus Covered|$. At each step, we extract the top candidate; if its stored gain is stale, we reinsert it with the updated gain (lazy evaluation). Once the gain is confirmed current, we select $x$ as a pivot and assign all uncovered $R_k(x)$ to slide from $x$.
    \item \textbf{Parallel Execution (Lines 13-18):} Finally, we execute the merge in parallel. Pivots perform a \textit{Naive Search} on $G_2$. Their search results become the entry points for the \textit{Local Sliding} of their followers.
\end{enumerate}

It is worth noting that both \textbf{Reverse Inversion} and \textbf{Pivot Selection} account for a negligible fraction of the total merge time (typically less than 5\%), as the merge process is compute-intensive and dominated by distance computations during \textbf{Neighbor Expansion} ($\approx$ 10\%) and \textbf{Parallel Execution} ($\approx$ 85\%). A detailed runtime breakdown of each RNSM phase is provided in Appendix Exp-a.

\begin{exmp}
    As illustrated in Fig~\ref{fig:RNSM}, RNSM merges source index $G_1$ into target index $G_2$ in four steps. Nearest neighbor expansion in Fig~\ref{fig:RNSM}(a) first obtains AKNNs for each node in $G_1$, which is a cheap beam search starting from the nodes themselves. Next, in Fig~\ref{fig:RNSM}(b), given the AKNNs, we get the RKNNs of each node in $G_1$ and rank them by the number of RKNNs for pivot selection. Then, nodes 5 and 8 are greedily selected as pivots for their higher number of RKNNs, while node 4 is excluded since it has already been assigned as a follower of pivot 5. In Fig~\ref{fig:RNSM}(c), we finally start sliding merge on the target index with node 8 and its followers \{3, 9, 10\} as an example. Node 8 first performs a naive search from a distant initial entry point $b$ in $G_2$ and retrieves $i$ as an ANN. Then, its followers use $i$ for local sliding. Because they are close to pivot 8, their nearest neighbors \{g, f, j\} lie near $i$ and are reached with fewer hops.
\end{exmp}

\begin{figure*}[!t]
    \centering
    \includegraphics[width=0.9\textwidth]{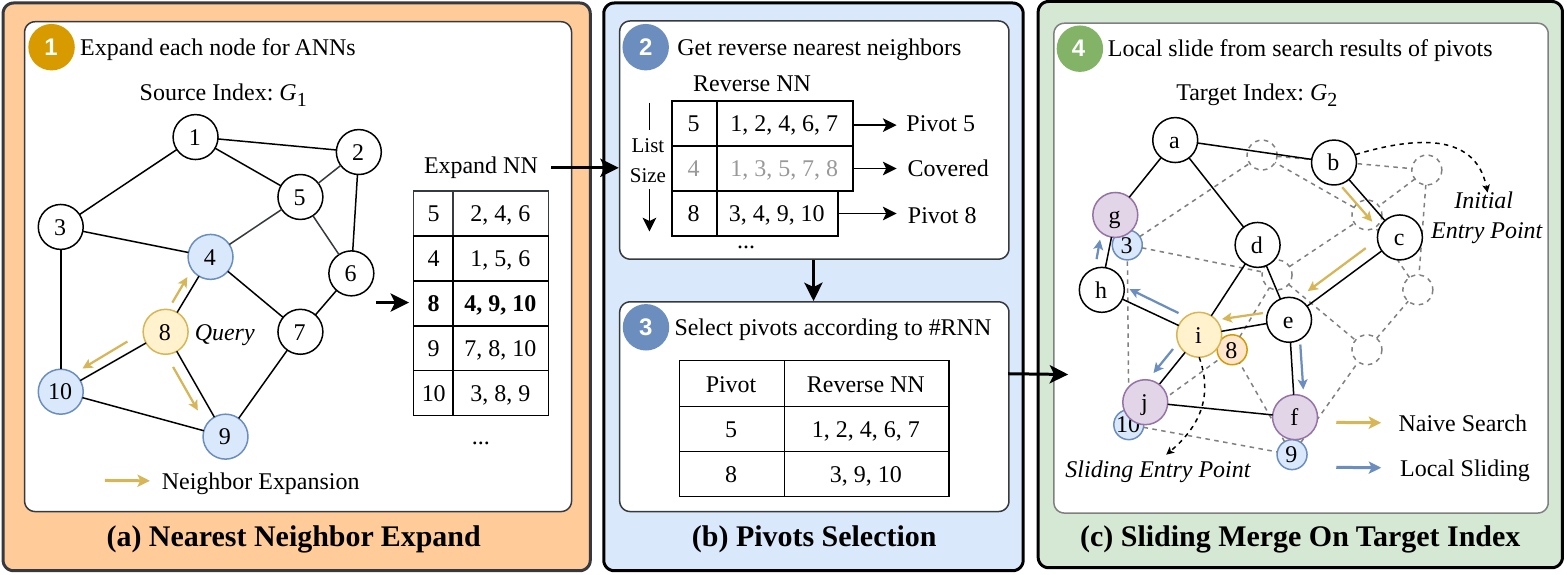}
    \vspace{-3mm}
        \caption{Example of Two-Index Merging with RNSM}
        \vspace{-3mm}
  \label{fig:RNSM}
\end{figure*}

\subsection{Performance Analysis}\label{sec:Analysis}

\stitle{Advantages of RNSM.}
Using RKNNs yields several structural benefits:
(1) \textbf{Natural Hub Identification:} The metric $|R_k(x)|$ identifies density peaks (hubs) purely from graph topology without requiring expensive clustering or centroid computations.
(2) \textbf{Dynamic Load Balancing:} Unlike fixed $k$ approaches, RNSM allows variable cluster sizes. Dense regions generate pivots with large $|R_k(x)|$, amortizing the naive search cost over more nodes, while sparse regions degenerate gracefully.
(3) \textbf{Cost Reduction:} By selecting pivots with the most extensive coverage first, we greedily reduce the number of required pivots $|P|$, directly reducing the naive search term in our cost model. The independent-set constraint further boosts parallel-friendly execution.
% (4) \textbf{Proximity Guarantee:} A follower $y$ is only assigned to pivot $x$ if $x \in N_k(y)$. This guarantees that the starting point of the sliding search is strictly within the $k$-nearest neighbors of the query, ensuring a low $\delta(x,y)$ and thus a minimal sliding cost.

\stitle{Complexity Analysis.}
We analyze the time complexity of merging a source graph index~\(G_1\) into a target graph index~\(G_2\) to produce a unified index~\(G = (V, E)\), where \(n = |V|\) denotes the total number of vertices. Without loss of generality, we assume the source index is not larger than the target index (\ie \(|G_1| \leq |G_2|\)). If initially \(|G_1| > |G_2|\), the two are swapped so that we always merge the smaller structure into the larger one—a common optimization that speeds up merging.
A key property of these graph-based indices, established in prior work~\cite{NSG-fu-VLDB-2019,WorstCase-indyk-nips-2023}, is that the maximum node degree is bounded by a constant. This sparsity guarantee is crucial for achieving an efficient merge procedure.
The merging process consists of three sequential stages:

\noindent\textbf{1. Neighbor Expansion:}
Expanding neighbors involves an incremental search with the graph index. Start with the nearest neighbor; the incremental cost of searching the top-$k^+$ neighbors can be expressed as $O(k^+)$ with a small constant factor~\cite{ESG-yang-arxiv-2025}. This efficiency stems from interpreting the existing proximity graph (PG) index as a compressed representation of the $k$-nearest neighbor graph~\cite{NSG-fu-VLDB-2019}.
\textit{Total Complexity:} $O(n  \cdot k^+)$.

\noindent\textbf{2. Pivot Selection:}
Constructing the RKNN graph requires iterating over all edges in the expanded graph. Given that the total number of edges is \( n \cdot k^+ \), the inversion step takes \( O(n \cdot k^+) \) time. Subsequently, sorting the nodes on the RKNN graph using bucket sort can be performed in \( O(|V| + |E|) \) time. Since \( |V| = n \) and \( |E| = n \cdot k < n \cdot k^+ \), the \textit{total complexity} of the entire process remains \( O(n \cdot k^+) \).

\noindent\textbf{3. Sliding Merge:}
Let $|P|$ be the number of selected pivots.
\begin{itemize}[leftmargin=10pt]
    \item \textit{Pivots:} Perform naive AKNN search. In typical graph structures (NSG, HNSW), the search complexity is bounded by $c_{naive} = O(\log n \cdot ef \cdot d)$, where $ef$ is the search parameter and $d$ is the average degree. Total: $O(|P| \cdot c_{naive})$.
    \item \textit{Followers:} Perform local sliding. Since the pivot $p \in N_k(y)$ by construction, the sliding path length is proportional to the distance $\delta(p, y)$ and is small. In practice, the search complexity is bounded by a small constant $c_{slide}$. Total: $O((n - |P|) \cdot c_{slide})$.
\end{itemize}

\noindent\textbf{Overall Complexity:}
Both local sliding and naive search have a higher cost than neighbor expansion, where $ c_{naive} > c_{slide}>k^+$. The overall complexity can be derived to $O(|P|c_{naive} +(n-|P|)c_{slide})$. Since $|P| < n$ (depending on $k$, typically around $20\%$ of nodes), the dominant term is reduced from $n \cdot c_{naive}$ (all-naive search) to roughly $n \cdot c_{slide}$, representing a significant speedup.

%\section{Merge Order Selection for Multiple Indexes}
\section{MOS for Multiple Indexes Merge}
\label{sec:MOS-technical}
In this section, we minimize the multi-index merge cost by reducing redundant operations. We first formalize the merge order selection task. Then, we propose algorithms for our formalized problem.

\stitle{Merge Order Selection.}\label{sec:MOS-problem}
Recall that merging requires finding the KNNs in all other indexes for every node in each source index. Naive approach requires $O(m^2)$ index merge operations between $m$ indexes.
However, this approach contains significant redundancy in practice. Besides, our empirical analysis (Fig~\ref{fig:observation2}) shows that different merge order strategies yield substantially different search performance. Thus, we formally define the task as follows:
Given $m$ dataset partitions $\{P_1, P_2, \ldots, P_m\}$, we represent partitions as vertices $V = \{v_1, \ldots, v_m\}$ and merge operations as edges $E$, forming a merge order graph $\mathbb{G} = (V, E)$. Each edge $(v_i, v_j) \in E$ indicates that we perform index merging between index $G_i$ and $G_j$. The merge order selection (MOS) task aims to adjust the merge order graph $\mathbb{G}$ to achieve better individual sub-index connections at low cost.

\stitle{Cost Model}. Given the notation above, our objective is to extract a subgraph from the complete graph of pairwise merges. We aim to identify a merge order—defined by this subgraph—that achieves search performance comparable to either a full index rebuild or an exhaustive pairwise merge strategy.  
%For clarity, we analyze random and clustered partitions separately. For random partitions, we assume that merging any two partitions incurs a unit cost. 
We first prioritize merges using centroid distance, since closer partitions are more likely to share KNNs~\cite{douze2026FaissLibrary}. For random partitions, the exchangeability of partition labels implies identical expected centroid distances for all partition pairs. Therefore, we assign a unit merge cost to every pair.
% \begin{theorem}\label{thm:random-partition-centroid-distance}
% Let $X=\{x_1,\ldots,x_n\}\subset\mathbb{R}^D$ be uniformly split into $m$ disjoint equal-size random partitions $\{P_1,\ldots,P_m\}$, and let $c_i=\frac{1}{|P_i|}\sum_{x\in P_i}x$ be the centroid of $P_i$. For any two partition pairs $(i,j)$ and $(i',j')$ with $i\neq j$ and $i'\neq j'$, we have
% \[
% \mathbb{E}\|c_i-c_j\|=\mathbb{E}\|c_{i'}-c_{j'}\|.
% \]
% \end{theorem}
% \noindent\textit{Proof sketch.}
% Equal-size random assignment makes the partition labels exchangeable. Hence, for any two valid pairs $(i,j)$ and $(i',j')$, the random variables $\|c_i-c_j\|$ and $\|c_{i'}-c_{j'}\|$ are identically distributed and therefore have the same expectation. A covariance expansion further gives $\mathbb{E}\|c_i-c_j\|^2=\frac{2m}{m-1}\mathrm{tr}(\Sigma)$, where $\Sigma=\mathrm{Cov}(c_i,c_i)$, which is independent of the pair $(i,j)$.
For clustered partitions, we directly compute the centroid distances for all the partitions.
Moreover, \S~\ref{sec:analysis-motivation} indicates that different merge order strategies can lead to substantially different search performance. 
MOS therefore does not rely solely on a proximity-based centroid distance proxy. Instead, it combines it with two additional proxies: degree and graph diameter. Degree $R$ is introduced for merge efficiency, since it controls the number of pairwise merges. Experimentally, a high-degree partition may receive an overwhelming number of new out-neighbor candidates and hurt quality, as shown in \textit{Star} of Fig~\ref{fig:observation2}. Diameter $\Delta$ is introduced for global search quality, given that a large diameter may reduce recall, as shown in \textit{Path} of Fig~\ref{fig:observation2}. This is because the diameter indicates the maximum number of intermediate partitions needed to reach a cross-partition nearest neighbor. We now formally define the MOS problem with these proxies:

\stitle{MOS Problem}. Therefore, we propose extracting subgraphs from multiple perspectives with minimal cost, specifically:
\begin{definition}[Minimum Cost Merge Order Selection (MOS)]
\
\begin{itemize}[leftmargin=3\labelsep]
    \item \textbf{Input:} Given a set of $m$ vertices $V = \{v_1, v_2, \ldots, v_m\}$ and a symmetric cost matrix $C$ where $C_{ij}$ is the cost of connecting vertex $v_i$ and $v_j$, diameter bound $\Delta$ and degree bound $R$.
    \item \textbf{Output:} An undirected graph $\mathbb{G} = (V, E)$ %where $E \subseteq \binom{V}{2}$ 
    such that:
    \item \textbf{Constraint:} (1) Connectivity: The graph $\mathbb{G}$ is connected. (2) Maximum Degree: For all vertices $v \in V$, the degree $\deg(v) \leq R$. (3) Graph Diameter: The diameter of $\mathbb{G}$ is at most $\Delta$.
    \item \textbf{Cost Objective:} Minimize the total cost of the graph:
    \begin{align*}\min \sum_{(v_i,v_j) \in E} C_{ij}\end{align*}
\end{itemize}
\end{definition}

The constraints serve distinct objectives:
\begin{enumerate}[leftmargin=3\labelsep]
\item \textbf{Connectivity:} This is the fundamental completeness requirement, ensuring that the search space is fully traversable.
\item \textbf{Degree Bound ($R$):} This constraint is merge efficiency and parallel-friendly design. Inserting new edges locks a large proportion of proximity graph nodes, so excessive degree serializes updates and causes parallel workload imbalance.
\item \textbf{Diameter Bound ($\Delta$):} This constraint governs the worst-case search latency. The diameter determines the maximum number of intermediate index hops required to traverse between any two partitions.
\end{enumerate}
While minimizing the total edge weight reduces the direct merge costs, the graph edge weights also affect the runtime performance as shown in our empirical analysis in \S~\ref{sec:Analysis}. 
Ignoring the diameter bound reduces the problem to the NP-hard Bounded Degree Minimum Spanning Tree Problem~\cite{21NP-karp-1972}.
To address this NP-hard problem efficiently, we propose a greedy heuristic as follows. We initialize the merge order graph as an empty structure and iteratively insert nodes to satisfy the Connectivity constraint. Subsequently, for each node, we establish connections with nodes located beyond $\Delta$-hop distance (typically $\Delta$=2) to ensure compliance with the Diameter Bound. When either node possesses $R$ neighbors, we relay the connection to its $(\Delta-1)$-hop neighbors to meet both the Degree and Diameter Bounds whenever a degree-feasible relay exists. Throughout this process, connections are prioritized according to the cost matrix, with lower-cost pairs being added first. The detailed MOS algorithm is presented below.

\begin{figure}[!t]
    \centering
    \includegraphics[width=1.0\columnwidth]{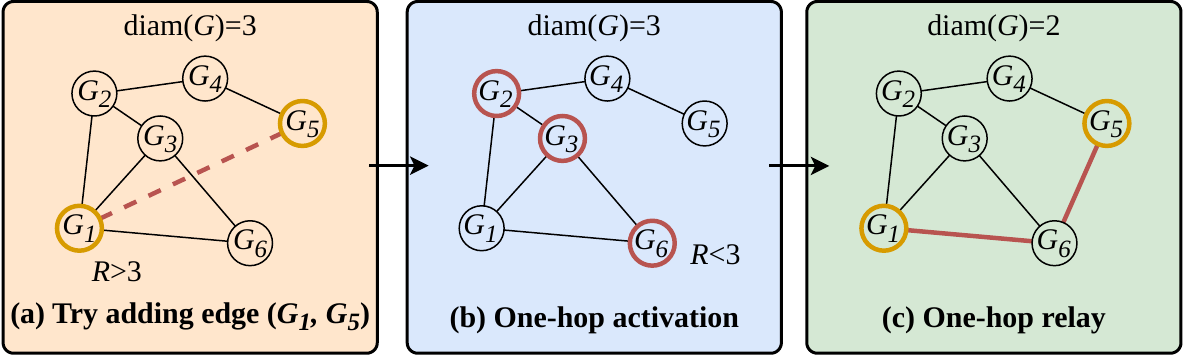}
    \vspace{-6mm}
        \caption{Example of Merge Order Selection}
        \vspace{-3mm}
  \label{fig:MOS}
\end{figure}

\stitle{MOS Algorithm.}
Algorithm~\ref{algo:mos} introduces the detailed procedure of MOS. We begin with an empty graph $\mathbb{G}$ and iteratively add edges (lines 2-3). We process vertices $v \in V$ in ascending order of their average distance to the $R$ nearest centroids, prioritizing centrally-located partitions (lines 4-6). For each $v$, we iteratively connect it to its nearest neighbor $x$ that is more than $\Delta$ hops away (line 8). If both $v$ and $x$ have degree less than $R$, we directly add edge $(v,x)$ (lines 9-10). If exactly one vertex has reached the degree limit $R$, we employ a 1-hop relay strategy (lines 11-15), select a neighbor $x'$ of the degree-saturated vertex, and connect it to the unsaturated one. Note that the subgraph may not exist due to the Moore bound~\cite{MooreGraph-Hoffman-IBM}; however, we still enforce connectivity in practice for completeness of results.
\begin{exmp}
    As illustrated in Fig~\ref{fig:MOS}, MOS constructs a workload-balanced merge order graph with diameter 2 by inserting and relaying edges within 1-hop neighborhoods. Suppose we insert edge $(G_1, G_5)$ when the degree limit of 3 is exceeded for $G_1$ (i.e., $\deg(G_1)=4$). We then relay this connection to one of the 1-hop neighbors of $G_1$, and $G_6$ is selected since $\deg(G_6)<3$. By relaying edge $(G_1, G_5)$ as $(G_5, G_6)$, we obtain a graph with maximum degree at most 3 and diameter 2.
\end{exmp}

Given the merge-order graph constructed by MOS, we execute each selected edge as a pairwise RNSM operation according to two commonly adopted principles. First, we prioritize edges between closer partitions because nearby partitions are more likely to share nearest neighbors~\cite{douze2026FaissLibrary}. Second, for each selected pair, we merge the smaller sub-index into the larger one~\cite{Milvus-wang-ACM-2021,hnswmerger-jin-sigmod-2026}. Consequently, only vectors in the smaller source index require Naive Search or Local Sliding to obtain cross-partition edges, further reducing the merge time. During each pairwise operation, newly identified cross-partition candidates are concatenated with existing neighbors and jointly processed by the edge-pruning rule of the PG index to meet the maximum outdegree bound. Accordingly, when a sub-index is incident to multiple selected edges, the centroid-distance priority determines the next pair to process, while the partition sizes determine the merge direction.

\stitle{Complexity.}
We now analyze the time complexity of Algorithm~\ref{algo:mos}. Constructing the cost matrix requires pairwise distance computations between $m$ $D$-dimensional centroids, resulting in $O(m^2\cdot D)$. The ordering requires $O(m^2 \log R)$ for computing average distances of $R$ nearest neighbors and sorting vertices. The main loop adds $O(m\cdot R)$ edges total, with each edge addition requiring $O(m\cdot R)$ time for BFS-based incremental hop-distance matrix updates, yielding a complexity of $O(m^2R^2)$.
Overall, the time complexity of our MOS algorithm is $O(m^2(R^2+D))$, independent of dataset size $n$.

Since $m$, $R$, and $D$ are typically much smaller than the dataset size $n$, MOS is substantially less expensive than rebuilding or merging graph indexes over all $n$ vectors. Across our evaluated configurations, the end-to-end time cost of the MOS accounts for only approximately $0.1\%$--$1\%$ of the corresponding Rebuild time.

\begin{algorithm}[!t]
  \caption{Merge Order Selection MOS}
  \label{algo:mos}
  \scriptsize
  %\begin{small}
  \KwIn{Partition Centroids $\{c_1, c_2,\ldots, c_m\}$, Cost Matrix $C=\mathbf{0}$, Degree Bound $R$, Diameter Bound $\Delta=2$}
  \KwOut{Merge order graph $\mathbb{G}=\left(V,E\right)$}
  $C \gets \text{compute cost matrix}$\;
  $V \gets \{c_1, c_2,\ldots, c_m\}$;  $E \gets \emptyset$\;
  $H \gets \text{distance matrix in hops, initialized to } \infty$; $Distances\gets\emptyset$\;
  \For{$c_i\in V$}{
    $Distances[c_i]\gets \text{average distance of }R\text{ nearest centroids}$\;
  }
  $\text{order}\gets V\text{ sorted in ascending order of } Distances$\;
  \For{$v\in \text{order}$}{
    \While{$\text{has nearest vertex } x \text{ with } H[v][x]>\Delta$}{
        \If{$\text{deg}(v)<R \text{ \textbf{and} } \text{deg}(x)<R$}{
            $E\gets E\cup(v,x)$\;
        }
        \ElseIf{$\text{deg}(v)<R \text{ \textbf{and} } \text{deg}(x)\ge R$}{
            $x' \gets \mathop{\arg\min}_{x' \in \text{Neighbor}(x), \deg(x') < R} C_{x', v}$;
            $E\gets E\cup(v,x')$\;
        }
        \ElseIf{$\text{deg}(v)\ge R \text{ \textbf{and} } \text{deg}(x)<R$}{
            $v' \gets \mathop{\arg\min}_{v' \in \text{Neighbor}(v), \deg(v') < R} C_{v', x}$; 
            $E\gets E\cup(v',x)$\;
        }
        \Else{
            $e\gets(v,x)$ \tcp*{feasibility fallback}
        }
        $\text{Update}\ H$\;
    }
  }
  \Return{$G$}
  %\end{small}
\end{algorithm}

\section{Experiment}\label{sec:Exp}
% In this section, we present a comprehensive experimental evaluation of our proposed framework. We begin by detailing the experimental setup, followed by an analysis of the merge efficiency and search performance using HNSW across varying partition counts ($\{2, 4, 6, 8, 10\}$). We benchmark our multi-index merge framework against the overlapping strategy utilized in DiskANN. Furthermore, we demonstrate scalability on the DEEP100M dataset with up to 50 partitions. To establish generality, we extend our evaluation to three additional graph-based indices: NSG, SSG, and $\tau$-MNG. Finally, we conduct ablation studies to quantify the impact of the individual optimization techniques employed.

\begin{table}[!t] 
\centering 
\caption{Dataset Statistics} 
\label{tab:dataset_details} 
\vspace*{-4mm}
%\begin{small}
\resizebox{0.8\columnwidth}{!}{%
\begin{tabular}{c|c c c c} 
\toprule
\textbf{Dataset}& \textbf{Dimension} & \textbf{Base Size} & \textbf{Query Size} & \textbf{LID} \cite{scikit-lid-2021} \\ 
\midrule
SIFT & 128 & 1,000,000 & 10,000 & 19.4\\  
DEEP & 96 & 100,000,000 & 1,000 & 19.5\\
GIST & 960 & 1,000,000 & 1,000 & 38.9 \\
IMAGE & 768 & 10,000,000 & 10,000 & 21.2\\
ANTON & 1024 & 10,000,000 & 1,000 & 28.9\\
MARCO & 1024 & 10,000,000 & 1,000 & 31.0\\
\bottomrule
\end{tabular} 
}
\vspace*{-6mm}
%\end{small}
\end{table}

\subsection{Experiment Setting.}

\stitle{Datasets.}We use six public datasets with diverse scales and types. Including datasets that are widely used in the benchmarks (SIFT, DEEP, GIST)\footnote{https://www.cse.cuhk.edu.hk/systems/hash/gqr/datasets.html} and datasets generated from the latest embedding models (IMAGE\footnote{https://huggingface.co/datasets/kinianlo/imagenet\_embeddings}, ANTON\footnote{https://huggingface.co/datasets/anton-l/wiki-embed-mxbai-embed-large-v1}, MARCO\footnote{https://huggingface.co/datasets/Cohere/msmarco-v2.1-embed-english-v3}). The statistics of those datasets are summarized in Table~\ref{tab:dataset_details}, where \emph{Base Size} denotes the dataset cardinality and \emph{Query Size} denotes the number of queries for testing. We also take several slices of distinct sizes from the original datasets to test the scalability of our framework. For example, we use DEEP1M to represent 1 million slices of the DEEP dataset, and use ANTON10M to indicate the 10 million scale in the experiment.

\input{figures/bi-index-merge-QPS-appendix}

\stitle{Methods.}  
We compare the following index merge methods:

\noindent
$\bullet$ \textit{Rebuild}: Rebuild the complete proximity graph directly by incrementally inserting nodes into an empty graph.

\noindent
$\bullet$ $\CM$: Cross-query Merge (CM) using naive search. It uses a smaller candidate set $ef \ll ef_{construction}$ and a one-way search (only searches the target index, not bidirectional), which are the two key optimizations used in HNSW-merger~\cite{hnswmerger-jin-sigmod-2026}.

\noindent
$\bullet$ $\FGIM$: Method for merging PG indexes and transforming it into a KNN graph merging and refining problem~\cite{FGIM_wu_2026_SIGMOD}.

\noindent
$\bullet$ $\DISK$: Overlapping-based method that merges indexes through duplicate data points~\cite{DiskANN-subramanya-NIPS-2019}.

\noindent
$\bullet$ $\RNSM$: Our Reverse Neighbor Sliding Merge (RNSM) method without Merge Order Selection (MOS).

\noindent
$\bullet$ $\RNSM^*$/$\RNSM^+$: Our RNSM with MOS for random/cluster partitions.

\noindent
$\bullet$ $\CM^*$/$\CM^+$: CM with pairwise order for random/cluster partitions.

\stitle{Performance Metrics.} We evaluate the merged graph index based on the following two primary metrics: %search performance and index construction (merge) efficiency.
\begin{itemize}[leftmargin=2\labelsep]
    \item \textbf{Search Performance:} We measure efficiency via Queries Per Second (QPS) and accuracy via Recall@K as defined in \S~\ref{sec:preliminary}. Consistent with standard practice, we report the average Recall@10 and QPS with the top-10 AKNN search over the query set.
    \item \textbf{Merge Efficiency:} We quantify efficiency using the speed-up ratio relative to the $\CM$, $\FGIM$, \textit{Rebuild} baselines, and the overlapping construction strategy employed by DiskANN.
\end{itemize}

\stitle{Parameter Setting}. We integrate our methods with various graph-based indexing algorithms, each with its own parameters. We leave the other parameters as default~\cite{GraphBasedANNS-wang-vldb-2021} and only consider the beam search effort $ef$. In particular, we examine several widely adopted graph-based AKNN search algorithms: HNSW~\cite{HNSW-malkov-IEEE-2020}, NSG~\cite{NSG-fu-VLDB-2019}, SSG~\cite{SSG-fu-IEEE-2022}, and $\tau$-MNG~\cite{Efficient-peng-SIGMOD-2023}. 
% To determine the parameters $ef_N$ in Algorithm~\ref{algo:bi-merge}, we set it as the default candidate size that achieves an average Recall@10>0.98 on separated indexes based on observation of Fig~\ref{fig:observation1}.
We employ grid search to determine the optimal parameter values of $k^+$, $k$, and $ef_L$ for Algorithm~\ref{algo:bi-merge}. For fair comparison, we set the same beam search parameter $ef$ and maximum PG out-degree for $\CM$ and $\RNSM$. The parameters of $\FGIM$ are set according to its default settings~\cite{FGIM_wu_2026_SIGMOD}.

\stitle{Implementation Environment.} All experiments were conducted on a server running Ubuntu 22.04 LTS, equipped with dual Intel Xeon Platinum 8370C CPUs (2.80GHz, 32 physical cores each, supporting AVX-512) and 512 GB of DRAM. The proposed framework and baselines were implemented in C++ based on the \texttt{hnswlib} library~\cite{hnswlib} and compiled using GCC 13.1.0 with \texttt{-O3} optimization.

\begin{figure}[!t]
    \centering
    \begin{minipage}[t]{0.48\columnwidth}
        \centering
        %\begin{figure*}[!htbp]
%\centering
%\begin{small}

\begin{tikzpicture}
    \begin{customlegend}[legend columns=2,
        legend entries={CM$^*$, CM$^+$, FGIM, Rebuild},
        legend style={at={(0.5,1.15)},anchor=north,draw=none,font=\scriptsize,column sep=0.2cm},legend image post style={scale=0.8}]
    \addlegendimage{area legend,fill=amaranth!50,draw=amaranth,pattern=north west lines,pattern color=amaranth}
    \addlegendimage{area legend,fill=amber!50,draw=amber,pattern=north east lines,pattern color=amber}
    \addlegendimage{area legend,fill=navy!50,draw=navy,pattern=grid,pattern color=navy}
    \addlegendimage{area legend,fill=orange!50,draw=orange,pattern=crosshatch,pattern color=orange}
    \end{customlegend}
\end{tikzpicture}
\\[-\lineskip]

\begin{tikzpicture}[scale=1]
\begin{axis}[
    height=3.0cm,
    width=1.1\linewidth,
    ybar=0.5pt, ymode=log, log basis y={2}, ymin=1,
    bar width=3.5pt,
    ylabel=Speed-up Ratio,ylabel style={yshift=-20pt}, 
    ytick={0,1,2,4},
    enlarge x limits=0.20,
    xtick={0,1,2,3},
    xticklabels={4 parts, 6 parts, 8 parts, 10 parts},
    xticklabel style={rotate=45, anchor=east, font=\scriptsize},
    legend style={font=\scriptsize,draw=none},
    label style={font=\scriptsize},
    tick label style={font=\scriptsize},
    ymajorgrids=true,
    grid style=dashed,
]
% NMr
\addplot[ybar,fill=amaranth!50,draw=amaranth,pattern=north west lines,pattern color=amaranth] coordinates {
    (0,1.63) (1,1.74) (2,2.31) (3,2.06)
};
% NMc
\addplot[ybar,fill=amber!50,draw=amber,pattern=north east lines,pattern color=amber] coordinates {
   (0,1.75) (1,1.90) (2,1.85) (3,1.62)
};
% FGIM
\addplot[ybar,fill=navy!50,draw=navy,pattern=grid,pattern color=navy] coordinates {
    (0,1.38) (1,1.03) (2,1.26) (3,1.09)
};
% BuildAsOne
\addplot[ybar,fill=orange!50,draw=orange,pattern=crosshatch,pattern color=orange] coordinates {
    (0,4.26) (1,2.59) (2,2.96) (3,1.91)
};
\end{axis}
\end{tikzpicture}
\vgap
\vspace*{-2mm}
\caption{Speed-up Ratio Comparison across Different Partitions with MARCO10M}
\vgap
\label{fig:speedup-comparison}
%\end{small}
%\end{figure*}
    \end{minipage}
    \hfill
    \begin{minipage}[t]{0.48\columnwidth}
        \centering
        %\begin{figure*}[!htbp]
%\centering
%\begin{small}

\begin{tikzpicture}
    \begin{customlegend}[legend columns=3,
        legend entries={CM$^*$, CM$^+$, Rebuild},
        legend style={at={(0.5,1.15)},anchor=north,draw=none,font=\scriptsize,column sep=0.2cm},legend image post style={scale=0.8}]
    \addlegendimage{area legend,fill=amaranth!50,draw=amaranth,pattern=north west lines,pattern color=amaranth}
    \addlegendimage{area legend,fill=amber!50,draw=amber,pattern=north east lines,pattern color=amber}
    \addlegendimage{area legend,fill=orange!50,draw=orange,pattern=crosshatch,pattern color=orange}
    \end{customlegend}
\end{tikzpicture}
\\[-\lineskip]

% 第一个柱状图
\begin{tikzpicture}[scale=1]
\begin{axis}[
    height=3.0cm,
    width=1.1\linewidth,
    ybar=0.5pt, ymode=log, log basis y={2}, ymin=1,
    bar width=3.5pt,
    ylabel=Speed-up Ratio,ylabel style={yshift=-20pt}, 
    xmin=-0.5, xmax=3.5,
    enlarge x limits=0.20,
    xtick={0,1,2,3},
    xticklabels={20 parts, 30 parts, 40 parts, 50 parts},
    xticklabel style={rotate=45, anchor=east, font=\scriptsize},
    legend style={font=\scriptsize,draw=none},
    label style={font=\scriptsize},
    tick label style={font=\scriptsize},
    ymajorgrids=true,
    grid style=dashed,
]
% CM$^*$ (random)
\addplot[ybar,fill=amaranth!50,draw=amaranth,pattern=north west lines,pattern color=amaranth] coordinates {
    (0,1.35) (1,1.40) (2,1.47) (3,1.58)
};
% CM$^+$ (cluster)
\addplot[ybar,fill=amber!50,draw=amber,pattern=north east lines,pattern color=amber] coordinates {
   (0,1.45) (1,1.65) (2,2.09) (3,2.31)
};
% Rebuild (20--40 parts placeholders; 50 parts measured)
\addplot[ybar,fill=orange!50,draw=orange,pattern=crosshatch,pattern color=orange] coordinates {
   (0,4.90) (1,2.69) (2,1.88) (3,1.35)
};
\end{axis}
\end{tikzpicture}

\vgap
\vspace*{-2mm}
\caption{Scalability Study of Merge Efficiency on DEEP100M}
\vgap
\label{fig:scale-speedup-comparison}
%\end{small}
%\end{figure*}
    \end{minipage}
\end{figure}

% Required packages: \usepackage{booktabs,multirow}
\begin{table}[!t]
  \centering
  \caption{ Graph Structure Analysis with Average Out-Degree (AD) and Out-Degree Percentiles $D_{5/50/99}$.}
  \label{tab:final-graph-structure}
  \begingroup
  \scriptsize
  \setlength{\tabcolsep}{1.6pt}
  \renewcommand{\arraystretch}{1.08}
  \begin{tabular*}{0.9\columnwidth}{@{\extracolsep{\fill}}c|c|*{6}{c}@{}}
    \toprule
    \textbf{Method} & \textbf{Metric}
      & \textbf{SIFT} & \textbf{DEEP} & \textbf{GIST}
      & \textbf{IMAGE} & \textbf{ANTON} & \textbf{MARCO} \\
    \midrule
    \multirow{3}{*}{\textbf{Rebuild}}
      % & \textbf{GQ} & 0.393 & 0.380 & 0.162 & 0.299 & 0.334 & 0.249 \\
      & \textbf{AD} & 27.30 & 31.43 & 15.83 & 23.83 & 28.36 & 25.29 \\
      & \textbf{$D_{5/50/99}$} & 11/26/50 & 12/30/59 & 2/11/58 & 7/21/58 & 10/26/59 & 9/22/59 \\
      % & \textbf{CC} & 1 & 1 & 1 & 1 & 1 & 1 \\
    \midrule
    \multirow{3}{*}{\textbf{FGIM}}
      % & \textbf{GQ} & 0.362 & 0.354 & 0.152 & 0.265 & 0.306 & 0.210 \\
      & \textbf{AD} & 20.48 & 22.60 & 12.05 & 15.24 & 19.08 & 14.55 \\
      & \textbf{$D_{5/50/99}$} & 9/19/44 & 10/22/46 & 2/9/56 & 6/14/35 & 8/18/46 & 7/14/32 \\
      % & \textbf{CC} & 1 & 1 & 1 & 1 & 1 & 1 \\
    \midrule
    \multirow{3}{*}{\textbf{CM}}
      % & \textbf{GQ} & 0.356 & 0.343 & 0.145 & 0.266 & 0.303 & 0.224 \\
      & \textbf{AD} & 24.33 & 28.29 & 14.70 & 21.41 & 25.94 & 23.79 \\
      & \textbf{$D_{5/50/99}$} & 10/23/49 & 11/27/59 & 2/11/58 & 7/18/57 & 9/24/58 & 8/21/58 \\
      % & \textbf{CC} & 1 & 1 & 1 & 1 & 1 & 1 \\
    \midrule
    \multirow{3}{*}{\textbf{RNSM}}
      % & \textbf{GQ} & 0.343 & 0.339 & 0.205 & 0.312 & 0.301 & 0.267 \\
      & \textbf{AD} & 23.56 & 28.48 & 19.46 & 23.94 & 25.59 & 25.14 \\
      & \textbf{$D_{5/50/99}$} & 9/22/49 & 13/26/59 & 3/17/60 & 10/22/57 & 10/23/58 & 10/23/58 \\
      % & \textbf{CC} & 1 & 1 & 1 & 1 & 1 & 1 \\
    \bottomrule
  \end{tabular*}
  \endgroup
\end{table}

% \subsection{Efficiency Evaluation}
\subsection{Experiment Results}
\stitle{Exp-1: Test of Indexes Merging Performance}. We begin by evaluating the merging efficiency and subsequent search performance of our $\RNSM$ framework. 
For the setup, we randomly partition each dataset in Table~\ref{tab:dataset_details} into two disjoint subsets of 0.5 million vectors and construct indexes using \texttt{hnswlib}. 
We compare $\RNSM$ against $\CM$, $\FGIM$, and a reconstruct-from-scratch baseline \textit{Rebuild}.

Table.~\ref{tab:index-size-memory-time} reports the merge memory overhead (MO) and merge time (MT) across six datasets. $\FGIM$ has the highest MO, requiring $2.72\times$--$4.07\times$ the memory of $\RNSM$ because its KNNG transformation and refinement materialize a dense KNNG while simultaneously retaining the sub-indexes and the merged index, thereby duplicating vector data. We reduce this redundancy in both $\CM$ and $\RNSM$ by separating merged graph topology from the final searchable PG: input sub-indexes provide the vectors, the output stores only the new topology during merging, and vectors are attached after sub-indexes are released. $\CM$ has the lowest MO, while $\RNSM$ increases it by only $2.5\%$--$11.1\%$ for RKNN metadata. This is reasonable given that these metadata require only $O(nk)$ space ($k\le 20$), substantially smaller than the $O(nD)$ space required by the dataset vectors. Merge on DEEP100M has a similar result: $\RNSM$ uses less MO than \textit{Rebuild} and only $14.0\%$ more than $\CM$.
For MT, $\RNSM$ delivers speed-up ratios up to 1.74$\times$ over $\CM$, 3.86$\times$ over $\FGIM$, and 9.92$\times$ over \textit{Rebuild}. 
Notably, $\CM$ even outperforms the $\FGIM$. According to Theorem 4.1 in FGIM~\cite{FGIM_wu_2026_SIGMOD}, when the number of partitions is 2, the cross-index search in $\FGIM$ reduces to a naive search of $\CM$. However, $\FGIM$ incurs additional overhead from refining KNN graphs and transforming the KNN graph into PG indexes.

% Required packages: \usepackage{booktabs,multirow}
\begin{table}[!t]
  \centering
  \caption{Index statistics across datasets. MO denotes merge memory overhead (MB), and MT denotes merge time (s).}
  \label{tab:index-size-memory-time}
  \begingroup
  \scriptsize
  \setlength{\tabcolsep}{1.6pt}
  \renewcommand{\arraystretch}{1.08}
  \begin{tabular*}{0.9\columnwidth}{@{\extracolsep{\fill}}c|c|*{6}{c}@{}}
    \toprule
    \textbf{Method} & \textbf{Metric}
      & \textbf{SIFT} & \textbf{DEEP} & \textbf{GIST}
      & \textbf{IMAGE} & \textbf{ANTON} & \textbf{MARCO} \\
    \midrule
    \multirow{3}{*}{\textbf{Rebuild}}
      % & \textbf{IS} & 699 & 615 & 3911 & 3178 & 4155 & 4155 \\
      & \textbf{MO} & 1530 & 1308 & 8386 & 6814 & 8911 & 8910 \\
      & \textbf{MT} & 65.75 & 61.91 & 304.56 & 167.68 & 260.11 & 255.65 \\
    \midrule
    \multirow{3}{*}{\textbf{FGIM}}
      % & \textbf{IS} & 805 & 683 & 3979 & 3247 & 4223 & 4223 \\
      & \textbf{MO} & 5361 & 5129 & 13037 & 11471 & 14467 & 14131 \\
      & \textbf{MT} & 25.43 & 26.55 & 85.56 & 65.29 & 85.92 & 82.65 \\
    \midrule
    \multirow{3}{*}{\textbf{CM}}
      % & \textbf{IS} & \textbf{695} & \textbf{611} & \textbf{3907} & \textbf{3174} & \textbf{4151} & \textbf{4151} \\
      % & \textbf{MO} & 2356 & 2046 & 12662 & 10302 & 13448 & 13449 \\
      & \textbf{MO} & \textbf{1188} & \textbf{1139} & \textbf{4678} & \textbf{3892} & \textbf{4939} & \textbf{4939} \\
      & \textbf{MT} & 10.68 & 9.65 & 53.52 & 28.37 & 43.54 & 46.21 \\
    \midrule
    \multirow{3}{*}{\textbf{RNSM}}
      % & \textbf{IS} & \textbf{695} & \textbf{611} & \textbf{3907} & \textbf{3174} & \textbf{4151} & \textbf{4151} \\
      % & \textbf{MO} & 2484 & 2172 & 12786 & 10428 & 13574 & 13575 \\
      & \textbf{MO} & 1317 & 1265 & 4797 & 4016 & 5066 & 5065 \\
      & \textbf{MT} & \textbf{7.52} & \textbf{6.98} & \textbf{42.31} & \textbf{16.90} & \textbf{29.96} & \textbf{26.51} \\
    \bottomrule
  \end{tabular*}
  \endgroup
\end{table}

\input{figures/multi-index-merge-QPS-appendix}

Fig~\ref{fig:bimerge-recall-qps-ndc-appendix} further demonstrates the search performance (Recall vs. QPS/NDC).
$\RNSM$ achieves a trade-off comparable to existing methods and the \textit{Rebuild} baseline on all datasets.
These results confirm that our framework substantially improves merge efficiency without compromising the search performance of the resulting index. It is worth noting that the merged indexes can slightly outperform Rebuild in some configurations, a phenomenon also observed in FGIM~\cite{FGIM_wu_2026_SIGMOD}. This may arise because merging benefits from the accumulated construction effort of pre-built subgraphs and cross-partition searches. Moreover, graphs on sparse subsets yield longer edges, enabling beam search to cover more ground~\cite{CSPG-yang-NIPS-2024}, while cross-partition edges preserve connectivity and recall.

\stitle{Exp-2: Merged Graph Structure Analysis.}
We further analyse the graph index structure for merged indexes and rebuilt indexes. We investigate the graph degree distribution with the 5th, 50th, and 99th out-degree percentiles ($D_{5/50/99}$) and the average out-degree (AD). Table~\ref{tab:final-graph-structure} shows that $\RNSM$ does not obtain its performance by constructing a denser graph. Although $\RNSM$ introduces cross-partition candidates, neighbor pruning keeps its AD comparable to $\CM$ and \textit{Rebuild}. In contrast, $\FGIM$ yields the lowest AD and the smallest $D_{5/50/99}$ values. This sparsity reflects its KNNG transformation, where a small $k$ is used for efficiency and leads to fewer edge candidates for edge pruning compared with \textit{Rebuild}. Although $\CM$ and $\RNSM$ also use small edge candidate sets, they inherit edges that already survived pruning from a PG index, leading to similar out-degree distributions to \textit{Rebuild}.

\stitle{Exp-3: Test of Multiple Indexes Merging.} 
We evaluate the search performance and merge efficiency when merging multiple indexes.
We vary the number of partitions ($\{4, 6, 8, 10\}$) and assess robustness across different data distributions using both random ($\RNSM^*$) and cluster ($\RNSM^+$) partitioning.
Experiments are conducted on DEEP10M, MARCO10M, IMAGE10M, and ANTON10M.
Note that we compare $\RNSM^*$ against both $\CM$ and $\FGIM$, but compare $\RNSM^+$ only against $\CM$, as $\FGIM$ is incompatible with cluster partitions.
Fig~\ref{fig:speedup-comparison} reports the merge efficiency on MARCO10M (full results in Appendix).
Our methods achieve substantial speedups over baselines.
For example, on MARCO10M, $\RNSM^*$ delivers speed-up ratios of up to $2.31\times$ over $\CM$, $1.38\times$ over $\FGIM$, and $4.26\times$ over \textit{Rebuild}, while $\RNSM^+$ achieves a $1.90\times$ speed-up over $\CM$.
Regarding search quality, as shown in Fig~\ref{fig:multi-merge-recall-qps-appendix}, our framework maintains performance comparable to baselines across varying partition numbers and distributions.

We further assess the scalability of our methods regarding dataset size and partition count on the large-scale DEEP100M dataset. We partition the data into $\{20, 30, 40, 50\}$ subsets using both random and cluster partitioning and merge them using $\RNSM^*$ and $\RNSM^+$, respectively. We exclude $\FGIM$ here as its Minimum Candidate Set Strategy requires $ef$ to grow with partition number ($ef>m-1$). Experimentally, its speedup over \textit{Rebuild} drops from 1.97$\times$ with 4 partitions to 0.99$\times$ with 10 on IMAGE10M (see Appendix Exp-c).
Results in Fig~\ref{fig:scale-speedup-comparison} show that $\RNSM$ consistently speeds up merging compared with the $\CM$ and Rebuild baseline with comparable search performance (see Appendix Exp-d). Showing that $\RNSM$ is suited for massive multi-index merging scenarios.

\stitle{Exp-4: Comparison with Overlap-based Method.} 
We compare our framework against the overlapping partition strategy used in DiskANN on 10M-scale datasets.
Unlike DiskANN, which partitions data into overlapping shards (exactly one duplication for each vector) and builds sub-indexes from scratch, $\RNSM$ directly merges existing sub-indexes.
For a fair comparison, we align the maximum shard size of $\DISK$ with the total vector count of the indexes merged by $\RNSM$ (\eg if $\DISK$ uses a 5M shard size, $\RNSM$ merges four 2.5M sub-indexes). We also use the same index of $\DISK$ to focus on the merge algorithm itself.
We include both partitioning and construction time for both methods and report the QPS with 0.95 recall in Table~\ref{tab:diskann-comparison}.
Our method achieves comparable search efficiency while significantly reducing index construction time.
Specifically, $\RNSM$ is consistently faster across all datasets even without any pre-built indexes, delivering speedups of up to $2.01\times$ compared to the scratch-build approach of DiskANN. Furthermore, the search performance is comparable to the $\DISK$ method, with higher QPS$_{0.95}$ in 7 of 8 configurations.

%\begin{table*}[!htbp]
%\centering
%\caption{Comparison of DiskANN and RNSM+. (Time: seconds, QPS$_x$: Query Per Second with recall@10=$x$)}
%\begin{tabular}{l|ccc|ccc|ccc|ccc}
%\toprule
%\multirow{3}{*}{Dataset} & \multicolumn{6}{c|}{Maximum size: 5M} & \multicolumn{6}{c}%{Maximum size: 2.5M} \\
%\cline{2-13}
%& \multicolumn{3}{c|}{Overlap} & \multicolumn{3}{c|}{$\RNSM^*$} & \multicolumn{3}{c|}{Overlap} & \multicolumn{3}{c}{$\RNSM^*$} \\
%\cline{2-13}
%& Time & QPS$_{0.90}$ & QPS$_{0.95}$ & Time & QPS$_{0.90}$ & QPS$_{0.95}$ & Time & QPS$_{0.90}$ & QPS$_{0.95}$ & Time & QPS$_{0.90}$ & QPS$_{0.95}$ \\
%\midrule
%DEEP10M     & 1247 & 4538 & 3042 & \textbf{630} & \textbf{4554} & \textbf{3175} & 1206 & \textbf{5295} & 3225 & \textbf{767} & 4901 & \textbf{3493} \\
%ANTON10M    & 4116 & \textbf{1672} & 783 & \textbf{2928} & 1637 & \textbf{819} & 4084 & \textbf{1803} & 780 & \textbf{2649} & 1791 & \textbf{797} \\
%MARCO10M  & 4410 & 2021 & 1436 & \textbf{3138} & \textbf{2573} & \textbf{1751} & 4446 & 2308 & 1528 & \textbf{3295} & \textbf{2429} & \textbf{1662} \\
%IMAGE10M & 2575 & \textbf{3509} & 2314 & \textbf{1518} & 3423 & \textbf{2395} & 3185 & 3603 & \textbf{2550} & \textbf{1582} & \textbf{3807} & 2456 \\
%\bottomrule
%\end{tabular}
%\label{tab:diskann-comparison}
%\end{table*}

\begin{table}[!t]
\centering
\caption{Comparison of $\DISK$ and $\RNSM^*$. (Time: seconds, QPS$_x$: Query Per Second with recall@10=$x$)}
\vspace*{-4mm}
\resizebox{\columnwidth}{!}{%
\begin{tabular}{c|c|ccc|ccc}
\toprule
\multirow{3}{*}{\textbf{\begin{tabular}[c]{@{}c@{}}Max\\ Size\end{tabular}}} & \multirow{3}{*}{\textbf{Dataset}} & \multicolumn{3}{c|}{\multirow{2}{*}{\textbf{Overlap DiskANN}}}        & \multicolumn{3}{c}{\multirow{2}{*}{\textbf{$\RNSM^*$}}}       \\
                                                                                 &                                   & \multicolumn{3}{c|}{}                                         & \multicolumn{3}{c}{}                                          \\ \cline{3-8} 
                                                                                 &                                   & \textbf{Time} & \textbf{QPS$_{0.90}$} & \textbf{QPS$_{0.95}$} & \textbf{Time} & \textbf{QPS$_{0.90}$} & \textbf{QPS$_{0.95}$} \\ \midrule
\multirow{4}{*}{\textbf{5M}}                                                     & DEEP10M                           & 1247          & 4538                  & 3042                  & \textbf{630}  & \textbf{4554}         & \textbf{3175}         \\
                                                                                 & ANTON10M                          & 4116          & \textbf{1672}         & 783                   & \textbf{2928} & 1637                  & \textbf{819}          \\
                                                                                 & MARCO10M                          & 4410          & 2021                  & 1436                  & \textbf{3138} & \textbf{2573}         & \textbf{1751}         \\
                                                                                 & IMAGE10M                          & 2575          & \textbf{3509}         & 2314                  & \textbf{1518} & 3423                  & \textbf{2395}         \\ \midrule
\multirow{4}{*}{\textbf{2.5M}}                                                   & DEEP10M                           & 1206          & \textbf{5295}         & 3225                  & \textbf{767}  & 4901                  & \textbf{3493}         \\
                                                                                 & ANTON10M                          & 4084          & \textbf{1803}         & 780                   & \textbf{2649} & 1791                  & \textbf{797}          \\
                                                                                 & MARCO10M                          & 4446          & 2308                  & 1528                  & \textbf{3295} & \textbf{2429}         & \textbf{1662}         \\
                                                                                 & IMAGE10M                          & 3185          & 3603                  & \textbf{2550}         & \textbf{1582} & \textbf{3807}         & 2456                  \\ 
\bottomrule
\end{tabular}
}
\vspace*{-5mm}
\label{tab:diskann-comparison}
\end{table}

\begin{figure}[!t]
    \centering
    \begin{minipage}[b]{0.48\columnwidth}
        \centering
        % \begin{figure}[t!]
% \centering
% \begin{small}

\begin{tikzpicture}
    \begin{customlegend}[legend columns=3,
        legend entries={IMAGE, ANTON, MARCO, DEEP, GIST, SIFT},
        legend style={at={(0.5,1.15)},anchor=north,draw=none,font=\scriptsize,column sep=0.05cm},legend image post style={scale=0.8}]  % 减小column sep
    \addlegendimage{area legend,fill=amaranth!50,draw=amaranth,pattern=north west lines,pattern color=amaranth}
    \addlegendimage{area legend,fill=amber!50,draw=amber,pattern=north east lines,pattern color=amber}
    \addlegendimage{area legend,fill=blue!50,draw=navy,pattern=grid,pattern color=navy}
    \addlegendimage{area legend,fill=orange!50,draw=orange,pattern=crosshatch,pattern color=orange}
    \addlegendimage{area legend,fill=green!50,draw=forestgreen,pattern=horizontal lines,pattern color=forestgreen}
    \addlegendimage{area legend,fill=purple!50,draw=violet,pattern=vertical lines,pattern color=violet}
    \end{customlegend}
\end{tikzpicture}
\\[-\lineskip]

\begin{tikzpicture}[scale=1]
\begin{axis}[
    height=3.0cm,
    width=1.1\linewidth, 
    ybar=0.5pt, ymode=log, log basis y={2}, ymin=1,
    ytick={1,2},
    bar width=3pt,
    ylabel=Speed-up Ratio,ylabel style={yshift=-20pt}, 
    enlarge x limits=0.20,
    xmin=-0.15, xmax=2.15,
    xtick={0,1,2},
    xticklabels={NSG, SSG, $\tau$-MNG},
    xticklabel style={font=\scriptsize},
    legend style={font=\scriptsize,draw=none},
    label style={font=\scriptsize},
    tick label style={font=\scriptsize},
    ymajorgrids=true,
    grid style=dashed,
]
% imagenet
\addplot[ybar,fill=amaranth!50,draw=amaranth,pattern=north west lines,pattern color=amaranth] coordinates {
    (0,2.49) (1,1.23) (2,1.97) 
};
% anton
\addplot[ybar,fill=amber!50,draw=amber,pattern=north east lines,pattern color=amber] coordinates {
   (0,1.57) (1,1.15) (2,1.22)  
};
% MARCO
\addplot[ybar,fill=blue!50,draw=navy,pattern=grid,pattern color=navy] coordinates {
   (0,2.07) (1,1.25) (2,1.47)  
};
% deep
\addplot[ybar,fill=orange!50,draw=orange,pattern=crosshatch,pattern color=orange] coordinates {
   (0,1.40) (1,1.33) (2,1.34) 
};
% gist
\addplot[ybar,fill=green!50,draw=forestgreen,pattern=horizontal lines,pattern color=forestgreen] coordinates {
   (0,1.26) (1,1.20) (2,1.16) 
};
% sift
\addplot[ybar,fill=purple!50,draw=violet,pattern=vertical lines,pattern color=violet] coordinates {
   (0,1.54) (1,1.27) (2,1.43) 
};
\end{axis}
\end{tikzpicture}

\vgap
\vspace*{-3mm}
\caption{Speed-up Ratio Comparison of Different Methods}
\vgap
\label{fig:applicability-speedup-comparison}
% \end{small}
% \end{figure} 
    \end{minipage}
    \hfill
    \begin{minipage}[b]{0.48\columnwidth}
        \centering
        % \begin{figure}[t!]
% \centering
% \begin{small}
\begin{tikzpicture}
    \begin{customlegend}[legend columns=2,
        legend entries={NSM, $\RNSM^-$, $\RNSM$, SimJoin},
        legend style={at={(0.5,1.15)},anchor=north,draw=none,font=\scriptsize,column sep=0.05cm},legend image post style={scale=0.8}]  % 减小column sep
    \addlegendimage{area legend,fill=amber!50,draw=amber,pattern=north east lines,pattern color=amber}
    \addlegendimage{area legend,fill=navy!50,draw=navy,pattern=grid,pattern color=navy}
    \addlegendimage{area legend,fill=orange!50,draw=orange,pattern=crosshatch,pattern color=orange}
    \addlegendimage{dashed, thick, amaranth!70}
    \end{customlegend}
\end{tikzpicture}
\\[-\lineskip]

\begin{tikzpicture}[scale=1]
\begin{axis}[
    height=3.0cm,
    width=1.1\linewidth,
    ybar=0.5pt,                
    bar width=3.5pt,       
    ylabel=Percentage (\%),
    ylabel style={yshift=-20pt}, 
    ymin=0, ymax=100,
    xtick={1, 5, 9, 13},
    xticklabels={$k=3$, $k=5$, $k=10$, $k=20$},
    xticklabel style={font=\scriptsize},
    enlarge x limits=0.20,
    legend style={font=\scriptsize,draw=none}, 
    label style={font=\scriptsize},
    tick label style={font=\scriptsize},
    ymajorgrids=true,
    grid style=dashed,
]

% --- 图例项 1: NSM ---
\addplot[fill=amber!50,draw=amber,pattern=north east lines,pattern color=amber] 
coordinates {
    (1, 58.0)   % k=3 (原 L 部分数据)
    (5, 70.5)   % k=5
    (9, 80.7)   % k=10
    (13, 87.9)  % k=20
};

% --- 图例项 2: RNSM* ---
\addplot[fill=navy!50,draw=navy,pattern=grid,pattern color=navy] 
coordinates {
    (1, 54.3)   % k=3 (原 L 部分数据)
    (5, 66.2)   % k=5
    (9, 78.1)   % k=10
    (13, 86.4)  % k=20
};

% --- 图例项 3: RNSM ---
\addplot[fill=orange!50,draw=orange,pattern=crosshatch,pattern color=orange] 
coordinates {
    (1, 66.2)   % k=3 (原 L 部分数据)
    (5, 75.8)   % k=5
    (9, 83.9)   % k=10
    (13, 90.1)  % k=20
};

% --- SimJoin baseline (horizontal dashed line) ---
\draw[dashed, thick, amaranth!70]
    ({rel axis cs:0,0} |- {axis cs:0,48.67}) --
    ({rel axis cs:1,0} |- {axis cs:0,48.67});

\end{axis}
\end{tikzpicture}

\vgap
\vspace*{-3mm}
\caption{Comparison of Local Sliding Ratio with SIFT1M}
\vgap
\label{fig:ablation-rknn}
% \end{small}
% \end{figure}
    \end{minipage}
    \vspace*{-2mm}
\end{figure}

% \stitle{Exp-5: Scalability Study for Index Merging.} 
% We assess the scalability of our methods regarding dataset size and partition count on the large-scale DEEP100M dataset.
% We partition the data into $\{20, 30, 40, 50\}$ subsets using both random and cluster partitioning and merge them using $\RNSM^*$ and $\RNSM^+$, respectively.
% We exclude $\FGIM$ here as its theoretical guarantees do not hold when the partition number approaches the maximum degree $R$ (usually $\le$30) of the sub-indexes. Empirically, its speedup over \textit{Rebuild} drops from 1.97$\times$ with 4 partitions to 0.99$\times$ with 10 on IMAGE10M (see Appendix Exp-c).
% Results are shown in Fig~\ref{fig:scale-speedup-comparison} and Fig~\ref{fig:scalability-recall-qps-appendix}.
% $\RNSM$ consistently outperforms the $\CM$ baseline without compromising search performance.
% Notably, the efficiency advantage of our framework amplifies as the system scales: the speedup ratio grows as the number of partitions increases, demonstrating that $\RNSM$ is well-suited for massive multi-index merging scenarios.

\begin{figure}[!t]
\centering
\begin{small}

% Global legend
\begin{tikzpicture}
\begin{customlegend}[legend columns=3,
legend entries={RNSM, SimJoin, Speedup},
legend style={at={(0.5,1.15)},anchor=north,draw=none,font=\scriptsize,column sep=0.2cm},legend image post style={scale=0.8}]
\addlegendimage{area legend,fill=amaranth!50,draw=amaranth,pattern=north west lines,pattern color=amaranth}
\addlegendimage{area legend,fill=amber!50,draw=amber,pattern=north east lines,pattern color=amber}
\addlegendimage{line width=0.2mm,color=navy,mark=o,mark size=0.5mm}
\end{customlegend}
\end{tikzpicture}
\\[-\lineskip]

% ======================
% MARCO10M Dataset
% ======================
\subfloat[MARCO10M]{\vspace{-2mm}
\begin{tikzpicture}[scale=1]
% Left y-axis for Time
\begin{axis}[
    height=3.0cm,
    width=\columnwidth/2.10,
    ybar=0.5pt,
    bar width=4.0pt,
    xmode=log, log basis x={2},
    xlabel={Number of Threads},
    ylabel={Time ($10^3$ s)}, 
    axis y line*=left,
    ylabel near ticks, 
    yticklabel={\pgfmathparse{\tick/1000}\pgfmathprintnumber{\pgfmathresult}}, 
    enlarge x limits=0.20,
    xtick={1,2,4,8,16,32},
    xticklabels={1,2,4,8,16,32},
    label style={font=\scriptsize},
    tick label style={font=\scriptsize},
    ymajorgrids=true,
    grid style=dashed,
    legend style={font=\scriptsize,draw=none}
]
% RNSM Time (bars)
\addplot[ybar,fill=amaranth!50,draw=amaranth,pattern=north west lines,pattern color=amaranth] coordinates {
    (1,3191.56) (2,1335.45) (4,658.682) (8,369.206) (16,286.966) (32,235.719)
};
% SimJoin Time (bars)
\addplot[ybar,fill=amber!50,draw=amber,pattern=north east lines,pattern color=amber] coordinates {
    (1,3677.32) (2,2055.63) (4,1021.5) (8,609.302) (16,461.688) (32,459.324)
};
\end{axis}

% Right y-axis for Speedup ratio
\begin{axis}[
    height=3.0cm,
    width=\columnwidth/2.10,
    xmode=log, log basis x={2},
    ylabel={Speedup Ratio}, 
    axis y line*=right,
    axis x line=none, 
    ylabel near ticks, 
    enlarge x limits=0.20,
    xtick={1,2,4,8,16,32}, 
    label style={font=\scriptsize},
    tick label style={font=\scriptsize},
    legend style={font=\scriptsize,draw=none},
]
% Speedup ratio (SimJoin_time / RNSM_time)
\addplot[line width=0.2mm,color=navy,mark=o,mark size=0.5mm] coordinates {
    (1,1.15) (2,1.54) (4,1.55) (8,1.65) (16,1.61) (32,1.95)
};
\end{axis}
\end{tikzpicture}}\hfill 
% ======================
% IMAGENET10M Dataset
% ======================
\subfloat[IMAGE10M]{\vspace{-2mm}
\begin{tikzpicture}[scale=1]
% Left y-axis for Time
\begin{axis}[
    height=3.0cm,
    width=\columnwidth/2.10,
    ybar=0.5pt,
    bar width=4.0pt,
    xmode=log, log basis x={2},
    xlabel={Number of Threads}, 
    ylabel={Time ($10^3$ s)}, 
    axis y line*=left,
    ylabel near ticks, 
    yticklabel={\pgfmathparse{\tick/1000}\pgfmathprintnumber{\pgfmathresult}}, 
    enlarge x limits=0.20,
    xtick={1,2,4,8,16,32},
    xticklabels={1,2,4,8,16,32},
    label style={font=\scriptsize},
    tick label style={font=\scriptsize},
    ymajorgrids=true,
    grid style=dashed,
    legend style={font=\scriptsize,draw=none}
]
% RNSM Time (bars)
\addplot[ybar,fill=amaranth!50,draw=amaranth,pattern=north west lines,pattern color=amaranth] coordinates {
    (1,1799.97) (2,812.06) (4,411.47) (8,246.33) (16,185.29) (32,157.72)
};
% SimJoin Time (bars)
\addplot[ybar,fill=amber!50,draw=amber,pattern=north east lines,pattern color=amber] coordinates {
    (1,2887.02) (2,1593.52) (4,768.42) (8,452.26) (16,372.35) (32,328.41)
};
\end{axis}

% Right y-axis for Speedup ratio
\begin{axis}[
    height=3.0cm,
    width=\columnwidth/2.10,
    xmode=log, log basis x={2},
    ylabel={Speedup Ratio}, 
    axis y line*=right,
    axis x line=none, 
    ylabel near ticks, 
    enlarge x limits=0.20,
    xtick={1,2,4,8,16,32}, 
    label style={font=\scriptsize},
    tick label style={font=\scriptsize},
    legend style={font=\scriptsize,draw=none},
]
% Speedup ratio (SimJoin_time / RNSM_time)
\addplot[line width=0.2mm,color=navy,mark=o,mark size=0.5mm] coordinates {
    (1,1.60) (2,1.96) (4,1.87) (8,1.84) (16,2.01) (32,2.08)
};
\end{axis}
\end{tikzpicture}}

\vgap
\vspace*{-2mm}
\caption{Parallel Scalability Comparison}
\vgap
\label{fig:parallel-scalability}
\end{small}
\end{figure}

\stitle{Exp-5: Test of Incremental Merging.}
We evaluate incremental merging on 10M-scale datasets by splitting each dataset into ten 1M-vector partitions. Starting from one partition, we merge one additional 1M-vector sub-index into the current merged index at each step until all ten partitions form the final 10M-vector index.
As shown in Fig~\ref{fig:incremental-10m-recall-qps}, $\RNSM$ achieves search performance comparable to or better than the baselines. In terms of merge efficiency, $\RNSM$ delivers speedups of up to $1.57\times$ over $\CM$ and $6.44\times$ over $\FGIM$.
\begin{figure}[!t]
\centering
\begin{small}

% Global legend
\begin{tikzpicture}
    \begin{customlegend}[legend columns=3,
        legend entries={CM,RNSM,FGIM},
        legend style={at={(0.5,1.15)},anchor=north,draw=none,font=\scriptsize,column sep=0.2cm}]
    \addlegendimage{line width=0.2mm,color=violate,mark=o,mark size=0.5mm}
    \addlegendimage{line width=0.2mm,color=amaranth,mark=triangle,mark size=0.5mm}
    \addlegendimage{line width=0.2mm,color=navy,mark=square,mark size=0.5mm}
    % \addlegendimage{line width=0.2mm,color=amber,mark=star,mark size=0.5mm}
    \end{customlegend}
\end{tikzpicture}
\\[-\lineskip]
\vspace*{-1mm}

    \subfloat[IMAGE10M]{\vspace{-2mm}
\begin{tikzpicture}[scale=1]
\begin{axis}[
height=3.0cm,
width=\columnwidth/2.00,
xlabel=recall@10, xlabel style={yshift=+8pt}, xmin=0.55,
ylabel=Query Per Second, ylabel style={yshift=-15pt},
label style={font=\scriptsize},
tick label style={font=\scriptsize},
ymajorgrids=true,
xmajorgrids=true,
grid style=dashed,
]

    % CM
\addplot[line width=0.2mm,color=violate,mark=o,mark size=0.5mm]
plot coordinates {
(0.8023,26160.43244)
(0.89972,18313.20869)
(0.93664,14506.31962)
(0.9547,12251.20876)
(0.96402,10659.66511)
(0.97091,9404.427544)
(0.97511,8445.342682)
(0.97844,7573.958376)
(0.98083,6821.084892)
(0.98268,6498.691635)
(0.98385,6220.967348)
(0.98449,5574.357234)
(0.98527,5114.077693)
(0.98615,4777.225851)
(0.98689,4506.849)
};
    % RGTM
\addplot[line width=0.2mm,color=amaranth,mark=triangle,mark size=0.5mm]
plot coordinates {
(0.78396,35721.62574)
(0.88693,24099.06007)
(0.92475,19613.01113)
(0.94473,15683.89575)
(0.95576,14253.59499)
(0.96441,11818.85106)
(0.96954,11158.85949)
(0.97294,10142.49074)
(0.97612,9573.645812)
(0.9784,8698.630978)
(0.98002,7900.822243)
(0.9812,7717.058198)
(0.9826,7411.2959)
(0.98337,6931.034169)
(0.98437,6629.387816)
};
    % FGIM
\addplot[line width=0.2mm,color=navy,mark=square,mark size=0.5mm]
plot coordinates {
(0.60245,23191.55436)
(0.75474,15540.82702)
(0.81562,12798.96761)
(0.85041,11177.30626)
(0.87374,8006.230929)
(0.89003,6868.730877)
(0.90288,6531.118906)
(0.91179,6041.515425)
(0.92066,5589.499446)
(0.92807,5245.195417)
(0.93252,4883.604247)
(0.93599,4540.277974)
(0.94,4285.140683)
(0.94407,4069.943246)
(0.9458,3872.431591)
};
%     % BuildAsOne
% \addplot[line width=0.2mm,color=amber,mark=star,mark size=0.5mm]
% plot coordinates {
% (0.86468,35045.2)
% (0.93177,24508.3)
% (0.95548,19229.9)
% (0.96603,15902.7)
% (0.97257,13654.1)
% (0.97604,11706.5)
% (0.97827,10669.7)
% (0.98055,10139.3)
% (0.98205,9305.28)
% (0.98338,8492.19)
% (0.9837,7963.92)
% (0.98436,7344.1)
% (0.98517,6995.2)
% (0.98597,6239.43)
% (0.98648,6236.67)
% };
\end{axis}
\end{tikzpicture}}\hspace{2mm}
    \subfloat[MARCO10M]{\vspace{-2mm}
\begin{tikzpicture}[scale=1]
\begin{axis}[
height=3.0cm,
width=\columnwidth/2.00,
xlabel=recall@10, xlabel style={yshift=+8pt}, xmin=0.55,
ylabel=Query Per Second, ylabel style={yshift=-15pt},
label style={font=\scriptsize},
tick label style={font=\scriptsize},
ymajorgrids=true,
xmajorgrids=true,
grid style=dashed,
]

    % CM
\addplot[line width=0.2mm,color=violate,mark=o,mark size=0.5mm]
plot coordinates {
(0.7657,23699.91453)
(0.8623,17277.01972)
(0.9038,14187.24427)
(0.9258,11987.69032)
(0.9434,10555.47033)
(0.9507,9459.073964)
(0.9577,8141.461475)
(0.963,7644.867643)
(0.9679,7075.527735)
(0.9701,6483.413082)
(0.974,5967.191629)
(0.9784,5530.043866)
(0.9811,5276.717531)
(0.9824,4817.644148)
(0.984,4664.20323)
};
    % RGTM
\addplot[line width=0.2mm,color=amaranth,mark=triangle,mark size=0.5mm]
plot coordinates {
(0.7251,26083.54162)
(0.8368,19223.90888)
(0.8852,15834.45194)
(0.9081,13788.46691)
(0.9233,11908.43637)
(0.9348,10636.01158)
(0.9457,9515.812978)
(0.9504,8317.162551)
(0.9549,7987.759876)
(0.9595,7492.517129)
(0.9625,6803.889969)
(0.9665,6468.229367)
(0.97,5744.407516)
(0.9713,5583.605761)
(0.9728,5320.058836)
};
    % FGIM
\addplot[line width=0.2mm,color=navy,mark=square,mark size=0.5mm]
plot coordinates {
(0.7726,14781.67736)
(0.8924,11827.9719)
(0.9259,10486.34086)
(0.9419,8693.350931)
(0.9515,7406.33689)
(0.9591,6387.856654)
(0.9639,5636.894376)
(0.9686,5069.214677)
(0.9734,4608.935849)
(0.9763,4227.498521)
(0.9776,3878.418574)
(0.9796,3599.380408)
(0.9821,3351.024767)
(0.9836,3147.128325)
(0.9858,2968.696981)
};
%     % BuildAsOne
% \addplot[line width=0.2mm,color=amber,mark=star,mark size=0.5mm]
% plot coordinates {
% (0.7966,25435.4)
% (0.8818,19113.2)
% (0.9208,15883.6)
% (0.941,13245.4)
% (0.953,11670.1)
% (0.9588,10180.5)
% (0.9644,9210.5)
% (0.9689,8296.67)
% (0.9725,7643.78)
% (0.9747,7084.76)
% (0.9761,6532.1)
% (0.978,6110.35)
% (0.9784,5606.05)
% (0.9793,5389.37)
% (0.9819,5133.22)
% };
\end{axis}
\end{tikzpicture}}\hspace{2mm}

\vgap
\vspace*{-2mm}
\caption{Recall@10-QPS Tradeoff for Incremental Merge}
\vgap
\label{fig:incremental-10m-recall-qps}
\end{small}
\end{figure}

\stitle{Exp-6: Applicability Study.} 
We evaluate $\RNSM$ on three additional PG indexes: NSG, SSG, and $\tau$-MNG.
We use recommended settings for NSG and SSG~\cite{GraphBasedANNS-wang-vldb-2021} and implement $\tau$-MNG on the NSG codebase~\cite{Efficient-peng-SIGMOD-2023}.
Fig~\ref{fig:applicability-speedup-comparison} reports $\RNSM$ achieves speedups of up to $2.49\times$, $1.33\times$, and $1.97\times$ for NSG, SSG, and $\tau$-MNG, respectively. $\RNSM$ also maintains search performance comparable to $\CM$ across these indexes (see Appendix Exp-b), demonstrating its broad applicability.

\stitle{Exp-7: Effects of Reverse Neighbor Pivots Selection.} 
We analyze the impact of RKNN-based pivot selection on the fraction of Local Sliding operations. On SIFT1M with $k\in\{3,5,10,20\}$, we compare four variants: (1) NSM: randomly select pivots, slide to KNNs; (2) SimJoin: MST roots as pivots, slide to internal nodes and leaves; (3) $\RNSM^-$: randomly select pivots, slide to RKNNs; and (4) $\RNSM$: pivots and corresponding followers selected by our greedy algorithm. As shown in Fig~\ref{fig:ablation-rknn}, SimJoin has the lowest ratio, and its fixed MST topology determines the sliding proportion. For the other methods, the ratio increases with $k$. $\RNSM$ rises from 66.2\% at $k=3$ to 90.1\% at $k=20$, showing that larger $k$ exposes more followers. $\RNSM$ consistently outperforms all other methods. $\RNSM^-$ underperforms NSM because NSM iteratively propagates previous results instead of sliding to a fixed follower set.

\stitle{Exp-8: Parallel Scalability.} We compare $\RNSM$ and SimJoin on MARCO10M and IMAGE10M using $\{1,2,4,8,16,32\}$ threads at comparable merged-index search quality for parallelism comparison.
Fig~\ref{fig:parallel-scalability} shows that $\RNSM$ outperforms SimJoin across all thread configurations, with its advantage increasing with more threads: from $1.15\times$ to $1.95\times$ on MARCO10M and from $1.60\times$ to $2.08\times$ on IMAGE10M. The MST propagation of SimJoin introduces inter-node dependencies that limit parallelism, while $\RNSM$ decouples sliding through independent pivot selection.

\begin{figure}[!t]
\centering
\begin{small}

% 全局图例
\begin{tikzpicture}
\begin{customlegend}[legend columns=4,
legend entries={Random, KNNG, MOS, MOS($\Delta$=3)},%legend entries={random, circulant, MOS},
legend style={at={(0.5,1.15)},anchor=north,draw=none,font=\scriptsize,column sep=0.2cm}]
% \addlegendimage{line width=0.2mm,color=violate,mark=o,mark size=0.5mm}
\addlegendimage{line width=0.2mm,color=amaranth,mark=triangle,mark size=0.5mm}
\addlegendimage{line width=0.2mm,color=navy,mark=o,mark size=0.5mm}
\addlegendimage{line width=0.2mm,color=amber,mark=triangle,mark size=0.5mm}
% \addlegendimage{line width=0.2mm,color=blue,mark=diamond,mark size=0.5mm}
\addlegendimage{line width=0.2mm,color=forestgreen,mark=star,mark size=0.5mm}
\addlegendimage{line width=0.2mm,color=violate,mark=square,mark size=0.5mm}
\end{customlegend}
\end{tikzpicture}
\\[-\lineskip]
\vspace*{-1mm}
% ======================
% 第1行：6个 recall@10 - QPS 曲线图
% ======================
\begin{tikzpicture}[scale=1]
\begin{axis}[
    height=3.0cm,
    width=\columnwidth/2.00,
    xlabel=recall@10, xlabel style={yshift=+8pt},
    xmin=0.9,xmax=1.0,
    ylabel=Query Per Second, ylabel style={yshift=-15pt}, 
    label style={font=\scriptsize},
    tick label style={font=\scriptsize},
    ymajorgrids=true,
    xmajorgrids=true,
    grid style=dashed,
]
\addplot[line width=0.2mm,color=amaranth,mark=triangle,mark size=0.5mm]
plot coordinates {
(0.6155, 14386.5)
(0.7482, 9294.59)
(0.8188, 7015.22)
(0.86, 5697.44)
(0.89, 4768.89)
(0.9092, 4154.73)
(0.9233, 3442.32)
(0.934, 3314.72)
(0.9416, 3009.84)
(0.9498, 2736.15)
(0.9578, 2263.01)
(0.9611, 2310.11)
(0.9639, 2234.2)
(0.9665, 2044.03)
(0.9693, 2095.23)
};
%KNNG
\addplot[line width=0.2mm,color=navy,mark=o,mark size=0.5mm]
plot coordinates {
(0.5679, 17097.5)
(0.7115, 11282.2)
(0.7858, 7875.5)
(0.8328, 6911.82)
(0.8665, 5862.1)
(0.8897, 5037.97)
(0.9085, 4410.85)
(0.9219, 4029.06)
(0.9333, 3440.11)
(0.9419, 3387.12)
(0.9484, 3106.51)
(0.9544, 2912.42)
(0.96, 2635.16)
(0.9642, 2563.46)
(0.9674, 2416.17)
(0.9697, 2288.08)
};
%MOS
\addplot[line width=0.2mm,color=amber,mark=triangle,mark size=0.5mm]
plot coordinates {
(0.6398, 16548.5)
(0.7714, 9457.11)
(0.8386, 8096.5)
(0.8794, 6561.65)
(0.9077, 5548.61)
(0.9261, 4714.25)
(0.9387, 4272.77)
(0.9477, 3750.52)
(0.955, 3284.81)
(0.9609, 3079.58)
(0.9672, 2894.46)
(0.9702, 2722.71)
(0.9725, 2451.48)
(0.9758, 2381.84)
(0.9777, 2253.36)
};
%DELTA3
\addplot[line width=0.2mm,color=forestgreen,mark=star,mark size=0.5mm]
plot coordinates {
(0.4607, 17222.7)
(0.6299, 10964.8)
(0.7218, 8281.17)
(0.7832, 6115.62)
(0.8273, 5721.94)
(0.8579, 4974.29)
(0.8814, 4294.34)
(0.8987, 3836.91)
(0.9133, 3301.2)
(0.9245, 2956.32)
(0.9335, 2924.65)
(0.9402, 2693.07)
(0.9466, 2518.55)
(0.9519, 2272.41)
(0.9555, 2262.28)
};
% %DELTA2-unitweight
% \addplot[line width=0.2mm,color=violate,mark=square,mark size=0.5mm]
% plot coordinates {
% (0.5763, 15379.3)
% (0.7083, 10136.7)
% (0.7825, 7680.34)
% (0.8262, 6175.72)
% (0.8571, 4827.87)
% (0.8771, 4549.38)
% (0.8937, 4014.57)
% (0.9068, 3651.36)
% (0.9163, 3292.24)
% (0.9234, 2925.51)
% (0.9293, 2807.18)
% (0.9345, 2620.9)
% (0.9383, 2364.68)
% (0.9425, 2302.77)
% (0.9479, 2186.09)
% };
\end{axis}
\end{tikzpicture}
\begin{tikzpicture}[scale=1]
\begin{axis}[
    height=3.0cm,
    width=\columnwidth/2.00,
    xlabel=recall@1, xlabel style={yshift=+8pt},
    xmin=0.9,xmax=1.0,
    ylabel=Query Per Second, ylabel style={yshift=-15pt}, 
    label style={font=\scriptsize},
    tick label style={font=\scriptsize},
    ymajorgrids=true,
    xmajorgrids=true,
    grid style=dashed,
]
\addplot[line width=0.2mm,color=amaranth,mark=triangle,mark size=0.5mm]
plot coordinates {
(0.688, 13972.2)
(0.81, 9418.37)
(0.863, 6426.2)
(0.9, 5420.51)
(0.922, 4890.7)
(0.932, 4000.51)
(0.945, 3767.57)
(0.949, 3148.89)
(0.959, 2981.66)
(0.965, 2872.87)
(0.969, 2595.5)
(0.974, 2443.07)
(0.974, 2306.93)
(0.976, 2112.73)
(0.977, 2069.81)
};
%KNNG
\addplot[line width=0.2mm,color=navy,mark=o,mark size=0.5mm]
plot coordinates {
(0.667, 17291.5)
(0.798, 11297.7)
(0.852, 8561.82)
(0.895, 6936.1)
(0.918, 5895.48)
(0.937, 5133.22)
(0.949, 4524.43)
(0.96, 4030.58)
(0.967, 3696.41)
(0.971, 3253.56)
(0.978, 3116.46)
(0.981, 2889.29)
(0.986, 2710.85)
(0.99, 2560.31)
(0.991, 2235.17)
(0.991, 2416.97)
};
%MOS
\addplot[line width=0.2mm,color=amber,mark=triangle,mark size=0.5mm]
plot coordinates {
(0.709, 16661.6)
(0.833, 10721)
(0.901, 8120.83)
(0.934, 6117.88)
(0.953, 5583.44)
(0.961, 4837.12)
(0.967, 4261.31)
(0.974, 3860.45)
(0.978, 3587.19)
(0.981, 3093.31)
(0.985, 3021.32)
(0.986, 2823)
(0.987, 2651.34)
(0.99, 2394.12)
(0.991, 2349.72)
};
%DELTA3
\addplot[line width=0.2mm,color=forestgreen,mark=star,mark size=0.5mm]
plot coordinates {
(0.4607, 17034.6)
(0.6299, 10968.3)
(0.7218, 7684.19)
(0.7832, 6799.4)
(0.8273, 5682.89)
(0.8579, 4899.74)
(0.8814, 4397.85)
(0.8987, 3906.01)
(0.9133, 3400.23)
(0.9245, 3295.22)
(0.9335, 3013.38)
(0.9402, 2834.33)
(0.9466, 2548.32)
(0.9519, 2469.37)
(0.9555, 2326.36)
};
\end{axis}
\end{tikzpicture}

\vgap
\vspace*{-3mm}
\caption{Merge Order Selection Comparison on DEEP10M}
\vgap
\label{fig:ablation-mos}
\end{small}
\end{figure}
\stitle{Exp-9: Effects of MOS Technique.} 
We validate the effectiveness of the MOS strategy for merging multiple partitions. On DEEP10M split into 10 partitions, we compare MOS with a random regular graph, a KNNG with the same average degree, and an MOS variant with diameter upper bound $\Delta=3$. We exclude MST, Star, and Path because they fail to reach 0.9 recall in Fig~\ref{fig:observation2}. Fig~\ref{fig:ablation-mos} shows that MOS consistently yields the best search performance.

% \subsection{Ablation Study}
\section{Related Works}
%We review existing approaches for multiple graph indexes for AKNN search.

% \sstitle{HNSW-merger.}
% HNSW-merger~\cite{hnswmerger-jin-sigmod-2026} uses one-way cross-querying from the smaller to the larger index with a small candidate set $ef$. It further exploits the HNSW hierarchy to cache higher-layer entry points. However, they handle multi-index merges solely by incrementally merging the two largest indexes, limiting speedups.
\sstitle{HNSW-merger.}
HNSW-merger~\cite{hnswmerger-jin-sigmod-2026} proposed merging HNSW indexes via cross-querying. Specifically, they use one-way cross-querying from the smaller to the larger index with a small candidate set $ef$. It further exploits the HNSW hierarchy to cache higher-layer entry points. However, they handle multi-index merges solely by incrementally merging the two largest indexes, limiting speedups.

\sstitle{Similarity Join.} 
SimJoin~\cite{SimJoin-xie-POM-2025} was originally proposed to report all cross-partition vector pairs within a distance threshold. Cross-index KNN discovery of index merging is a top-$k$ variant of this task. However, it requires an MST over one partition and cannot directly merge two PG indexes.

\sstitle{KNNG Merge.} 
S-merge~\cite{MergeOfKNNG-zhao-IEEE-2022} connects two KNNGs with random edges and refines them using NNDecent~\cite{NNDecent-dong-WWW-2011}. Multi-way Merge~\cite{GraphMerge-zhang-arxiv-2025} merges multiple KNNGs in one pass by cross-matching only newly added cross-index candidates. Batch processing in NNDecent lets KNNGs converge faster than PG-based methods~\cite{ExtendedNeighborhoodGraph-zhang-SIGIR-2025}, but KNNGs are generally less effective for AKNN search than navigational graph indexes~\cite{GraphTreeIndexes-wang-IEEE-2023, GraphBasedANNS-wang-vldb-2021, ANNBenchmarks-aumuller-IS-2020, HighDimensionalDataExperiments-li-IEEE-2016}.

\sstitle{Overlapping Partitioning.} 
This strategy is often employed for billion-scale graph index construction under memory constraints (\eg DiskANN). The dataset is partitioned into overlapping clusters via k-means, where each data point is assigned to at least 2 nearest partitions. Sub-indexes are constructed for each partition, and a compact graph is subsequently merged by concatenating the neighborhoods of identical nodes across different subgraphs. However, it requires clustering the dataset and building sub-indexes from scratch, making it inapplicable to merging existing indexes.

\section{Conclusion and Future Work}~\label{sec:conclude}
In this paper, we addressed the challenge of merging multiple proximity graph indexes to enable efficient multi-index searching. By identifying the core components of graph merging, we proposed the Reverse Neighbor Sliding Merge ($\RNSM$) algorithm to facilitate pairwise merging. Furthermore, based on merge order analysis, we designed the Merge Order Selection (MOS) algorithm. MOS is suitable for both random and cluster partitions, and it reduces redundancy in pairwise processing by pruning unnecessary merges. Our experimental results demonstrate that our framework significantly improves merging efficiency while preserving the comparable search performance characteristic of graph indexes. Notably, our method exhibits robust scalability, performing effectively even with large-scale datasets.

In practice, systems may involve overlapping, unequal-sized, or mixed partitions rather than simplified disjoint random or cluster partitions. While our method can be directly extended to these settings, it would be suboptimal, and designing more sophisticated merge-cost models tailored to them remains future work.
Moreover, the three proxies used in MOS are indirect and empirical. Richer size-aware, geometric, or learned proxies that directly optimize merge time and quality remain future work.

% 
% \stitle{Limitations and Future Work.}
% Our formulation simplifies real deployments with random or cluster partitions as a controlled simplification. Real systems may face more complex scenarios such as overlapping partitions, unequal partition sizes, or mixtures of random and cluster partitions. Although our method can handle them with unique vector IDs, a merge-direction principle, and considers mixed partitions as a special case of cluster partitions, dedicated merge-cost designs for more general and complex settings remain future work.

% MOS uses centroid distance, merge-order graph degree, and graph diameter as empirically motivated but indirect proxies for multi-index merge cost and quality. Direct objectives that minimize merge time or maximize final recall would be preferable, but such objectives are nontrivial to define because PG indexes are themselves empirical approximations~\cite{NSG-fu-VLDB-2019,HNSW-malkov-IEEE-2020}. Nevertheless, there are more objectives we can consider in the merge cost model. For instance, existing methods use partition size only to merge a smaller index into a larger one to reduce merge time, without considering its effect on merge quality. Likewise, centroid distance only empirically measures the likelihood that two partitions share KNNs. Richer geometric or learned proxies may identify better merge pairs~\cite{LIRA-zeng-ACM-2025}.
% 
% [SHRINK 1 limitation 1 future]

\balance
\bibliographystyle{ACM-Reference-Format}
\bibliography{sample}

\clearpage

\section*{Appendix}
\subsection*{Proof}\label{apendixsec:proof}
\stitle{Proof of Theorem~\ref{thm:dps-nphard}}
\begin{proof}
To prove that the Minimum Cost Independent Dominating Pivot Selection (IDPS) problem is NP-hard, we establish a polynomial-time reduction from the NP-complete \textbf{Minimum Independent Dominating Set (MIDS)} problem~\cite{MIDS-Irving-1991}.

First, we define the decision versions of both problems.
\begin{itemize}[leftmargin=1.5em]
    \item \textbf{MIDS-Decision:} Given an undirected graph $G'=(V', E')$ and a positive integer $k$, does there exist an independent dominating set $D \subseteq V'$ with $|D| \leq k$? (i.e., every $v \notin D$ has a neighbor in $D$, and no two nodes in $D$ are adjacent.)
    \item \textbf{IDPS-Decision:} Given a vector dataset $X$, a graph $G=(V, E)$, a distance function $\delta(\cdot, \cdot)$, a naive search cost $\gamma$, and a target cost $C$, does there exist a pivot subset $P \subseteq V$ satisfying the independent dominating set constraint such that $\mathcal{L}(P) \leq C$?
\end{itemize}

\noindent\textbf{NP Membership.}
IDPS-Decision is in NP: given a candidate $P$, one can verify in polynomial time that (i) every $v \notin P$ has a neighbor in $P$ (domination), (ii) no two nodes in $P$ are adjacent (independence), and (iii) $\mathcal{L}(P) \leq C$.

\noindent\textbf{Polynomial-Time Construction.}
Given an instance $\langle G', k \rangle$ of MIDS-Decision, we construct an instance $\langle X, G, \delta, \gamma, C \rangle$ of IDPS-Decision as follows:
\begin{enumerate}
    \item Let $G = G'$, so $V = V'$ and $E = E'$.
    \item Map each node $v_i \in V$ to a distinct vector $x_i \in X$. Define the discrete metric: $\delta(x_u, x_u) = 0$ and $\delta(x_u, x_v) = 1$ for all $u \neq v$. (This trivially satisfies all metric space properties.)
    \item Set $\gamma = 2$, $\beta = 1$, and $C = k + |V|$.
\end{enumerate}
This construction runs in polynomial time $\mathcal{O}(|V|^2)$.

\noindent\textbf{Cost Simplification.}
For any $P$ satisfying the independent dominating set constraint, $P$ is in particular a dominating set, so every $v \notin P$ has a neighbor $p \in P \cap N(v, G)$. With $\delta \equiv 1$ for distinct nodes, the sliding cost is:
\begin{equation*}
    \mathcal{L}_{slide}(P) = \sum_{v \in V \setminus P} \min_{p \in P \cap N(v,G)} \delta(x_v, x_p) = |V| - |P|
\end{equation*}
Substituting $\gamma = 2$:
\begin{equation*}
    \mathcal{L}(P) = 2|P| + (|V| - |P|) = |P| + |V|
\end{equation*}

\noindent\textbf{Correctness of Reduction.}
\begin{itemize}[leftmargin=1.5em]
    \item \textbf{($\Rightarrow$)} Suppose there exists an independent dominating set $D$ in $G'$ with $|D| \leq k$. Let $P = D$. Since $G = G'$, $P$ is an independent dominating set in $G$, satisfying both IDPS constraints. The total cost is $\mathcal{L}(P) = |P| + |V| \leq k + |V| = C$. Thus, a valid IDPS solution exists.

    \item \textbf{($\Leftarrow$)} Suppose there exists a pivot subset $P \subseteq V$ satisfying the IDPS constraints with $\mathcal{L}(P) \leq C$. By the IDPS constraints, $P$ is an independent dominating set of $G = G'$. From $\mathcal{L}(P) \leq C$ we get $|P| + |V| \leq k + |V|$, so $|P| \leq k$. Thus, an independent dominating set of size at most $k$ exists in $G'$.
\end{itemize}

Since MIDS-Decision is NP-complete and MIDS-Decision $\leq_p$ IDPS-Decision, IDPS-Decision is NP-complete. Consequently, the IDPS optimization problem is NP-hard.
\end{proof}

\subsection*{Extra Experiments}\label{apendixsec:extra-exp}

\stitle{Exp-a: Runtime Breakdown of RNSM.}
In this section, we analyze the execution time of RNSM by decomposing it into its core components (32 threads): neighbor expansion, reverse inversion and pivot selection, and execution of merging.
As shown in Fig~\ref{fig:rnsm-breakdown}, the merging step dominates runtime across all datasets, accounting for more than 83.1\% of total time, since it involves the main graph-update operations driven by the sliding cost model. Neighbor expansion contributes 8.7\%--15.3\%, reflecting the overhead of retrieving nearest neighbors of every node from the source index. Pivot selection is consistently negligible ($\leq$5.0\%), confirming that the sequential pivot in Algorithm~\ref{algo:bi-merge} adds virtually no overhead.

\begin{figure}[!htbp]
\centering
\begin{small}

% Global legend
\begin{tikzpicture}
\begin{customlegend}[legend columns=3,
legend entries={Expand Neighbor, Reverse Select Pivots, Merge},
legend style={at={(0.5,1.15)},anchor=north,draw=none,font=\scriptsize,column sep=0.2cm},legend image post style={scale=0.8}]
\addlegendimage{area legend,fill=amber!50,draw=amber,pattern=north east lines,pattern color=amber}
\addlegendimage{area legend,fill=navy!50,draw=navy,pattern=crosshatch,pattern color=navy}
\addlegendimage{area legend,fill=amaranth!50,draw=amaranth,pattern=north west lines,pattern color=amaranth}
\end{customlegend}
\end{tikzpicture}
\\[-\lineskip]

\begin{tikzpicture}[scale=1]
\begin{axis}[
    height=\columnwidth/2.6,
    width=\columnwidth/1.9,
    ybar stacked,
    bar width=6pt,
    ymin=0,
    ymax=105,
    ytick={0,20,40,60,80,100},
    xtick=data,
    xticklabels={MARCO10M, IMAGE10M, DEEP10M, ANTON10M},
    x tick label style={rotate=45, anchor=east, font=\scriptsize},
    ylabel={Rate (\%)},ylabel style={yshift=-20pt},
    label style={font=\scriptsize},
    tick label style={font=\scriptsize},
    ymajorgrids=true,
    grid style=dashed,
    legend style={font=\scriptsize,draw=none},
    enlarge x limits=0.20,
]
% Merge (bottom layer)
\addplot[fill=amaranth!50,draw=amaranth,pattern=north west lines,pattern color=amaranth] coordinates {
    (1,83.1)
    (2,88.0)
    (3,86.3)
    (4,86.9)
};
% Expand Neighbor (middle layer)
\addplot[fill=amber!50,draw=amber,pattern=north east lines,pattern color=amber] coordinates {
    (1,15.3)
    (2,9.5)
    (3,8.7)
    (4,10.9)
};
% Reverse Select Pivots (top layer)
\addplot[fill=navy!50,draw=navy,pattern=crosshatch,pattern color=navy] coordinates {
    (1,1.7)
    (2,2.4)
    (3,5.0)
    (4,2.3)
};
\end{axis}
\end{tikzpicture}

\vgap
\caption{Time Breakdown of RNSM Components}
\vgap
\label{fig:rnsm-breakdown}
\end{small}
\end{figure}

\input{figures/Applicability}

\begin{figure}[!htbp]
\centering
\begin{small}

% 全局图例
\begin{tikzpicture}
    \begin{customlegend}[legend columns=4,
        legend entries={$\CM$,$\RNSM^*$, $\RNSM^+$, Rebuild},
        legend style={at={(0.5,1.15)},anchor=north,draw=none,font=\scriptsize,column sep=0.2cm}]
    \addlegendimage{line width=0.2mm,color=violate,mark=o,mark size=0.5mm}
    \addlegendimage{line width=0.2mm,color=amaranth,mark=triangle,mark size=0.5mm}
    \addlegendimage{line width=0.2mm,color=forestgreen,mark=+,mark size=0.5mm}
    %\addlegendimage{line width=0.2mm,color=navy,mark=square,mark size=0.5mm}
    \addlegendimage{line width=0.2mm,color=amber,mark=star,mark size=0.5mm}
    \end{customlegend}
\end{tikzpicture}
\vspace*{-2mm}
%\\[-\lineskip]

        \subfloat[DEEP100M-20parts]{\vspace{-2mm}
\begin{tikzpicture}[scale=1]
\begin{axis}[
height=3.0cm,
width=\columnwidth/2.10,
xlabel=recall@10, xlabel style={yshift=+8pt},  
ylabel=Query Per Second, ylabel style={yshift=-15pt}, 
label style={font=\scriptsize},
tick label style={font=\scriptsize},
ymajorgrids=true,
xmajorgrids=true,
grid style=dashed,
]
    
    % NM
\addplot[line width=0.2mm,color=violate,mark=o,mark size=0.5mm]
plot coordinates {
(0.5863,11737.3) 
(0.7227,7057.59) 
(0.7915,5904.11) 
(0.838,4729.88) 
(0.8682,3903.07) 
(0.8883,3589.66) 
(0.9044,3071.1) 
(0.9174,2869.7) 
(0.9257,2537.9) 
(0.9347,2414.84) 
(0.9421,2159.64) 
(0.946,2014.81) 
(0.9512,1885.2) 
(0.9542,1775.57) 
(0.9586,1678.34) 
};
% RNSM
\addplot[line width=0.2mm,color=amaranth,mark=triangle,mark size=0.5mm]
plot coordinates {
(0.5748,12216.2) 
(0.7158,8058.33) 
(0.7845,6148.59) 
(0.8248,4503.89) 
(0.8578,4258.45) 
(0.8818,3718.29) 
(0.8989,3160.84) 
(0.9116,2979.93) 
(0.9205,2713.43) 
(0.928,2496.71) 
(0.9343,2247.78) 
(0.9387,2143.57) 
(0.9432,2010.2) 
(0.9475,1845.67) 
(0.9497,1778.92) 
};
% RNSM-cluster
\addplot[line width=0.2mm,color=forestgreen,mark=+,mark size=0.5mm]
plot coordinates {
(0.5788,12748.4) 
(0.7189,7577.16) 
(0.7874,6392.21) 
(0.8336,4808.25) 
(0.8647,4364.6) 
(0.888,3816.25) 
(0.9024,3401.15) 
(0.9168,2942.28) 
(0.9255,2795.62) 
(0.9326,2569.33) 
(0.9384,2286.95) 
(0.945,2132.63) 
(0.9492,2045.33) 
(0.9531,1878.7) 
(0.9568,1814.18) 
};
% Rebuild
\addplot[line width=0.2mm,color=amber,mark=star,mark size=0.5mm]
plot coordinates {
(0.6383, 11108.7)
(0.7701, 7537.43)
(0.8295, 5747.74)
(0.8728, 4452.37)
(0.8982, 3881.3)
(0.9175, 3361.39)
(0.931, 2950.15)
(0.9416, 2629.5)
(0.9499, 2381.61)
(0.9559, 2169.43)
(0.9606, 1992.85)
(0.9647, 1846.55)
(0.967, 1722.58)
(0.9697, 1584.17)
(0.9722, 1512)
};
\end{axis}
\end{tikzpicture}}\hspace{2mm}
    \subfloat[DEEP100M-30parts]{\vspace{-2mm}
\begin{tikzpicture}[scale=1]
\begin{axis}[
height=3.0cm,
width=\columnwidth/2.10,
xlabel=recall@10, xlabel style={yshift=+8pt},  
ylabel=Query Per Second, ylabel style={yshift=-15pt}, 
label style={font=\scriptsize},
tick label style={font=\scriptsize},
ymajorgrids=true,
xmajorgrids=true,
grid style=dashed,
]
    
    % NM
\addplot[line width=0.2mm,color=violate,mark=o,mark size=0.5mm]
plot coordinates {
(0.5897,10974.6) 
(0.7243,7346.95) 
(0.7906,5616.24) 
(0.8374,4221.09) 
(0.8663,3920.75) 
(0.8914,3130.42) 
(0.9064,3023.31) 
(0.9179,2655.68) 
(0.9278,2492.64) 
(0.9361,2230.84) 
(0.9425,2029.55) 
(0.9475,1927.85) 
(0.9509,1811.77) 
(0.9559,1694.34) 
(0.9589,1602.92) 
};
% RNSM
\addplot[line width=0.2mm,color=amaranth,mark=triangle,mark size=0.5mm]
plot coordinates {
(0.5595,11168.2) 
(0.7033,7299.65) 
(0.7742,5567.9) 
(0.8194,4330.45) 
(0.8508,3978.99) 
(0.8742,3423.42) 
(0.8911,3077.4) 
(0.9055,2766.57) 
(0.9165,2563.38) 
(0.9243,2323.27) 
(0.933,2176.66) 
(0.9387,1940.76) 
(0.943,1734.45) 
(0.9477,1727.46) 
(0.9514,1633.25) 
};
% RNSM-cluster
\addplot[line width=0.2mm,color=forestgreen,mark=+,mark size=0.5mm]
plot coordinates {
(0.5702,12812.9) 
(0.7074,8436.75) 
(0.7807,6435.82) 
(0.8234,4730.33) 
(0.8535,4198.0) 
(0.878,3842.32) 
(0.8978,3459.39) 
(0.9103,3088.38) 
(0.9207,2821.26) 
(0.9291,2511.35) 
(0.9382,2395.09) 
(0.9439,2222.73) 
(0.9484,2025.7) 
(0.9524,1958.93) 
(0.9561,1844.13) 
};
% Rebuild
\addplot[line width=0.2mm,color=amber,mark=star,mark size=0.5mm]
plot coordinates {
(0.6383, 11108.7)
(0.7701, 7537.43)
(0.8295, 5747.74)
(0.8728, 4452.37)
(0.8982, 3881.3)
(0.9175, 3361.39)
(0.931, 2950.15)
(0.9416, 2629.5)
(0.9499, 2381.61)
(0.9559, 2169.43)
(0.9606, 1992.85)
(0.9647, 1846.55)
(0.967, 1722.58)
(0.9697, 1584.17)
(0.9722, 1512)
};
\end{axis}
\end{tikzpicture}}\hspace{2mm}

    \subfloat[DEEP100M-40parts]{\vspace{-2mm}
\begin{tikzpicture}[scale=1]
\begin{axis}[
height=3.0cm,
width=\columnwidth/2.10,
xlabel=recall@10, xlabel style={yshift=+8pt},  
ylabel=Query Per Second, ylabel style={yshift=-15pt}, 
label style={font=\scriptsize},
tick label style={font=\scriptsize},
ymajorgrids=true,
xmajorgrids=true,
grid style=dashed,
]
    
    % NM
\addplot[line width=0.2mm,color=violate,mark=o,mark size=0.5mm]
plot coordinates {
(0.6002,12021.0) 
(0.7277,7861.97) 
(0.7968,5937.97) 
(0.8398,4236.16) 
(0.8714,3142.2) 
(0.8916,3133.77) 
(0.9083,3145.64) 
(0.9206,2715.09) 
(0.9279,2494.84) 
(0.9366,2301.64) 
(0.9429,2143.31) 
(0.9482,2002.06) 
(0.9522,1873.96) 
(0.9569,1747.4) 
(0.9596,1640.15) 
};
% RNSM
\addplot[line width=0.2mm,color=amaranth,mark=triangle,mark size=0.5mm]
plot coordinates {
(0.5308,12794.1) 
(0.6815,8350.49) 
(0.7598,6035.52) 
(0.8087,4782.07) 
(0.8402,4124.86) 
(0.8658,3794.95) 
(0.8852,3223.93) 
(0.9013,2908.38) 
(0.9138,2600.79) 
(0.9214,2445.08) 
(0.9296,2317.05) 
(0.9356,2099.6) 
(0.94,1966.34) 
(0.9452,1854.14) 
(0.9491,1753.18) 
};
% RNSM-cluster
\addplot[line width=0.2mm,color=forestgreen,mark=+,mark size=0.5mm]
plot coordinates {
(0.5613,13114.0) 
(0.6982,8557.16) 
(0.7683,6457.97) 
(0.8116,5226.71) 
(0.8417,3230.47) 
(0.8682,3668.79) 
(0.8834,3456.56) 
(0.8986,3092.0) 
(0.91,2830.14) 
(0.9187,2600.83) 
(0.9263,2399.89) 
(0.933,2232.33) 
(0.9384,2086.03) 
(0.9417,1957.24) 
(0.946,1846.67) 
};
% Rebuild
\addplot[line width=0.2mm,color=amber,mark=star,mark size=0.5mm]
plot coordinates {
(0.6383, 11108.7)
(0.7701, 7537.43)
(0.8295, 5747.74)
(0.8728, 4452.37)
(0.8982, 3881.3)
(0.9175, 3361.39)
(0.931, 2950.15)
(0.9416, 2629.5)
(0.9499, 2381.61)
(0.9559, 2169.43)
(0.9606, 1992.85)
(0.9647, 1846.55)
(0.967, 1722.58)
(0.9697, 1584.17)
(0.9722, 1512)
};
\end{axis}
\end{tikzpicture}}\hspace{2mm}
    \subfloat[DEEP100M-50parts]{\vspace{-2mm}
\begin{tikzpicture}[scale=1]
\begin{axis}[
height=3.0cm,
width=\columnwidth/2.10,
xlabel=recall@10, xlabel style={yshift=+8pt},  
ylabel=Query Per Second, ylabel style={yshift=-15pt}, 
label style={font=\scriptsize},
tick label style={font=\scriptsize},
ymajorgrids=true,
xmajorgrids=true,
grid style=dashed,
]
    
    % NM
\addplot[line width=0.2mm,color=violate,mark=o,mark size=0.5mm]
plot coordinates {
(0.5822,12046.2) 
(0.7231,7835.79) 
(0.7914,5772.46) 
(0.837,4246.04) 
(0.8691,3673.79) 
(0.8902,3528.06) 
(0.9084,3139.72) 
(0.9199,2826.49) 
(0.9291,2563.03) 
(0.9369,2352.32) 
(0.9433,2126.24) 
(0.9488,1948.3) 
(0.9534,1882.41) 
(0.9579,1726.35) 
(0.9613,1662.59) 
};
% RNSM
\addplot[line width=0.2mm,color=amaranth,mark=triangle,mark size=0.5mm]
plot coordinates {
(0.5224,12850.6) 
(0.6774,8366.37) 
(0.753,6343.82) 
(0.8007,5192.51) 
(0.8329,4398.65) 
(0.8586,3764.22) 
(0.8778,3397.05) 
(0.8955,3055.34) 
(0.9064,2773.4) 
(0.9145,2546.64) 
(0.923,2350.76) 
(0.93,2183.13) 
(0.9358,2033.08) 
(0.9406,1912.05) 
(0.9446,1801.35) 
};
% RNSM-cluster
\addplot[line width=0.2mm,color=forestgreen,mark=+,mark size=0.5mm]
plot coordinates {
(0.5657,12808.6) 
(0.6979,8377.82) 
(0.772,6388.05) 
(0.8178,5232.1) 
(0.8535,4439.19) 
(0.8746,3883.33) 
(0.8927,3447.53) 
(0.9065,3110.67) 
(0.9165,2834.11) 
(0.9259,2598.2) 
(0.9325,2403.66) 
(0.9382,2207.15) 
(0.9441,2088.33) 
(0.9475,1960.54) 
(0.9511,1850.19) 
};
% Rebuild
\addplot[line width=0.2mm,color=amber,mark=star,mark size=0.5mm]
plot coordinates {
(0.6383, 11108.7)
(0.7701, 7537.43)
(0.8295, 5747.74)
(0.8728, 4452.37)
(0.8982, 3881.3)
(0.9175, 3361.39)
(0.931, 2950.15)
(0.9416, 2629.5)
(0.9499, 2381.61)
(0.9559, 2169.43)
(0.9606, 1992.85)
(0.9647, 1846.55)
(0.967, 1722.58)
(0.9697, 1584.17)
(0.9722, 1512)
};
\end{axis}
\end{tikzpicture}}\hspace{2mm}

\vgap
\vspace*{-2mm}
\caption{Recall@10-QPS Tradeoff for More Data Partitions}
\vgap
\label{fig:scalability-recall-qps-appendix}
\end{small}
\end{figure}

\stitle{Exp-b: Applicability Study of Search Performance.}
This section evaluates the applicability of RNSM across different graph index structures beyond HNSW. We compare RNSM against the Cross-query Merge (CM) baseline on three representative proximity graph types: NSG~\cite{NSG-fu-VLDB-2019}, SSG~\cite{SSG-fu-IEEE-2022}, and $\tau$-MNG~\cite{Efficient-peng-SIGMOD-2023}.

\begin{figure}[!htbp]
\centering
\begin{small}

\begin{tikzpicture}
    \begin{customlegend}[legend columns=4,
        legend entries={$\CM^*$~(random), $\CM^+$~(cluster), $\FGIM$, Rebuild},
        legend style={at={(0.5,1.15)},anchor=north,draw=none,font=\scriptsize,column sep=0.2cm},legend image post style={scale=0.8}]
    \addlegendimage{area legend,fill=amaranth!50,draw=amaranth,pattern=north west lines,pattern color=amaranth}
    \addlegendimage{area legend,fill=amber!50,draw=amber,pattern=north east lines,pattern color=amber}
    \addlegendimage{area legend,fill=navy!50,draw=navy,pattern=grid,pattern color=navy}
    \addlegendimage{area legend,fill=orange!50,draw=orange,pattern=crosshatch,pattern color=orange}
    \end{customlegend}
\end{tikzpicture}
\\[-\lineskip]

\subfloat[MARCO10M]{\vspace{-2mm}
\begin{tikzpicture}[scale=1]
\begin{axis}[
    height=3.0cm,
    width=\columnwidth/2.10,
    ybar=0.5pt, ymode=log, log basis y={2}, ymin=1,
    bar width=3.0pt,
    ylabel=Speed-up Ratio,ylabel style={yshift=-20pt}, 
    ytick={0,1,2,4},
    enlarge x limits=0.20,
    xtick={0,1,2,3},
    xticklabels={4 parts, 6 parts, 8 parts, 10 parts},
    xticklabel style={rotate=45, anchor=east, font=\scriptsize},
    legend style={font=\scriptsize,draw=none},
    label style={font=\scriptsize},
    tick label style={font=\scriptsize},
    ymajorgrids=true,
    grid style=dashed,
]
% NMr
\addplot[ybar,fill=amaranth!50,draw=amaranth,pattern=north west lines,pattern color=amaranth] coordinates {
    (0,1.63) (1,1.74) (2,2.31) (3,2.06)
};
% NMc
\addplot[ybar,fill=amber!50,draw=amber,pattern=north east lines,pattern color=amber] coordinates {
   (0,1.75) (1,1.90) (2,1.85) (3,1.62)
};
% FGIM
\addplot[ybar,fill=navy!50,draw=navy,pattern=grid,pattern color=navy] coordinates {
    (0,1.38) (1,1.03) (2,1.26) (3,1.09)
};
% BuildAsOne
\addplot[ybar,fill=orange!50,draw=orange,pattern=crosshatch,pattern color=orange] coordinates {
    (0,4.26) (1,2.59) (2,2.96) (3,1.91)
};
\end{axis}
\end{tikzpicture}}\hspace{2mm}
\subfloat[IMAGE10M]{\vspace{-2mm}
\begin{tikzpicture}[scale=1]
\begin{axis}[
    height=3.0cm,
    width=\columnwidth/2.10,
    ybar=0.5pt, ymode=log, log basis y={2}, ymin=1,
    bar width=3.0pt,
    ylabel=Speed-up Ratio,ylabel style={yshift=-20pt}, 
    ytick={0,1,2,4},
    enlarge x limits=0.20,
    xtick={0,1,2,3},
    xticklabels={4 parts, 6 parts, 8 parts, 10 parts},
    xticklabel style={rotate=45, anchor=east, font=\scriptsize},
    legend style={font=\scriptsize,draw=none},
    label style={font=\scriptsize},
    tick label style={font=\scriptsize},
    ymajorgrids=true,
    grid style=dashed,
]
% NMr
\addplot[ybar,fill=amaranth!50,draw=amaranth,pattern=north west lines,pattern color=amaranth] coordinates {
    (0,2.66) (1,2.18) (2,2.33) (3,2.11)
};
% NMc
\addplot[ybar,fill=amber!50,draw=amber,pattern=north east lines,pattern color=amber] coordinates {
   (0,1.56) (1,1.46) (2,1.98) (3,1.68)
};
% FGIM
\addplot[ybar,fill=navy!50,draw=navy,pattern=grid,pattern color=navy] coordinates {
    (0,3.26) (1,2.04) (2,1.97) (3,1.84)
};
% BuildAsOne
\addplot[ybar,fill=orange!50,draw=orange,pattern=crosshatch,pattern color=orange] coordinates {
    (0,6.43) (1,3.55) (2,2.42) (3,1.86)
};
\end{axis}
\end{tikzpicture}}\hspace{2mm}
\\
\subfloat[DEEP10M]{\vspace{-2mm}
\begin{tikzpicture}[scale=1]
\begin{axis}[
    height=3.0cm,
    width=\columnwidth/2.10,
    ybar=0.5pt, ymode=log, log basis y={2}, ymin=1,
    bar width=3.0pt,
    ylabel=Speed-up Ratio,ylabel style={yshift=-20pt}, 
    ytick={0,1,2,4},
    enlarge x limits=0.20,
    xtick={0,1,2,3},
    xticklabels={4 parts, 6 parts, 8 parts, 10 parts},
    xticklabel style={rotate=45, anchor=east, font=\scriptsize},
    legend style={font=\scriptsize,draw=none},
    label style={font=\scriptsize},
    tick label style={font=\scriptsize},
    ymajorgrids=true,
    grid style=dashed,
]
% NMr
\addplot[ybar,fill=amaranth!50,draw=amaranth,pattern=north west lines,pattern color=amaranth] coordinates {
    (0,1.79) (1,1.42) (2,1.59) (3,1.28)
};
% NMc
\addplot[ybar,fill=amber!50,draw=amber,pattern=north east lines,pattern color=amber] coordinates {
   (0,1.91) (1,1.49) (2,1.43) (3,1.30)
};
% FGIM
\addplot[ybar,fill=navy!50,draw=navy,pattern=grid,pattern color=navy] coordinates {
    (0,5.48) (1,3.44) (2,4.67) (3,3.70)
};
% BuildAsOne
\addplot[ybar,fill=orange!50,draw=orange,pattern=crosshatch,pattern color=orange] coordinates {
    (0,5.90) (1,2.36) (2,2.32) (3,1.56)
};
\end{axis}
\end{tikzpicture}}\hspace{2mm}
\subfloat[ANTON10M]{\vspace{-2mm}
\begin{tikzpicture}[scale=1]
\begin{axis}[
    height=3.0cm,
    width=\columnwidth/2.10,
    ybar=0.5pt, ymode=log, log basis y={2}, ymin=1,
    bar width=3.0pt,
    ylabel=Speed-up Ratio,ylabel style={yshift=-20pt}, 
    ytick={0,1,2,4},
    enlarge x limits=0.20,
    xtick={0,1,2,3},
    xticklabels={4 parts, 6 parts, 8 parts, 10 parts},
    xticklabel style={rotate=45, anchor=east, font=\scriptsize},
    legend style={font=\scriptsize,draw=none},
    label style={font=\scriptsize},
    tick label style={font=\scriptsize},
    ymajorgrids=true,
    grid style=dashed,
]
% NMr
\addplot[ybar,fill=amaranth!50,draw=amaranth,pattern=north west lines,pattern color=amaranth] coordinates {
    (0,1.56) (1,1.73) (2,2.44) (3,2.21)
};
% NMc
\addplot[ybar,fill=amber!50,draw=amber,pattern=north east lines,pattern color=amber] coordinates {
   (0,2.33) (1,1.34) (2,1.61) (3,1.39)
};
% FGIM
\addplot[ybar,fill=navy!50,draw=navy,pattern=grid,pattern color=navy] coordinates {
    (0,1.34) (1,1.16) (2,1.08) (3,1.11)
};
% BuildAsOne
\addplot[ybar,fill=orange!50,draw=orange,pattern=crosshatch,pattern color=orange] coordinates {
    (0,4.05) (1,3.12) (2,2.53) (3,1.97)
};
\end{axis}
\end{tikzpicture}}\hspace{2mm}

\vgap
\caption{Speed-up Ratio Comparison across Different Partitions}
\vgap
\label{fig:speedup-comparison-FULL}
\end{small}
\end{figure}
\begin{figure}[!htbp]
\centering
\begin{small}

\begin{tikzpicture}
    \begin{customlegend}[legend columns=2,
        legend entries={CM, FGIM},
        legend style={at={(0.5,1.15)},anchor=north,draw=none,font=\scriptsize,column sep=0.20cm},
        legend image post style={scale=0.8}]
    \addlegendimage{area legend,fill=amber!50,draw=amber,pattern=north east lines,pattern color=amber}
    \addlegendimage{area legend,fill=navy!50,draw=navy,pattern=grid,pattern color=navy}
    \end{customlegend}
\end{tikzpicture}
\\[-\lineskip]

\begin{tikzpicture}[scale=1]
\begin{axis}[
    height=\columnwidth/2.6,
    width=\columnwidth/1.9,
    ybar=0.5pt,
    ymin=0,
    ymax=10.5,
    bar width=6pt,
    enlarge x limits=0.18,
    xmin=-0.5,
    xmax=3.5,
    ylabel=Speed-up Ratio,
    ylabel style={yshift=-20pt},
    xlabel style={yshift=+2pt},
    xtick={0,1,2,3},
    xticklabels={DEEP, ANTON, IMAGE, MARCO},
    x tick label style={rotate=45, anchor=east, font=\scriptsize},
    legend style={font=\scriptsize,draw=none},
    label style={font=\scriptsize},
    tick label style={font=\scriptsize},
    ymajorgrids=true,
    grid style=dashed,
    axis line style={line width=0.2mm},
]
% CM: cumulative NGM merge time divided by cumulative RGTM merge time.
\addplot[ybar,fill=amber!50,draw=amber,pattern=north east lines,pattern color=amber]
coordinates {
    (0,1.36) (1,1.55) (2,1.43) (3,1.57)
};
% FGIM: cumulative FGIM merge time divided by cumulative RGTM merge time.
\addplot[ybar,fill=navy!50,draw=navy,pattern=grid,pattern color=navy]
coordinates {
    (0,9.79) (1,6.68) (2,5.57) (3,6.44)
};
\end{axis}
\end{tikzpicture}

\vgap
\vspace*{-3mm}
\caption{Incremental Merge Speed-up Ratio on 10M Datasets.}
\vgap
\label{fig:incremental-10m-merge-speedup}
\end{small}
\end{figure}

Fig~\ref{fig:applicability} presents the complete results. Across all configurations, RNSM consistently maintains search performance comparable to or better than CM. This demonstrates that RNSM preserves the navigability and search efficiency of the underlying graph indexes, regardless of their specific construction algorithms or edge selection policies.

\stitle{Exp-c: Comprehensive Results for Exp-3.}
This section presents the comprehensive experimental results for Exp-3 (Multiple Index Merging). It covers all evaluated datasets (DEEP10M, MARCO10M, IMAGE10M, and ANTON10M) using both random and cluster partitioning strategies across varying partition counts $\{4, 6, 8, 10\}$.

% \noindent\textbf{Merge Efficiency.}
Fig~\ref{fig:speedup-comparison-FULL} reports the complete speedup ratios across all datasets and partition numbers. Our method shows the best merge efficiency in all cases, achieving up to 2.33$\times$, 5.48$\times$, and 6.43$\times$ acceleration compared to CM, FGIM, and Rebuild, respectively.

% \noindent\textbf{Search Performance.}
Fig~\ref{fig:multi-merge-recall-qps-appendix-full} presents the complete recall@10-QPS tradeoff curves for all experimental configurations. The results demonstrate that our RNSM and MOS framework maintains the expected search performance of graph indexes and outperforms the existing merge method, FGIM.

\input{figures/multi-index-merge-QPS-appendix-FULL}

\stitle{Exp-d: Search Performance of Scalability Study for Exp-3.}
This experiment complements the DEEP100M scalability study in Exp-3: Test of Multiple Indexes Merging. We partition DEEP100M into $\{20,30,40,50\}$ sub-indexes and evaluate $\RNSM^*$ and $\RNSM^+$ on random and cluster partitions, using $\CM$ as the baseline. Fig~\ref{fig:scalability-recall-qps-appendix} reports the recall@10-QPS tradeoff curves for the scalability experiments on DEEP100M and the merge of \{20,30,40,50\} indexes. Results show that our RNSM+MOS method achieves comparable search performance to the $\CM$ and Rebuild baselines across all settings.

\stitle{Exp-e: Incremental Merge Efficiency for Exp-5.}
This experiment complements the search-performance results of Exp-5 with the corresponding merge-efficiency measurements. For each dataset, we sum the merge time over the nine incremental steps from 1M to 10M vectors and report the cumulative merge time of each baseline divided by that of $\RNSM$; a ratio greater than one therefore indicates a speedup by $\RNSM$.

As shown in Fig~\ref{fig:incremental-10m-merge-speedup}, $\RNSM$ consistently reduces the cumulative merge time across all four datasets. Its speedup ranges from $1.36\times$ to $1.57\times$ over $\CM$ and from $5.57\times$ to $9.79\times$ over $\FGIM$. On MARCO10M, $\RNSM$ achieves the $1.57\times$ and $6.44\times$ speedups over $\CM$ and $\FGIM$, respectively, reported in Exp-5. Together with the recall@10--QPS results in Fig~\ref{fig:incremental-10m-recall-qps}, these measurements show that $\RNSM$ improves incremental merge efficiency while maintaining comparable search performance.

\end{document}